\documentclass[nospthms, 10pt, final]{svjour3_arXiv}%

\usepackage[a4paper, headheight = 2.0cm, tmargin=1.5cm, bmargin = 1.5cm, hmargin={1.9cm,1.9cm} ]{geometry}

\smartqed
\usepackage[sort&compress, square, authoryear]{natbib}
\pdfmapfile{+zi4.map}
\usepackage{zi4}

\usepackage{booktabs} 
\usepackage{tablefootnote}
\usepackage[ruled]{algorithm2e} 

\usepackage{enumitem}
\usepackage{amsmath}
\usepackage{amsthm}
\usepackage{amsfonts}
\usepackage{amssymb}
\usepackage{graphicx}
\usepackage{lineno}
\usepackage[unicode]{hyperref} 
\usepackage{color}%
\usepackage[usenames,dvipsnames,table]{xcolor}

\newtheorem{theorem}{Theorem}[section]
\newtheorem{corollary}{Corollary}[section]
\newtheorem{definition}{Definition}[section]
\newtheorem{example}{Example}[section]
\newtheorem{remark}{Remark}[section]
\newtheorem{lemma}{Lemma}[section]

\providecommand{\U}[1]{\protect\rule{.1in}{.1in}}

\usepackage{xcolor}

\newcommand{\Black}[1]{{\color[rgb]{0.0, 0.0, 0.0}{#1}}}
\newcommand{\Red}[1]{{\color[rgb]{1.0, 0.0, 0.0}{#1}}}

\newcommand{\CRed}[1]{{\color[rgb]{0.5, 0.0, 0.0}{#1}}}
\newcommand{\Blue}[1]{{\color[rgb]{0.1, 0.5, 0.9}{#1}}}
\newcommand{\CBlue}[1]{{\color[rgb]{0.0, 0.0, 0.7}{#1}}}

\newcommand{\CPurple}[1]{{\color[rgb]{1.0, 0.0, 0.6}{#1}}}

\newcommand{\CDRPurple}[1]{{\color[rgb]{0.5921568627450980392156862745098, 0.27058823529411764705882352941176, 0.47058823529411764705882352941176}{#1}}}

\newcommand{\mref}[2]{#1{\relsize{-1}/#2}}

\usepackage{caption, listings}

\DeclareCaptionFont{white}{\color{white}}
\DeclareCaptionFormat{listingcaptionformat}{\colorbox[rgb]{0.15, 0.38, 0.0}{\parbox{0.985\textwidth}{\hspace{0pt}#1#2#3}}}
\DeclareCaptionLabelFormat{listingcaptionlabelformat}{#1 #2}
\makeatletter
\lst@Key{countblanklines}{true}[t]%
{\lstKV@SetIf{#1}\lst@ifcountblanklines}

\lst@AddToHook{OnEmptyLine}{%
	\lst@ifnumberblanklines\else%
	\lst@ifcountblanklines\else%
	\advance\c@lstnumber-\@ne\relax%
	\fi%
	\fi}
\makeatother

\usepackage{wallpaper}
\journalname{}

\hypersetup{
    colorlinks = true,
    citecolor = blue,
    filecolor=black,
    linkcolor = red,
    anchorcolor = red,
    urlcolor= red
}

\begin{document}

	\LRCornerWallPaper{0.225}{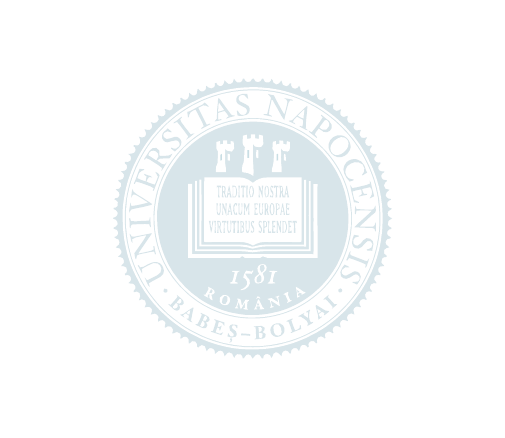}
	
\setcounter{tocdepth}{3}
\allowdisplaybreaks
\renewcommand{\thelstlisting}{\thesection.\arabic{lstlisting}}

\title{An OpenGL and C++ based function library for curve and surface modeling in a large class of extended Chebyshev spaces}

\titlerunning{Curve and surface modeling in a large class of extended Chebyshev spaces}

\author{\'{A}goston R\'{o}th\thanks{This research was supported by the European Union and the State of Hungary, co-financed by the European Social Fund in the framework of 	T\'{A}MOP-4.2.4.A/2-11/1-2012-0001 'National Excellence Program'.}}

\institute{
\'A. R\'oth \at Department of Mathematics and Computer Science of the Hungarian Line, Babe\c{s}--Bolyai University, RO--400084 Cluj-Napoca, Romania \\
Tel.: +40-264-405300\\
Fax:  +40-264-591906\\
\email{agoston\_roth@yahoo.com}
}

\date{Submitted to arXiv on October 14, 2018}

\maketitle

\begin{abstract}
We propose a platform-independent multi-threaded function library that provides data structures to generate, differentiate and render both the ordinary basis and the normalized B-basis of a user-specified extended Chebyshev (EC) space that comprises the constants and can be identified with the solution space of a constant-coefficient homogeneous linear differential equation defined on a sufficiently small interval. Using the obtained normalized B-bases, our library can also generate, (partially) differentiate, modify and visualize a large family of so-called B-curves and tensor product B-surfaces. Moreover, the library also implements methods that can be used to perform dimension elevation, to subdivide B-curves and B-surfaces by means of de Casteljau-like B-algorithms, and to generate basis transformations for the B-representation of arbitrary integral curves and surfaces that  are described in traditional parametric form by means of the ordinary bases of the underlying EC spaces. Independently of the algebraic, exponential, trigonometric or mixed type of the applied EC space, the proposed library is numerically stable and efficient up to a reasonable dimension number and may be useful for academics and engineers in the fields of Approximation Theory, Computer Aided Geometric Design, Computer Graphics, Isogeometric and Numerical Analysis.

\end{abstract}

\section{Introduction}\label{sec:introduction}

The following subsections detail our motivations, main objectives and the structure of the manuscript.

\subsection{Motivations}\label{subsec:motivations}
Normalized B-bases (a comprehensive study of which can be found in
\citep{Pena1999} and references therein) are normalized totally positive bases
that imply optimal shape preserving properties for the representation of
curves described as convex combinations of control points and basis functions. Considering a non-empty compact definition domain $\left[\alpha,\beta\right]\subset\mathbb{R}$, the most well-known representative of such bases are the classical Bernstein polynomials of degree $n\in\mathbb{N}$, cf.
\citep{Carnicer1993}. Similarly to them, normalized B-bases provide shape preserving properties
like closure for the affine transformations of the control polygon, convex
hull, variation diminishing (which also implies convexity preserving of plane
control polygons), endpoint interpolation, monotonicity preserving, hodograph
and length diminishing, and a recursive corner cutting algorithm (also called
B-algorithm) that is the analogue of the de Casteljau algorithm of classical B\'{e}zier
curves. Among all normalized totally positive bases of a given vector space of
functions a normalized B-basis is the least variation diminishing and the
shape of the generated curve more mimics its control polygon. Important curve
design algorithms like evaluation, subdivision, degree (or dimension) elevation or knot
insertion are in fact corner cutting algorithms that can be treated in a
unified way by means of B-algorithms induced by B-bases.

Curve and surface modeling tools based on non-polynomial normalized B-bases also ensure further advantages like: possible shape or design parameters; singularity free exact parametrization (e.g.\ parametrization of
conic sections may correspond to natural arc-length parametrization); higher or even infinite order of precision concerning (partial) derivatives; ordinary (i.e., traditionally parametrized) integral curves and surfaces can be described exactly by means of control points without any additional weights (the calculation of which, apart from some simple cases, is cumbersome for the designer); important transcendental curves and surfaces which are of interest in real-life applications can also be represented exactly (the standard rational B\'{e}zier or NURBS models cannot encompass these
geometric objects). Moreover, concerning condition numbers and stability, a normalized B-basis is the unique normalized totally positive basis that is optimally stable among all non-negative bases of a given vector space of functions, cf. \citep[Corollary 3.4, p. 89]{Pena1999}.

Apart from their interest in the classical contexts of Computer Aided Geometric Design, Numerical Analysis and Approximation Theory, normalized B-bases and their spline counterparts have also been used in Isogeometric Analysis recently (consider, e.g.,  \citep{ManniPelosiSampoli2011} and references therein). Compared with classical finite element methods, Isogeometric Analysis provides several advantages when one describes the geometry by generalized B-splines and invokes an isoparametric approach in order to approximate the unknown solutions of differential equations (e.g.,  of Poisson type problems) or Dirichlet boundary conditions by the same type of functions.

These advantageous properties make normalized B-bases ideal blending function system
candidates for curve and surface modeling.

\subsection{Preliminaries and objectives}
In order to be able to formulate the main objectives of the manuscript, we will use the following well-known notions. Let $n\geq1$ be a fixed integer and consider the extended Chebyshev (EC) system%
\begin{equation}
    \mathcal{F}_{n}^{\alpha,\beta}=\left\{  \varphi
    _{n,i}\left(  u\right)  :u\in\left[  \alpha,\beta\right]  \right\}  _{i=0}%
    ^{n},~\varphi_{n,0}	\equiv 1,~-\infty<\alpha<\beta<\infty\label{eq:ordinary_basis}%
\end{equation}
of basis functions in $C^{n}\left(  \left[  \alpha,\beta\right]  \right)  $,
i.e., by definition \citep{KarlinStudden1966}, for any integer $0\leq r\leq n$, any strictly increasing
sequence of knot values $\alpha\leq u_{0}<u_{1}<\ldots<u_{r}\leq\beta$, any
positive integers (or multiplicities) $\left\{  m_{k}\right\}  _{k=0}^{r}$
such that $\sum_{k=0}^{r}m_{k}=n+1$, and any real numbers $\left\{
\xi_{k,\ell}\right\}  _{k=0,~\ell=0}^{r,~m_{k}-1}$ there always exists a
unique function
\begin{equation}
    f:=\sum_{i=0}^{n}\lambda_{n,i}\varphi_{n,i}\in\mathbb{S}_{n}^{\alpha,\beta}:=\left\langle\mathcal{F}_{n}^{\alpha,\beta
    }\right\rangle:=\operatorname{span}\mathcal{F}_{n}^{\alpha,\beta},~\lambda_{n,i}\in\mathbb{R},~i=0,1,\ldots,n \label{eq:unique_solution}%
\end{equation}
that satisfies the conditions of the Hermite interpolation problem%
\begin{equation}
    f^{\left(  \ell\right)  }\left(  u_{k}\right)  =\xi_{k,\ell},~\ell
    =0,1,\ldots,m_{k}-1,~k=0,1,\ldots,r. \label{eq:Hermite_interpolation_problem}%
\end{equation}

In what follows, we assume that the sign-regular determinant of the coefficient
matrix of the linear system (\ref{eq:Hermite_interpolation_problem}) of
equations is strictly positive for any permissible parameter settings introduced above.
Under these circumstances, the vector space $\mathbb{S}_{n}^{\alpha,\beta}$ of functions is
called an EC space of dimension $n+1$. In terms of zeros, this
definition means that any non-zero element of $\mathbb{S}_{n}^{\alpha,\beta}$
vanishes at most $n$ times in the interval $\left[  \alpha,\beta\right]  $. Such spaces and their corresponding spline counterparts have been widely studied, consider, e.g.,  articles \citep{Schumaker2007, Lyche1985,Pottmann1993,MazureLaurent1998,MainarPena1999,MainarPenaSanchez2001,LuWangYang2002,CarnicerMainarPena2004,MainarPena2004,CostantiniLycheManni2005,CarnicerMainarPena2007,MainarPena2010,Roth2015a,Roth2015b} and many other references therein. (Concerning the definition of EC spaces, the condition $1\equiv\varphi_{n,0}\in\mathbb{S}_{n}^{\alpha,\beta}$ would be not necessary, nevertheless we have included the constants in $\mathbb{S}_{n}^{\alpha,\beta}$ in order to ensure that all its bases can be normalized.)

Hereafter we will also refer to $\mathcal{F}_{n}^{\alpha,\beta}$ as the
ordinary basis of $\mathbb{S}_{n}^{\alpha,\beta}$ and we will also assume that the ($n$-dimensional) space
$
D\mathbb{S}_{n}^{\alpha,\beta}:=
\big\{
f^{\left(1\right)} 
: f \in \mathbb{S}_{n}^{\alpha,\beta}
\big\}
$
of the derivatives 
is also EC over the interval $\left[\alpha,\beta\right]$. Using \citep[Theorem
5.1]{CarnicerPena1995} and \citep[Theorem 4.1]{CarnicerMainarPena2004}, it follows that under these conditions the vector space $\mathbb{S}_{n}^{\alpha,\beta}$ also has a unique
strictly totally positive normalizable basis\footnote{A basis $\left\{f_i\left(u\right):\left[\alpha,\beta\right]\right\}_{i=0}^{n}$ of $\mathbb{S}_{n}^{\alpha,\beta}$ is strictly totally positive, if all minors of all its collocation matrices $\left[f_i\left(u_j\right)\right]_{i=0,\,j=0}^{n,m}$ are strictly positive, where $\left\{u_j\right\}_{j=0}^m$ are arbitrary strictly increasing knot values within $\left[\alpha,\beta\right]$.}, called normalized B-basis
\begin{equation}
    \mathcal{B}_{n}^{\alpha,\beta}=\left\{  b_{n,i}\left(  u\right)  :u\in\left[
    \alpha,\beta\right]  \right\}  _{i=0}^{n} \label{eq:B-basis}%
\end{equation}
that apart from the identity%
\begin{equation}
    \sum_{i=0}^{n}b_{n,i}\left(  u\right)  \equiv1,~\forall u\in\left[
    \alpha,\beta\right]  \label{eq:partition_of_unity}%
\end{equation}
also fulfills the properties%
\begin{align}
    b_{n,0}\left(  \alpha\right)   &  =b_{n,n}\left(  \beta\right)
    =1,\label{eq:endpoint_interpolation}\\
    b_{n,i}^{\left(  j\right)  }\left(  \alpha\right)   &  =0,~j=0,\ldots
    ,i-1,~b_{n,i}^{\left(  i\right)  }\left(  \alpha\right)
    >0,\label{eq:Hermite_conditions_0}\\
    b_{n,i}^{\left(  j\right)  }\left(  \beta\right)   &  =0,~j=0,1,\ldots
    ,n-1-i,~\left(  -1\right)  ^{n-i}b_{n,i}^{\left(  n-i\right)  }\left(
    \beta\right)  >0 \label{eq:Hermite_conditions_alpha}%
\end{align}
conform \citep[Theorem 5.1]{CarnicerPena1995} and \citep[Equation (3.6)]{Mazure1999}. (In order to avoid ambiguity, in case of some figures we will also use the notation $\overline{\mathcal{F}}_n^{\alpha,\beta}$ instead of $\mathcal{B}_n^{\alpha,\beta}$.)

All algorithms that will be presented in the forthcoming sections are valid in case of any EC space $\mathbb{S}_n^{\alpha, \beta}$ that fulfills all conditions above, however in case of their C++ and OpenGL based implementation we always assume that $\mathbb{S}_n^{\alpha, \beta}$ can be identified with the solution space of the constant-coefficient homogeneous linear differential equation
\begin{equation}
    \sum_{i=0}^{n+1}\gamma_{i}v^{\left(  i\right)  }\left(  u\right)=0,~\gamma_{i}\in\mathbb{R},~u\in\left[  \alpha,\beta\right]  
    \label{eq:differential_equation}
\end{equation} 
of order $n+1$. Such a solution space:
\begin{itemize}
    \item 
    is translation invariant and it is spanned by those ordinary basis functions that are generated by the (higher order) zeros of the characteristic polynomial
    \begin{equation}
        p_{n+1}\left(  z\right) =\sum_{i=0}^{n+1}\gamma_{i}z^{i},~z\in \mathbb{C}
        \label{eq:characteristic_polynomial}%
    \end{equation}
    associated with the differential equation (\ref{eq:differential_equation}) (naturally, in order to ensure that $\varphi_{n,0}\equiv 1 \in \mathbb{S}_n^{\alpha,\beta}$, we will assume that $z=0$ is at least a first order zero of (\ref{eq:characteristic_polynomial}));
    
    \item
    is of class $C^{\infty}\left(\left[\alpha,\beta\right]\right)$ and is EC on intervals of sufficiently small length $\beta-\alpha \in \left(0,\ell_n\right)$, where the so-called critical length $\ell_n > 0$ (i.e., the supremum of the lengths of the intervals on which the given space is EC) can be determined as follows (see \citep[Proposition 3.1]{CarnicerMainarPena2004}):
    \begin{itemize}
        \item
        let
        \begin{equation}
            W_{\left[  v_{n,0},v_{n,1},\ldots,v_{n,n}\right]  }\left(  u\right)  :=
            \left[
            v_{n,i}^{\left(j\right)}\left(u\right)
            \right]_{i=0,~j=0}^{n,~n},~u\in\left[\alpha,\beta\right]
            \label{eq:forward_Wronskian}
        \end{equation}
        be the Wronskian matrix of those particular integrals
        \begin{equation}
            v_{n,i}:=\sum_{k=0}^{n}\rho_{i,k}\varphi_{n,k}\in\mathbb{S}_{n}^{\alpha,\beta
            },~\left\{\rho_{i,k}\right\}_{k=0}^{n}\subset\mathbb{R},~i=0,1,\ldots,n\label{eq:particular_integrals}%
        \end{equation}
        of (\ref{eq:differential_equation}) which correspond to the boundary conditions%
        \begin{equation}
            \def\arraystretch{1.3}
            \left\{
            \begin{array}{rcl}
                v_{n,i}^{\left(  j\right)  }\left(  \alpha\right)    & = & 0,~j=0,\ldots,i-1,\\
                v_{n,i}^{\left(  i\right)  }\left(  \alpha\right)    & = & 1,\\
                v_{n,i}^{\left(  j\right)  }\left(  \beta\right)    & =& 0 ,~j=0,\ldots,n-1-i,
            \end{array}
            \right.
            \def\arraystretch{1.0}
            \label{eq:boundary_conditions}
        \end{equation}
        i.e., the system $\left\{  v_{n,i}\left(  u\right)  :u\in\left[
        \alpha,\beta\right]  \right\}  _{i=0}^{n}$ is a bicanonical basis on the interval $\left[\alpha, \beta\right]$ such that the Wronskian (\ref{eq:forward_Wronskian}) at $u=\alpha$ is a
        lower triangular matrix with positive (unit) diagonal entries;
        \item
        consider the functions (or Wronskian determinants)%
        \begin{align}
            w_{n,i}\left(  u\right)  &:=\det W_{\left[  v_{n,i},v_{n,i+1},\ldots
                ,v_{n,n}\right]  }\left(  u\right)  ,~i=
            \left\lfloor \tfrac{n}{2}\right\rfloor +1,\ldots
            ,n,\label{eq:Wronskian_determinants}
            \\
            \theta_{n,i}\left(u\right) &:= \left(-1\right)^{n\left(n+1-i\right)}\det W_{\left[  v_{n,i},v_{n,i+1},\ldots,v_{n,n}\right]  }\left(  -u\right),~i=
            \left\lfloor \tfrac{n}{2}\right\rfloor +1,\ldots,n,
        \end{align}
        define the critical length%
        \begin{equation}
            \ell_{n}:=\min_{i=\left\lfloor \frac{n}{2}\right\rfloor +1,\ldots
                ,n}\min\left\{  \left\vert u - \alpha\right\vert :w_{n,i}\left(  u\right)
            =0\text{ or } \theta_{n,i}\left(u\right)=0,~u\neq\alpha\right\}
            \label{eq:critical_length}
        \end{equation}
        (we write $\ell_{n}=+\infty$ whenever the Wronskian determinants
        (\ref{eq:Wronskian_determinants}) do not have real zeros other than $\alpha$, moreover $\ell_{n}$ is infinite whenever the characteristic polynomial (\ref{eq:characteristic_polynomial}) has only real roots, otherwise it is finite number that, in general, also depends on parameters resulting from the differential equation (\ref{eq:differential_equation}), i.e., on the real and imaginary parts of the complex zeros of the characteristic polynomial (\ref{eq:characteristic_polynomial})).
    \end{itemize}
\end{itemize}

We will denote the critical length of the derivative space $D\mathbb{S}_{n}^{\alpha,\beta}$ by $\ell_n^{\prime}$ and, in order to ensure that $\mathbb{S}_n^{\alpha, \beta}$ is an EC space that also possesses a unique normalized B-basis (see \citep[Theorem 4.1 and Corollary 4.1]{CarnicerMainarPena2004}), \textit{from hereon we always assume that} $\beta \in \left(\alpha, \alpha + \ell_n^{\prime}\right)$. (The interval length $\beta - \alpha \in \left(0,\ell_n^{\prime}\right)$ can be considered as a shape or tension parameter. In order to avoid ambiguity, instead of $\ell_{n}$ and $\ell_n^{\prime}$, in some cases, we will also use the notations $\ell\left( \mathbb{S}_n^{\alpha,\beta}\right)$ and $\ell^{\prime}\left( \mathbb{S}_n^{\alpha,\beta}\right):=\ell\left(D\mathbb{S}_{n}^{\alpha,\beta}\right)$, respectively.) Following \citep[p.\ 67]{CarnicerMainarPena2004}, we will also refer to $\ell_n^{\prime}$ as \textit{critical length for design}.

\begin{remark}[Concerning the endpoints of the definition domain]
    Since the underlying vector space $\mathbb{S}_n^{\alpha,\beta}$ is translation invariant, it would be sufficient to study the properties of such spaces on intervals of the form $[0,h]$, where $h \in \left(0,\ell_n^{\prime}\right)$ and, consequently, the notation $\mathbb{S}_{n}^{\alpha,\beta}$ could be simplified to $\mathbb{S}_{n}^{h}$. Nevertheless, due to flexibility and some implementation details which may appear, e.g.,  in case of the control point based exact description (i.e., B-representation) of ordinary integral curves defined on intervals of appropriate length, we have designed/implemented all proposed algorithms on the interval $\left[\alpha,\beta\right]$, where $\beta - \alpha \in \left(0,\ell_n^{\prime}\right)$. In this way, during programming, users will have more comfort.
\end{remark}  

As we will see, the proposed algorithms do not assume that the characteristic polynomial (\ref{eq:characteristic_polynomial}) is an either odd or even function, but if this the case, the underlying EC space will also be reflection invariant, i.e., if $\mathbb{S}_n^{\alpha,\beta}$ denotes the solution space of (\ref{eq:differential_equation}), $v\in\mathbb{S}_n^{\alpha,\beta}$ and $x\in\mathbb{R}$ is fixed, then $v\left(u - x\right)$ also belongs to  $\mathbb{S}_n^{\alpha,\beta}$, moreover the normalized B-basis functions (\ref{eq:B-basis}) will share the symmetry property 
\begin{equation}
    b_{n,i}\left(u\right)=b_{n,n-i}\left(\alpha + \beta - u\right),~\forall u \in \left[\alpha, \beta\right]~i=0,1,\ldots,\left\lfloor \tfrac{n}{2}\right\rfloor.
    \label{eq:symmetry}
\end{equation}
In this case, the critical length (\ref{eq:critical_length}) can be determined (see \citep[Proposition 3.2]{CarnicerMainarPena2004}) by means of the simpler formula
\begin{equation}
    \ell_{n}:=\min_{i=\left\lfloor \frac{n}{2}\right\rfloor +1,\ldots
        ,n}\min\left\{  \left\vert u - \alpha\right\vert :w_{n,i}\left(  u\right)
    =0,~u\neq\alpha\right\}.
    \label{eq:critical_length_reflection_invariant}
\end{equation}

Based on both original and existing theoretical results, the main objective of the manuscript is to propose and implement general algorithms into a robust and flexible OpenGL and C++ based multi-threaded function library that can be used: 
\begin{itemize}
    \item 
    to automatically generate and evaluate the derivatives of any order of both the ordinary basis and the normalized B-basis of a not necessarily reflection invariant EC space that comprises the constants, can be identified with the translation invariant solution space of the differential equation (\ref{eq:differential_equation}) and whose derivative space is also EC;
    
    \item
    to describe, generate, manipulate and render so-called B-curves defined as convex combinations of control points and normalized B-basis functions;
    
    \item
    to generate, manipulate and render so-called B-surfaces defined as tensor products of B-curves;
    
    \item
    to elevate the dimension of the underlying EC space(s) and consequently the order(s) of the B-curve (surface) that is rendered;
    
    \item
    to subdivide B-curves by means of general B-algorithms implied by the normalized B-basis of the given EC space and to extend this subdivision technique to B-surfaces;
    
    \item
    to generate transformations matrices that map the normalized B-bases of the applied EC spaces to their ordinary bases, in order to ensure control point configurations for the exact B-representation of large classes of integral curves and surfaces that are described in traditional parametric form by means of the ordinary bases of the used EC spaces.
\end{itemize}

During our study, we will also investigate the correctness and computational complexity of the proposed algorithms. To the best of our knowledge, such a general programming framework was not presented in the literature for curve and surface modeling with normalized B-basis functions. Naturally, in certain special cases (like in EC spaces of pure traditional, trigonometric and hyperbolic polynomials of finite degree) one may provide more efficient curve and surface modeling techniques, since one may know the explicit expressions or other useful properties of the applied normalized B-basis functions that may lead to numerically more stable and efficient algorithms related to differentiation, order elevation, subdivision and basis transformations. However, in general, one does not even know the closed form of these ideal basis functions (e.g., they may appear in integral \citep{MainarPena2010} or in determinant form \citep{Mazure1999} that are computationally difficult and expensive to evaluate either by hand, or by numerical methods). Therefore, one has to make a compromise between a robust flexible design possibility that can be universally applied in a more general context and another special modeling technique that may be more efficient but it was developed for the solution of a very special design problem.

As we will see, each of the proposed algorithms relies on the successful evaluation of zeroth and higher order (endpoint) derivatives of either of the ordinary basis functions (\ref{eq:ordinary_basis}) or of the normalized B-basis (\ref{eq:B-basis}). The order of (endpoint) derivatives that have to be evaluated increases proportionally with the dimension $n+1$ of the underlying EC space. Due to floating point arithmetical operations, the maximal dimension for which one does not bump into numerical stability problems depends on the type of the ordinary basis functions of the given EC space -- depending on the case, it may be smaller or greater, but considering that, in practice, curves and surfaces are mostly composed of smoothly joined lower order arcs and patches, a clever implementation of the proposed algorithms can be useful in case of real-life applications.

\subsection{Structure}
The rest of the manuscript is organized as follows. Section \ref{sec:proposed_algorithms} consists of four subsections that detail and study \textit{general} algorithms that can be used:
\begin{enumerate}
    \item[(1)]
    to construct and differentiate the bases (\ref{eq:ordinary_basis}) and (\ref{eq:B-basis}) in EC spaces that comprise the constants, can be identified with the solution spaces of differential equations of type (\ref{eq:differential_equation}) and whose derivative spaces are also EC; 
    \item[(2)]
    to elevate the dimensions of the underlying EC spaces and consequently the order of the induced B-curves and B-surfaces; 
    \item[(3)]
    to subdivide B-curves and B-surfaces; 
    \item[(4)]
    to generate basis transformation matrices for the control point based exact description (i.e., B-representation) of ordinary integral curves and surfaces given in traditional parametric form.
\end{enumerate}
Since class diagrams, full implementation details and usage examples can be found in the attached user manual \citep{Roth2018b}, Section \ref{sec:implementation} presents very briefly the main packages, data structures and methods of our implementation. Section \ref{sec:examples_and_statistics} provides further examples, run-time statistics and gives advices for handling possible numerical instabilities. Section \ref{sec:closure} closes the manuscript with our summary and conclusions.

\setlength{\belowcaptionskip}{-9pt}
\setlength{\abovecaptionskip}{2pt}

\section{Theoretical results and proposed algorithms}\label{sec:proposed_algorithms}

In order to formulate the input and output of our algorithms, we define the following control point  based integral curves and surfaces.

\begin{definition}[B-curves]
    The convex combination 
    \begin{equation}
        \mathbf{c}_{n}\left(  u\right)  =\sum_{i=0}^{n}\mathbf{p}_{i}b_{n,i}\left(
        u\right)  ,~ u\in\left[  \alpha,\beta\right]  ,~\mathbf{p}_{i}=\left[
        p_{i}^{\ell}\right]  _{\ell=0}^{\delta-1}\in%
        \mathbb{R}
        ^{\delta},\,\delta \geq 2 \label{eq:B_curve}%
    \end{equation}
    described by means of the normalized B-basis (\ref{eq:B-basis}) is called an B-curve of order $n$, where $\left[\mathbf{p}_i\right]_{i=0}^n$ denotes its control polygon.
\end{definition}

\begin{definition}[B-surfaces]
    Denoting by
    \[
    \mathcal{B}_{n_r}^{\alpha_r, \beta_r}
    =
    \left\{
    b_{n_r,i_r}\left(u_r;\alpha_r,\beta_r\right): u_r \in \left[\alpha_r, \beta_r\right]
    \right\}_{i_r=0}^{n_r},~r=0,1
    \]
    two normalized B-basis of some EC spaces and using the tensor product of curves of type (\ref{eq:B_curve}), one can define the B-surface 
    \begin{equation}
        \label{eq:B-surface}
        \mathbf{s}_{n_0,n_1}\left(u_0,u_1\right)
        =
        \sum_{i_0=0}^{n_0} \sum_{i_1 = 0}^{n_1}
        \mathbf{p}_{i_0,i_1}
        b_{n_0,i_0}\left(u_0;\alpha_0, \beta_0\right)
        b_{n_1,i_1}\left(u_1;\alpha_1, \beta_1\right),\:\mathbf{p}_{i_0,i_1}=\left[p_{i_0,i_1}^{\ell}\right]_{\ell=0}^2\in\mathbb{R}^3
    \end{equation}
    of order $\left(n_0,n_1\right)$, where the matrix $\left[\mathbf{p}_{i_0,i_1}\right]_{i_0=0,\,i_1=0}^{n_0,\,n_1}$ forms a control net.
\end{definition}

\subsection{Construction and differentiation of normalized B-basis functions in a large class of EC spaces}\label{subsec:construction}

As already stated in Section \ref{sec:introduction}, our implementation assumes that the underlying EC space $\mathbb{S}_{n}^{\alpha,\beta}$ corresponds to the solution space of a constant-coefficient homogeneous linear differential equation of type (\ref{eq:differential_equation}), where $\beta-\alpha \in \left(0,\ell_n^{\prime}\right)$. This is not necessary for the correctness of the algorithms that will be presented in the forthcoming sections. The only reason for this additional assumption is the fact that in this case we have the possibility to handle a large family of (mixed) EC spaces in a unified way.

For example, if $\mathbf{i}=\sqrt{-1}$, $a, b \in \mathbb{R}$ and $z=a+\mathbf{i}b$ is an $m$th order ($m\geq 1$) zero of the characteristic polynomial (\ref{eq:characteristic_polynomial}), then based on its real and imaginary parts one has that:
\begin{itemize}
    \item
    if $a,b\in\mathbb{R}\setminus\left\{0\right\}$, the conjugate complex number $\overline{z}$ is also a root of multiplicity $m$ and consequently one obtains an algebraic-exponential-trigonometric (AET) mixed EC subspace
    \begin{equation}
        \label{eq:AET}
        \left\langle
        \left\{
        u^r e^{au} \cos\left(bu\right),
        u^r e^{au} \sin\left(bu\right):
        u \in \left[\alpha,\beta\right]
        \right\}_{r=0}^{m-1}
        \right\rangle \subseteq \mathbb{S}_n^{\alpha,\beta};
    \end{equation}
    \item
    if $a \neq 0$, but $b=0$, then one has the
    algebraic-exponential (AE) mixed EC subspace
    \begin{equation}
        \label{eq:AE}
        \left\langle
        \left\{
        u^r e^{au} : u \in \left[\alpha,\beta\right]
        \right\}_{r=0}^{m-1}
        \right\rangle \subseteq \mathbb{S}_n^{\alpha,\beta};
    \end{equation}
    \item
    if $a = 0$, but $b \neq 0$, then one obtains the algebraic-trigonometric (AT) mixed EC subspace
    \begin{equation}
        \label{eq:AT}
        \left\langle
        \left\{
        u^r \cos\left(bu\right),
        u^r \sin\left(bu\right):
        u \in \left[\alpha,\beta\right]
        \right\}_{r=0}^{m-1}
        \right\rangle \subseteq \mathbb{S}_n^{\alpha,\beta};
    \end{equation}
    \item
    if $a=b=0$ then one has the polynomial (P) EC subspace
    \begin{equation}
        \label{eq:P}	
        \left\langle
        \left\{
        u^r : u \in \left[\alpha,\beta\right]
        \right\}_{r=0}^{m-1}
        \right\rangle \subseteq \mathbb{S}_n^{\alpha,\beta}.
    \end{equation}
\end{itemize}
This means that one can easily define ordinary (mixed) basis functions by simply specifying the factorization of the characteristic polynomial (\ref{eq:characteristic_polynomial}), i.e., one can create (mixed) EC spaces at run-time in an interactive way. As we will see, the zeroth and higher order (endpoint) derivatives of the ordinary basis function will play an important role in case of all proposed algorithms. In case of the aforementioned (mixed) EC subspaces one can overload function operators to compute the required derivatives for arbitrarily fixed orders -- this possibility also motivates our assumption on the structure of the underlying EC space. Moreover, in real-world engineering, computer-aided design and manufacturing applications, usually one defines traditional parametric curves and surfaces by means of the ordinary basis functions presented above and, in our opinion, it would be nice to have a unified framework in which -- apart from general order elevation and subdivision -- one is also able to describe exactly important curves and surfaces by using control points and normalized B-basis functions.

In order to have normalizable bases in the vector space $\mathbb{S}_n^{\alpha,\beta}$, we also have to assume that $z=0$ is at least a first order zero of the characteristic polynomial (\ref{eq:characteristic_polynomial}).

Once we have created an EC space of type $\mathbb{S}_n^{\alpha,\beta}$ by defining its ordinary basis $\mathcal{F}_n^{\alpha,\beta}$ (where $\beta-\alpha \in \left(0,\ell_n^{\prime}\right)$), we also have to generate its unique normalized B-basis $\mathcal{B}_n^{\alpha,\beta}$.
In order to achieve this and to be as self-contained as possible, we recall the construction process \citep{CarnicerMainarPena2004} of $\mathcal{B}_n^{\alpha,\beta}$. As we will see, the steps of this process can be fully implemented in case of the aforementioned (mixed) EC spaces.

Consider the bicanonical basis $\left\{v_{n,i}\left(u\right):u\in\left[\alpha,\beta\right]\right\}_{i=0}^{n}$ formed by the particular integrals (\ref{eq:particular_integrals}) determined by the boundary conditions (\ref{eq:boundary_conditions}). Let $W_{\left[v_{n,n},v_{n,n-1},\ldots,v_{n,0}\right]  }\left(  \beta\right)  $ be the Wronskian matrix of the reverse ordered system
$\{  v_{n,n-i}\left(  u\right)  :\allowbreak{}u\in\left[  \alpha,\beta\right]  \}
_{i=0}^{n}$ at the parameter value $u=\beta$ and obtain its Doolittle-type $LU$ factorization %
\begin{equation}
    L\cdot U=W_{\left[  v_{n,n},v_{n,n-1},\ldots,v_{n,0}\right]  }\left(
    \beta\right),
    \label{eq:LU_factorization_of_reversed_system}
\end{equation}
where $L$ is a lower triangular matrix with unit diagonal,
while $U$ is a non-singular upper triangular matrix. Calculate the inverse matrices%
\[
U^{-1}:=\left[
\begin{array}
[c]{cccc}%
\mu_{0,0} & \mu_{0,1} & \cdots & \mu_{0,n}\\
0 & \mu_{1,1} & \cdots & \mu_{1,n}\\
\vdots & \vdots & \ddots & \vdots\\
0 & 0 & \cdots & \mu_{n,n}%
\end{array}
\right]  ,~L^{-1}:=\left[
\begin{array}
[c]{cccc}%
\lambda_{0,0} & 0 & \cdots & 0\\
\lambda_{1,0} & \lambda_{1,1} & \cdots & 0\\
\vdots & \vdots & \ddots & \vdots\\
\lambda_{n,0} & \lambda_{n,1} & \cdots & \lambda_{n,n}%
\end{array}
\right]
\]
and construct the normalized B-basis%
\begin{equation}
    \mathcal{B}_{n}^{\alpha,\beta}=\left\{  b_{n,i}\left(  u\right)  =\lambda
    _{n-i,0}\widetilde{b}_{n,i}\left(  u\right)  :u\in\left[  \alpha,\beta\right]
    \right\}  _{i=0}^{n}\label{eq:construction}%
\end{equation}
defined by%
\[
\left[
\begin{array}
[c]{cccc}%
\widetilde{b}_{n,n}\left(  u\right)   & \widetilde{b}_{n,n-1}\left(
u\right)   & \cdots & \widetilde{b}_{0}\left(
u\right)
\end{array}
\right]  :=\left[
\begin{array}
[c]{cccc}%
v_{n,n}\left(  u\right)   & v_{n,n-1}\left(  u\right)   & \cdots &
v_{n,0}\left(  u\right)
\end{array}
\right]  \cdot U^{-1}%
\]
and%
\[
\left[
\begin{array}
[c]{cccc}%
\lambda_{0,0} & \lambda_{1,0} & \cdots & \lambda_{n,0}%
\end{array}
\right]  ^{T}:=L^{-1}\cdot\left[
\begin{array}
[c]{cccc}%
1 & 0 & \cdots & 0
\end{array}
\right]  ^{T}.
\]

If the characteristic polynomial (\ref{eq:characteristic_polynomial}) is an either even or odd function, then the underlying EC space $\mathbb{S}_{n}^{\alpha,\beta}$ is invariant under reflections, and in this special case one obtains the symmetry (\ref{eq:symmetry}), i.e., we only need to determine the half of the basis functions
(\ref{eq:construction}). 

\begin{remark}[An alternative construction]The non-negative bicanonical basis $\big\{v_{n,i}:u\in\left[\alpha,\beta\right]\big\}_{i=0}^{n}$ formed by the particular integrals (\ref{eq:particular_integrals}) is a B-basis of the EC space $\mathbb{S}_n^{\alpha,\beta}$ (see \citep[Theorem 2.4/(ii)]{CarnicerMainarPena2004}). Compared with the previously described $LU$ decomposition based method, this B-basis can also be normalized by means of the normalizing coefficients
    \begin{equation}
        c_{n,i}:=-\sum_{r=0}^{i-1} c_{n,r} v_{n,r}^{\left(i\right)}\left(\alpha\right),~i=1,\ldots,n,
        \label{eq:alternative_normalizing_coefficients}
    \end{equation}
    where $c_{n,0}=1$. This means that the unique normalized B-basis functions of the underlying EC space $\mathbb{S}_{n}^{\alpha,\beta}$ could also be determined as the linear combinations
    \begin{equation}
        b_{n,i}\left(u\right) := \sum_{i=0}^{n} c_{n,i} v_{n,i}\left(u\right),~u \in \left[\alpha,\beta\right],~i=0,1,\ldots,n.
        \label{eq:alternative_construction_process}
    \end{equation}
    Although the construction process (\ref{eq:alternative_normalizing_coefficients})--(\ref{eq:alternative_construction_process}) requires less computational effort, several numerical tests show that this alternative method is numerically less stable than (\ref{eq:LU_factorization_of_reversed_system})--(\ref{eq:construction}). 
\end{remark}

Summarizing the calculations of the current section, we can state the next corollary that will be very important both in the formulation and in the implementation of all proposed algorithms.

\begin{corollary}[Differentiation of normalized B-basis functions]\label{cor:B_basis_derivatives}
    In general, the zeroth and higher order differentiation of the constructed normalized B-basis functions (\ref{eq:construction}) can be reduced to the evaluation of formulas
    \begin{align}
        b_{n,n-i}^{\left(  j\right)  }\left(  u\right)    
        & =\lambda_{i,0}\sum_{r=0}^{i}\mu_{r,i}\sum_{k=0}^{n}\rho_{n-r,k}\varphi
        _{n,k}^{\left(  j\right)  }\left(  u\right)  ,~\forall u \in \left[\alpha,\beta\right],~i=0,1,\ldots,n.
        \label{eq:mixed_b_derivatives}
    \end{align}
    Naturally, if the underlying EC space is reflection invariant, then formulas (\ref{eq:mixed_b_derivatives}) have to be applied only for indices $i=0,1,\ldots,\left\lfloor \frac{n}{2}\right\rfloor$, since in this special case one also has that
    \begin{align}
        b_{n,i}^{\left(  j\right)  }\left(  u\right)    & =\left(  -1\right)
        ^{j}b_{n,n-i}^{\left(  j\right)  }\left(  \alpha+\beta-u\right)  ,~\forall u \in \left[\alpha, \beta\right],~i=0,1,\ldots
        ,\left\lfloor \tfrac{n}{2}\right\rfloor.
        \label{eq:mixed_b_symmetric_derivatives}
    \end{align}
\end{corollary}

Examples \ref{exmp:algebraic_trigonometric_space}-- \ref{exmp:exponential_trigonometric_space} and associated Figs.\ \ref{fig:algebraic_trigonometric_order_2}--\ref{fig:exponential_trigonometric_space} provide short examples in algebraic-trigonometric and exponential-trigonometric EC spaces, respectively. Both figures were generated by the help of the proposed function library. (After evaluation and rendering, only the \LaTeX-like labels of the obtained basis functions and some other descriptive elements were added in a post-processing phase. Here we have only illustrated the zeroth order derivatives of the automatically generated ordinary and normalized B-basis functions.)

\begin{example}[A reflection invariant algebraic-trigonometric EC space]
    \label{exmp:algebraic_trigonometric_space}
    Consider the differential equation $v^{\left(9\right)}\left(u\right)+
    6v^{\left(7\right)}\left(u\right)+
    9v^{\left(5\right)}\left(u\right)+
    4v^{\left(3\right)}\left(u\right)=0, ~u\in \left[-\frac{\pi}{2},\frac{\pi}{2}\right]$ and observe that its characteristic polynomial is an odd function that admits the factorization $p_9\left(z\right) = z^3 \prod_{k=1}^2 \left(z^2+k^2\right)^{3-k},~z\in\mathbb{C},$ i.e., the $9$-dimensional algebraic-trigonometric solution space $\mathbb{AT}_8^{-\frac{\pi}{2}, \frac{\pi}{2}}$ of the equation is reflection invariant and is spanned by the ordinary basis $\mathcal{F}_{8}^{-\frac{\pi}{2},\frac{\pi}{2}} = \big\{\varphi_{8,0}\left(u\right)\equiv 1,\allowbreak{} \varphi_{8,1}\left(u\right) = u,\allowbreak{} \varphi_{8,2}\left(u\right) = u^2,\allowbreak{} \varphi_{8,3}\left(u\right)\allowbreak{}=\allowbreak{}\cos\left(u\right), \varphi_{8,4}\left(u\right)=\sin\left(u\right), \varphi_{8,5}\left(u\right)=u\cos\left(u\right),  \varphi_{8,6}\left(u\right)=u\sin\left(u\right), \varphi_{8,7}\left(u\right)=\cos\left(2u\right),\allowbreak{} \varphi_{8,8}\left(u\right)=\sin\left(2u\right) : u \in [-\frac{\pi}{2},\frac{\pi}{2}]\big\}$. Providing the (higher order) zeros of $p_9$ as input parameters, our function library is able to differentiate and render both the ordinary basis and the normalized B-basis of $\mathbb{AT}_8^{-\frac{\pi}{2}, \frac{\pi}{2}}$. The output of our implementation can be seen in Fig.\ \ref{fig:algebraic_trigonometric_order_2}. 
\end{example}

\begin{figure}[!h]
    \centering
    \includegraphics[]{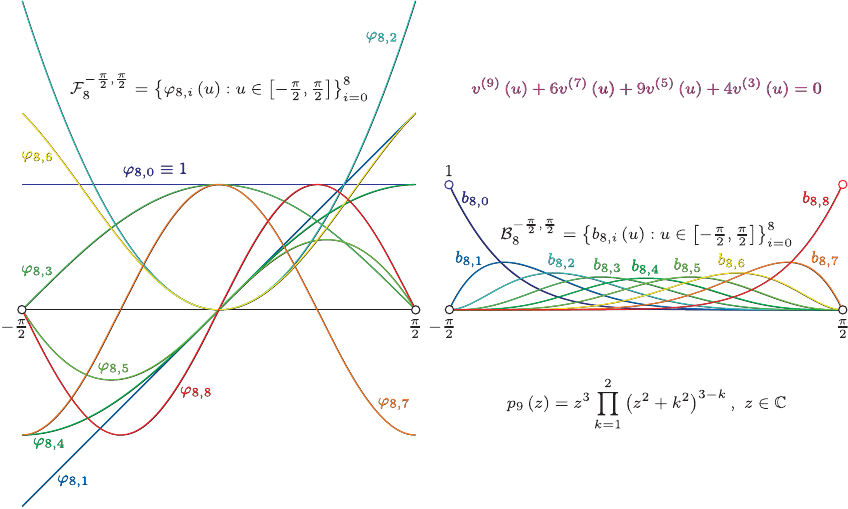}
    \caption{The figure illustrates both the ordinary basis and the normalized B-basis of the  $9$-dimensional reflection invariant algebraic-trigonometric EC space $\mathbb{AT}_8^{-\frac{\pi}{2}, \frac{\pi}{2}}$ described in Example \ref{exmp:algebraic_trigonometric_space}. Although the obtained B-basis functions are very similar to the octic Bernstein-polynomials and seemingly do not provide interesting shape parameters, they may be important in CAGD, since they ensure the integral B-representation of arcs/patches of important transcendental (cycloidal, helicoidal) curves/surfaces that cannot be described by the standard rational B\'ezier or NURBS models.}
    \label{fig:algebraic_trigonometric_order_2}
\end{figure}

\begin{example}[A not reflection invariant exponential-trigonometric EC space]
    \label{exmp:exponential_trigonometric_space}
    Consider the differential equation $v^{\left(  7\right)  }\left(  u\right)  -11v^{\left(  6\right)  }\left(
    u\right)  +44v^{\left(  5\right)  }\left(  u\right)  -78v^{\left(  4\right)
    }\left(  u\right)  +77v^{\left(  3\right)  }\left(  u\right)  -67v^{\left(
    2\right)  }\left(  u\right)  +34v^{\left(  1\right)  }\left(  u\right)  =0,~u\in\left[-2,\frac{1}{8}\right]$ and observe that its characteristic polynomial can be factorized into the form  $p_{7}\left(  z\right)  =z\left(  z-\mathbf{i}\right)  \left(  z+\mathbf{i}\right)  \left(  z-1\right)  \left(  z-2\right)  \left(  z-\left(4-\mathbf{i}\right)  \right)  \left(  z-\left(  4+\mathbf{i}\right)  \right),\allowbreak{}\,z\in\mathbb{C}$, i.e., the solution space of the equation is the $7$-dimensional exponential-trigonometric EC space $\mathbb{ET}_6^{-2,\frac{1}{8}}=\big\langle \mathcal{F}_6^{-2,\frac{1}{8}}\big\rangle = \big\langle\big\{\varphi_{6,0}\left(u\right)=1, \varphi_{6,1}\left(u\right) = \cos\left(u\right), \varphi_{6,2}\left(u\right) = \sin\left(u\right), \varphi_{6,3}\left(u\right) = e^{u}, \varphi_{6,4}\left(u\right) = e^{2u},\allowbreak{} \varphi_{6,5}\left(u\right) =\allowbreak{}e^{4 u}\cos\left(u\right),\allowbreak{} \varphi_{6,6}\left(u\right) = e^{4 u} \sin\left(u\right): u \in\big[-2,\tfrac{1}{8}\big]\big\}\big\rangle$. Fig.\ \ref{fig:exponential_trigonometric_space} illustrates both the ordinary basis and the normalized B-basis of the space. Compared with Example \ref{exmp:algebraic_trigonometric_space}, it can be observed that the obtained normalized B-basis functions are not symmetric under the reflection of the definition domain. (Observe that the underlying EC space also comprises transcendental functions that cannot represented by the unique normalized B-bases of M\"untz (or in special, polynomial) EC spaces.)
\end{example}

\begin{figure}[!h]
    \centering
    \includegraphics[]{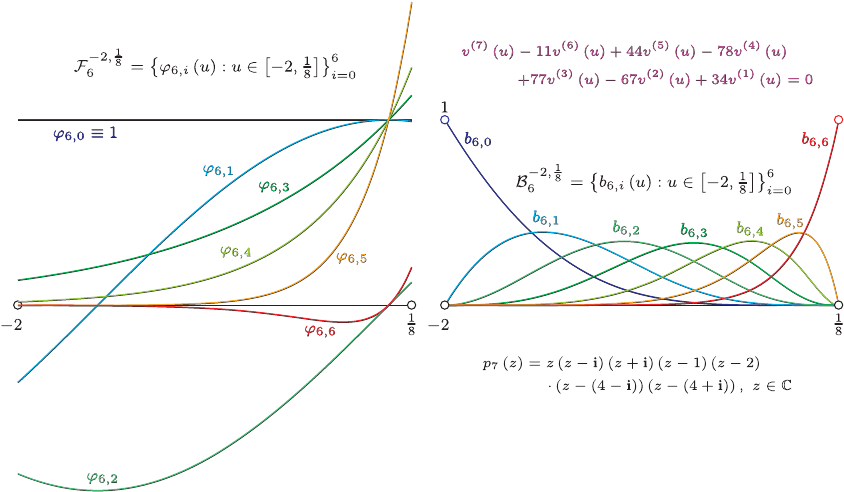}
    \caption{The figure illustrates both the ordinary basis and the normalized B-basis of the $7$-dimensional not reflection invariant exponential-trigonometric EC space $\mathbb{ET}_6^{-2, \frac{1}{8}}$ described in Example \ref{exmp:exponential_trigonometric_space}.} 
    \label{fig:exponential_trigonometric_space}
\end{figure}

\subsection{General dimension and order elevation}

Consider the EC spaces $\mathbb{S}_n^{\alpha,\beta}$ and $\mathbb{S}_{n+1}^{\alpha,\beta}$ such that $1\in \mathbb{S}_n^{\alpha,\beta} \subset \mathbb{S}_{n+1}^{\alpha,\beta}$ and assume that spaces $D\mathbb{S}_{n}^{\alpha,\beta}$ and $D\mathbb{S}_{n+1}^{\alpha,\beta}$ of the derivatives are also EC on $\left[\alpha,\beta\right]$, i.e., $
0<\beta-\alpha<\min\left\{\ell^{\prime}\left(\mathbb{S}_{n}^{\alpha,\beta}\right),\ell^{\prime}\left(\mathbb{S}_{n+1}^{\alpha,\beta}\right)\right\},
$
furthermore let us denote their unique normalized B-bases by $\left\{b_{n,i}\left(u\right):u\in\left[\alpha,\beta\right]\right\}_{i=0}^n$ and $\big\{b_{n+1,i}\left(u\right):u\in\left[\alpha,\beta\right]\big\}_{i=0}^{n+1}$, respectively.

Following the results of \citep[Theorem 3.1]{MazureLaurent1998} and, as a slight difference, considering both endpoints of the definition domain $\left[\alpha, \beta\right]$ in order to minimize the maximal differentiation order of the lower and higher order normalized B-basis functions, one can state the next lemma.

\begin{lemma}[General order elevation, \citep{MazureLaurent1998}]
    \label{lem:general_order_elevation}
    Using the notations of the section, the $n$th order B-curve (\ref{eq:B_curve}) fulfills the identity
    $$
    \mathbf{c}_n\left(u\right)=\sum_{i=0}^n \mathbf{p}_i b_{n,i}\left(u\right)\equiv\sum_{i=0}^{n+1}\mathbf{p}_{1,i} b_{n+1,i}\left(u\right)=:\mathbf{c}_{n+1}\left(u\right),~\forall u \in \left[\alpha, \beta\right],
    $$
    where $\mathbf{p}_{1,0} = \mathbf{p}_{0}, ~\mathbf{p}_{1,n+1} = \mathbf{p}_n$ and
    \begin{align}
        \mathbf{p}_{1,i} 
        &= 
        \left(
        1-
        \frac
        {
            b_{n,i}^{\left(i\right)}\left(\alpha\right)
        }
        {
            b_{n+1,i}^{\left(i\right)}\left(\alpha\right)
        }
        \right)
        \mathbf{p}_{i-1}
        +
        \frac
        {
            b_{n,i}^{\left(i\right)}\left(\alpha\right)
        }
        {
            b_{n+1,i}^{\left(i\right)}\left(\alpha\right)
        }
        \mathbf{p}_i,~i=1,\ldots,\left\lfloor \frac{n}{2} \right\rfloor,
        \label{eq:order_elevation_inner_points_first_half}
        \\
        \mathbf{p}_{1,n+1-i} 
        & = \frac{b_{n,n-i}^{\left(i\right)}\left(\beta\right)}{b_{n+1,n+1-i}^{\left(i\right)}\left(\beta\right)} \mathbf{p}_{n-i} + 
        \left(1-\frac{b_{n,n-i}^{\left(i\right)}\left(\beta\right)}{b_{n+1,n+1-i}^{\left(i\right)}\left(\beta\right)}\right)\mathbf{p}_{n+1-i},~i=1,\ldots,\left\lfloor\frac{n+1}{2}\right\rfloor.
        \label{eq:order_elevation_inner_points_second_half}
    \end{align}
\end{lemma}

Although Lemma \ref{lem:general_order_elevation} is valid for any nested EC spaces that fulfill the conditions $1\in \mathbb{S}_n^{\alpha,\beta} \subset \mathbb{S}_{n+1}^{\alpha,\beta}$ and for which the derivative spaces $D\mathbb{S}_{n}^{\alpha,\beta}$ and $D\mathbb{S}_{n+1}^{\alpha,\beta}$ are also EC, in case of our implementation, we always assume that the higher dimensional EC space $\mathbb{S}_{n+1}^{\alpha,\beta}$ can also be identified with the solution space of a constant-coefficient homogeneous linear differential equation. Naturally, the results of Lemma \ref{lem:general_order_elevation} can also be extended to the general order elevation of B-surfaces of type (\ref{eq:B-surface}) and our function library ensures this possibility as well as it is illustrated in Fig.\ \ref{fig:order_elevated_snails} associated with the next example.

\begin{example}[Order elevation of B-surfaces]
    \label{exmp:snail}
    Consider the ordinary exponential-trigonometric integral surface
    \begin{equation}
        \label{eq:snail}
        \def\arraystretch{1.3}
        \mathbf{s}\left(u_0, u_1\right) = 
        \left[
        \begin{array}{c}
            s^0\left(u_0, u_1\right)\\
            s^1\left(u_0, u_1\right)\\
            s^2\left(u_0, u_1\right)
        \end{array}
        \right]
        =
        \left[
        \begin{array}{c}
            \left(1-e^{\omega_0 u_0}\right) \cos\left(u_0\right) \left(\frac{5}{4}+\cos\left(u_1\right)\right)\\
            \left(e^{\omega_0 u_0}-1\right) \sin\left(u_0\right) \left(\frac{5}{4}+\cos\left(u_1\right)\right)\\
            7-e^{\omega_1 u_0} - \sin\left(u_1\right) + e^{\omega_0 u_0} \sin\left(u_1\right)
        \end{array}
        \right],
        \def\arraystretch{1.0}
    \end{equation}
    where 
    $
    \left(u_0, u_1\right)
    \in
    \left[
    \tfrac{7\pi}{2},\tfrac{49\pi}{8}
    \right]
    \times
    \left[
    -\tfrac{\pi}{3},\tfrac{5\pi}{3}
    \right], ~\omega_0 = \tfrac{1}{6\pi}$ and $ \omega_1 = \tfrac{1}{3\pi}. 
    $
    In order to describe any restriction (i.e., patch) $\left.\mathbf{s}\right|_{\left[\alpha_0,\beta_0\right]\times\left[\alpha_1,\beta_1\right]}$ by means of B-surfaces of type (\ref{eq:B-surface}), in directions $u_0$ and $u_1$ one needs to define parent EC spaces that include the subspaces
    $
        \mathbb{ET}_{6}^{\alpha_0, \beta_0} 
        := \big\langle \mathcal{E}\mathcal{T}_6^{\alpha_0, \beta_0} \big\rangle 
        := \big\langle \big\{1,\cos\left(u_0\right), \sin\left(u_0\right),\allowbreak{}  e^{\omega_0 u_0},\allowbreak{} e^{\omega_1 u_0}, \allowbreak{}e^{\omega_0 u_0}\allowbreak{} \cos\left(u_0\right), \allowbreak{}e^{\omega_0 u_0} \sin\left(u_0\right) : u_0\in\left[\alpha_0,\beta_0\right]\big\} \big\rangle
    $
    and 
    $
    \mathbb{T}_2^{\alpha_1, \beta_1}:=\allowbreak{}\big\langle\mathcal{T}_2^{\alpha_1, \beta_1}\big\rangle\allowbreak{}:=\allowbreak{}\big\langle\{1,\allowbreak{}\cos\left(u_1\right),\sin\left(u_1\right):\allowbreak{}u_1\in\left[\alpha_1,\beta_1\right]\}\big\rangle,$
    respectively, where the interval lengths $\beta_0-\alpha_0>0$ and $\beta_1-\alpha_1>0$ must be less than the critical lengths of the derivative spaces of the parent ones. Observe that $\mathbb{ET}_{6}^{\alpha_0, \beta_0}$ and $\mathbb{T}_2^{\alpha_1, \beta_1}$ can be identified with the solution spaces of those differential equations of type (\ref{eq:differential_equation}) whose characteristic polynomials can be factorized into 
    $
    p_7\left(z\right)=z\left(z-\mathbf{i}\right)\left(z+\mathbf{i}\right)\left(z-\omega_0\right)\left(z-\omega_1\right)\allowbreak{}\left(z-\left(\omega_0-\mathbf{i}\right)\right)\left(z-\left(\omega_0+\mathbf{i}\right)\right)
    $
    and
    $
    p_3\left(z\right)=z\left(z-\mathbf{i}\right)\left(z+\mathbf{i}\right),
    $
    respectively, where $z\in\mathbb{C}$. Providing the zeros of these polynomials as input parameters, the proposed function library is able to perform general order elevation either by increasing the order of one of these zeros, or by specifying new ones of single or higher multiplicity. Case (\textit{a}) of Fig.\ \ref{fig:order_elevated_snails} illustrates a control net of minimal size that allows the control point based exact description of the surface patch  $\left.\mathbf{s}\right|_{\left[\frac{11\pi}{2},\frac{49\pi}{8}\right]\times\left[-\frac{\pi}{3},\frac{\pi}{3}\right]}$, possible order elevations of which can be done in infinitely many ways, e.g., in case (\textit{b}) of Fig.\ \ref{fig:order_elevated_snails} in directions $u_0$ and $u_1$ we have applied the normalized B-bases of the higher dimensional algebraic-exponential-trigonometric and algebraic-trigonometric EC spaces $\mathbb{AET}_7^{\alpha_0,\beta_0}:=\big\langle\mathcal{E}\mathcal{T}_{6}^{\alpha_0, \beta_0} \cup\big\{ u_0 : u_0 \in \left[\alpha_0,\beta_0\right]
    \big\}\big\rangle$ and $\mathbb{AT}_3^{\alpha_1,\beta_1}:=\big\langle\mathcal{T}_2^{\alpha_1,\beta_1}\cup\big\{u_1 : u_1 \in \left[\alpha_1,\beta_1\right]
    \big\}\big\rangle$, respectively, while in case (\textit{c}) of Fig.\ \ref{fig:order_elevated_snails} we have considered the exponential-trigonometric and algebraic-trigonometric EC spaces $\mathbb{ET}_8^{\alpha_0,\beta_0}:= \big\langle\mathcal{E}\mathcal{T}_{6}^{\alpha_0, \beta_0} \cup\big\{e^{-\omega_0 u_0}, e^{-\omega_1 u_0} : u_0 \in \left[\alpha_0,\beta_0\right]
    \big\}\big\rangle$ and  $\mathbb{AT}_4^{\alpha_1,\beta_1}:=\big\langle\mathcal{T}_2^{\alpha_1,\beta_1}\cup\big\{u_1\cos\left(u_1\right),u_1\sin\left(u_1\right) : u_1 \in \left[\alpha_1,\beta_1\right]
    \big\}\big\rangle$, respectively. Observe that in case (\textit{b}) we have increased in both directions the order of the already existing zero $z=0$, i.e., we defined the new characteristic polynomials $p_8\left(z\right) = z p_7\left(z\right)$ and $p_4\left(z\right) = z p_3\left(z\right)$, while in case (\textit{c}) we applied the polynomials $p_9\left(z\right) = \left(z+\omega_0\right)\left(z+\omega_1\right)p_7\left(z\right)$ and $p_5\left(z\right) = \left(z-\mathbf{i}\right)\left(z+\mathbf{i}\right)p_3\left(z\right)$, i.e., in directions $u_0$ and $u_1$ we have introduced the new zeros $z = -\omega_0$ and $z = -\omega_1$ and increased the multiplicity of the existing zeros $z=\pm \mathbf{i}$ from $1$ to $2$, respectively.
\end{example}

\begin{figure}[!h]
    \centering
    \includegraphics[]{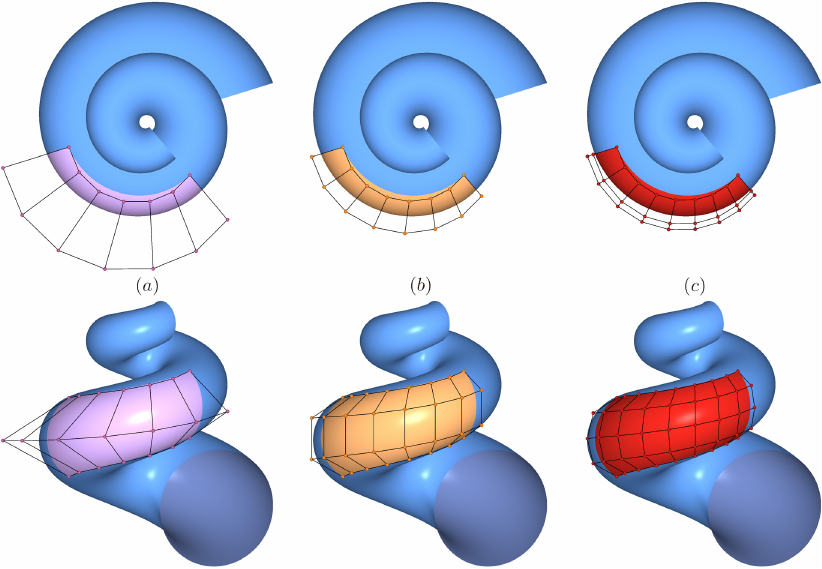}
    \caption{(\textit{a}) Control point based exact description of the patch $\left.\mathbf{s}\right|_{\left[\frac{11\pi}{2},\frac{49\pi}{8}\right]\times\left[-\frac{\pi}{3},\frac{\pi}{3}\right]}$ of the ordinary exponential-trigonometric integral surface (\ref{eq:snail}) by means of an B-surface of minimal order. Cases (\textit{b}) and (\textit{c}) illustrate two possible order elevations of the same patch by using the normalized B-bases of higher dimensional EC spaces. (More details can be found in Example \ref{exmp:snail}. All images were rendered by means of the proposed function library.) }
    \label{fig:order_elevated_snails}
\end{figure}

\subsection{General B-algorithm}\label{subsec:general_B_algorithm}

Theoretically, every normalized B-basis implies a B-algorithm for the subdivision of B-curves like (\ref{eq:B_curve}), i.e., for an arbitrarily fixed parameter value $\gamma \in \left(\alpha, \beta\right)$ there exists a recursive corner cutting de Casteljau-like algorithm that starts with the initial conditions %
$
\mathbf{p}_i^0\left(\gamma\right)\equiv\mathbf{p}_i,~ i=0,1,\ldots,n
$
and recursively defines the subdivision points
\begin{equation}
    \label{eq:subdivision}
    \mathbf{p}^{j}_i\left(\gamma\right)=
    \left(1-\xi_i^j\left(\gamma\right)\right)
    \cdot \mathbf{p}_{i}^{j-1}\left(\gamma\right)
    +
    \xi_i^j\left(\gamma\right)
    \cdot \mathbf{p}_{i+1}^{j-1}\left(\gamma\right),
    ~ i = 0,\ldots,n-j,~
    j = 1,\ldots,n,
\end{equation}
where the explicit closed forms of the blending functions
$\left\{\xi_i^j:\left[\alpha,\beta\right]\to\left[0,1\right]\right\}_{i=0,\,j=1}^{n-j,\,n}$, in general, are either not known, or, apart from some very special cases (like B\'ezier curves), usually have non-linear complicated expressions even in low-dimensional EC spaces.

Using blossoms, B-algorithms were theoretically  characterized in \citep[Theorem 2.4]{Pottmann1993} by means of a non-con\-struc\-tive procedure relying on unevaluated exterior products which, unfortunately, are not very useful concerning implementation. Even the author of \citep[Theorem 2.4]{Pottmann1993} states that the steps of his theoretical ``\textit{construction can be used for an implementation if the maps from the parameter interval to the axis [...] are simple. This is the case for rational B\'ezier curves where we have projective maps. Then we get exactly Farin's projective version of the de Casteljau algorithm ([Farin, '83]) see also [Farin \& Worsey, '91]. The value of the algorithm for general TB-curves lies on the theoretical side.}" 

Subdivision related constructive algorithms appear, e.g., in \citep[Section 3]{MainarPena1999}. These methods use Neville elimination and are based both on $LU$ decomposition of non-singular stochastic square matrices and on (unique) bidiagonal decompositions of non-singular lower triangular stochastic matrices \citep[Theorem 4.5]{GascaPena1996}.

Compared with subdivision techniques presented in \citep[Theorem 2.4]{Pottmann1993} and \citep[Section 3]{MainarPena1999}, in what follows, we propose simple, efficient and easily implementable constructive formulas which are based only on continuity conditions and avoid the evaluation of the non-diagonal entries of the triangular scheme that can be associated with every B-algorithm.

Let $\mathbb{S}_n^{\alpha,\beta}$ be an EC space, where $1\in\mathbb{S}_n^{\alpha,\beta}$ and $\beta-\alpha \in \left(0,\ell_n^{\prime}\right)$. Consider the B-curve (\ref{eq:B_curve}) and let $
\mathcal{B}_{n}^{\alpha,\gamma}:=\{  b_{n,i}\left(  u;\alpha
,\gamma\right)  :\allowbreak{}u\allowbreak{}\in\allowbreak{}\left[  \alpha,\gamma\right]  \}  _{i=0}^{n}%
$
and $\mathcal{B}_{n}^{\gamma,\beta}:=\{  b_{n,i}\left(  u;\gamma,\beta\right)
:u \in\allowbreak{}[  \gamma,\beta]  \}  _{i=0}^{n}$ be the unique normalized B-bases of the restricted EC spaces $\mathbb{S}_n^{\alpha,\gamma}:=\operatorname{span}\mathcal{F}_{n}^{\alpha,\gamma}\allowbreak{}:=\operatorname{span}\allowbreak{}\left.  \mathcal{F}_{n}^{\alpha,\beta
}\right\vert _{\left[  \alpha,\gamma\right]  }%
$ and $
\mathbb{S}_n^{\gamma,\beta}:=\operatorname{span}\mathcal{F}_{n}^{\gamma,\beta}:=\operatorname{span}\left.  \mathcal{F}_{n}^{\alpha,\beta
}\right\vert _{\left[  \gamma,\beta\right]  }
$, respectively. Consider also the diagonal entries
$$
\left\{\Red{\boldsymbol{\lambda}_i\left(\gamma\right)}:=\Red{\mathbf{p}_0^i\left(\gamma\right)}\right\}_{i=0}^n
\text{ and }
\left\{\Blue{\boldsymbol{\varrho}_i\left(\gamma\right)}:=\Blue{\mathbf{p}_{i}^{n-i}\left(\gamma\right)}\right\}_{i=0}^n
$$
of the triangular scheme 
\[%
\vspace{-3mm}
\def\arraystretch{1.3}
\begin{array}
[c]{lllll}%
\mathbf{p}_{0}=:%
\Red{\boldsymbol{\lambda}_{0}\left(\gamma\right)} &  &  &  & \\
\mathbf{p}_{1} & \mathbf{p}_{0}^{1}\left(  \gamma\right)  =:%
\Red{\boldsymbol{\lambda}_{1}\left(\gamma\right)} &  &  & \\
\mathbf{p}_{2} & \mathbf{p}_{1}^{1}\left(  \gamma\right)   & \mathbf{p}%
_{0}^{2}\left(  \gamma\right)  =:%
\Red{\boldsymbol{\lambda}_{2}\left(\gamma\right)} &  & \\
\vdots & \vdots & \vdots & \cdots & \mathbf{p}_{0}^{n}\left(  \gamma\right)  =:%
\Red{\boldsymbol{\lambda}_{n}\left(\gamma\right)}=:\Blue{\boldsymbol{\varrho}_{0}\left(\gamma\right)}\\
\mathbf{p}_{n-2} & \mathbf{p}_{n-2}^{1}\left(  \gamma\right)
& \mathbf{p}_{n-2}^{2}\left(  \gamma\right)  =:\Blue{\boldsymbol{\varrho}_{n-2}\left(\gamma\right)} &  & %
\\
\mathbf{p}_{n-1}
&
\mathbf{p}_{n-1}^{1}\left(  \gamma\right)
=:\Blue{\boldsymbol{\varrho}_{n-1}\left(\gamma\right)}
\\
\mathbf{p}_{n}=:\Blue{\boldsymbol{\varrho}_{n}\left(\gamma\right)}
\end{array}
\def\arraystretch{1.0}
\]
that can be associated with the recursive process (\ref{eq:subdivision}). Blending these points with the functions of the normalized B-bases $\mathcal{B}_n^{\alpha,\gamma}$ and $\mathcal{B}_n^{\gamma, \beta}$, the B-curve (\ref{eq:B_curve}) of order $n$ can be subdivided into the left and right arcs 
\begin{equation}
    \label{eq:left_arc}
    \Red{\boldsymbol{l}_n}\left(u\right) := \displaystyle{}\sum_{i=0}^{n}
    \Red{\boldsymbol{\lambda}_i\left(\gamma\right)
    } \cdot b_{n,i}\left(u;\alpha,\gamma\right)\equiv\mathbf{c}_n\left(u\right),~\forall u\in\left[\alpha,\gamma\right]
\end{equation}
and
\begin{equation}
    \label{eq:right_arc}
    \Blue{\mathbf{r}_n}\left(u\right) := \displaystyle{}\sum_{i=0}^{n}
    \Blue{
        \boldsymbol{\varrho}_i\left(\gamma\right)
    }\cdot b_{n,i}\left(u;\gamma,\beta\right) \equiv \mathbf{c}_n\left(u\right),~\forall u\in\left[\gamma,\beta\right],
\end{equation}
respectively, that also fulfill the identities 
\begin{align}
    \Red{\boldsymbol{l}_n^{\Black{\left(j\right)}}}\left(u\right)&=\mathbf{c}_n^{\left(j\right)}\left(u\right), ~\forall u \in \left[\alpha,\gamma\right],
    \label{eq:subdivision_left_alpha}
    \\
    \Blue{\mathbf{r}_n^{\Black{\left(j\right)}}}\left(u\right)&=\mathbf{c}_n^{\left(j\right)}\left(u\right),~\forall u \in \left[\gamma, \beta\right]
    \label{eq:subdivistion_right_beta}
\end{align}
for all differentiation orders $j\geq 0$.

We close the current section with a recursive method by means of which one can determine the unknown diagonal subdivision points  $\left\{\Red{\boldsymbol{\lambda}_i\left(\gamma\right)}\right\}_{i=0}^n$ and $\left\{\Blue{\boldsymbol{\varrho}_{i}\left(\gamma\right)}\right\}_{i=0}^n$ even in the absence of the usually unknown blending functions $\left\{\xi_i^j:\left[\alpha,\beta\right]\to\left[0,1\right]\right\}_{i=0,\,j=1}^{n-j,\,n}$.

\begin{theorem}[General B-algorithm]
    \label{thm:general_subdivision}
    Given an arbitrarily fixed parameter value $\gamma\in\left(\alpha, \beta\right)$ and starting with the initial conditions
    \begin{align}
        \Red{\boldsymbol{\lambda}_0\left(\gamma\right)}&=\mathbf{c}_n\left(\alpha\right)=\mathbf{p}_0,
        \label{eq:left_c_initial_condition_alpha}
        \\
        \Red{\boldsymbol{\lambda}_n\left(\gamma\right)}&=\mathbf{c}_n\left(\gamma\right)=\Blue{\boldsymbol{\varrho}_0\left(\gamma\right)},
        \label{eq:left_c_right_initial_condition_gamma}
        \\
        \Blue{\boldsymbol{\varrho}_n\left(\gamma\right)} &= \mathbf{c}_n\left(\beta\right) = \mathbf{p}_n,
        \label{eq:right_c_initial_condition_beta}
    \end{align}
    the unknown diagonal subdivision points $\left\{\Red{\boldsymbol{\lambda}_i\left(\gamma\right)}\right\}_{i=0}^n$ and $\left\{\Blue{\boldsymbol{\varrho}_{i}\left(\gamma\right)}\right\}_{i=0}^n$ can be iteratively determined by means of the recursive formulas
    \begin{align}
        \Red{\boldsymbol{\lambda}_i\left(\gamma\right)}
        &=
        \frac{1}{b_{n,i}^{\left(i\right)}\left(\alpha;\alpha,\gamma\right)}
        \left(
        \mathbf{c}_n^{\left(i\right)}\left(\alpha\right)
        -
        \sum_{j = 0}^{i-1} \Red{\boldsymbol{\lambda}_j\left(\gamma\right)}
        \cdot 
        b_{n,j}^{\left(i\right)}\left(\alpha; \alpha, \gamma\right)
        \right),~i = 1,\ldots,\left\lfloor\frac{n-1}{2}\right\rfloor,
        \label{eq:left_subdivision_points_first_half}
        \\
        \Red{\boldsymbol{\lambda}_{n-i}\left(\gamma\right)}
        &=
        \frac{1}{b_{n,n-i}^{\left(i\right)}\left(\gamma;\alpha,\gamma\right)}
        \left(
        \mathbf{c}_n^{\left(i\right)}\left(\gamma\right)
        -
        \sum_{j = 0}^{i-1}
        \Red{\boldsymbol{\lambda}_{n-j}\left(\gamma\right)}
        \cdot
        b_{n,n-j}^{\left(i\right)}\left(\gamma; \alpha, \gamma\right)
        \right),
        ~i=1,\ldots,\left\lfloor\frac{n}{2}\right\rfloor,
        \label{eq:left_subdivision_points_second_half}
        \\
        \Blue{\boldsymbol{\varrho}_i\left(\gamma\right)}
        &=
        \frac{1}{b_{n,i}^{\left(i\right)}\left(\gamma;\gamma,\beta\right)}
        \left(
        \mathbf{c}_n^{\left(i\right)}\left(\gamma\right)
        -
        \sum_{j = 0}^{i-1}
        \Blue{\boldsymbol{\varrho}_j\left(\gamma\right)}
        \cdot
        b_{n,j}^{\left(i\right)}\left(\gamma;\gamma,\beta\right)
        \right),
        ~i=1,\ldots,\left\lfloor\frac{n}{2}\right\rfloor,
        \label{eq:right_subdivision_points_first_half}
        \\
        \Blue{\boldsymbol{\varrho}_{n-i}\left(\gamma\right)}
        &=
        \frac{1}{b_{n,n-i}^{\left(i\right)}\left(\beta;\gamma,\beta\right)}
        \left(
        \mathbf{c}_n^{\left(i\right)}\left(\beta\right)
        -
        \sum_{j=0}^{i-1}
        \Blue{\boldsymbol{\varrho}_{n-j}\left(\gamma\right)}
        \cdot
        b_{n,n-j}^{\left(i\right)}\left(\beta;\gamma,\beta\right)
        \right),~
        i=1,\ldots,\left\lfloor\frac{n-1}{2}\right\rfloor.
        \label{eq:right_subdivision_points_second_half}
    \end{align}
\end{theorem}

\begin{proof}
    The calculation of the unknown subdivision points $\left\{\Red{\boldsymbol{\lambda}_i\left(\gamma\right)}\right\}_{i=0}^n$ and $\left\{\Blue{\boldsymbol{\varrho}_{i}\left(\gamma\right)}\right\}_{i=0}^n$ can be reduced to the combined application of differentiation identities (\ref{eq:subdivision_left_alpha})--(\ref{eq:subdivistion_right_beta}) and of those Hermite-type endpoint conditions  (\ref{eq:endpoint_interpolation})--(\ref{eq:Hermite_conditions_alpha})
    that are fulfilled by the normalized B-bases $\mathcal{B}_{n}^{\alpha,\beta}$, $\mathcal{B}_{n}^{\alpha,\gamma}$ and $\mathcal{B}_{n}^{\gamma,\beta}$ corresponding to the intervals $\left[\alpha,\beta\right]$, $\left[\alpha,\gamma\right]$ and $\left[\gamma,\beta\right]$, respectively.
    
    For example, the initial condition (\ref{eq:left_c_initial_condition_alpha}) follows from the endpoint interpolation properties of the B-curves (\ref{eq:B_curve}) and (\ref{eq:left_arc}), since
    \[
    \mathbf{c}_n\left(\alpha\right) = 
    \sum_{i=0}^n \mathbf{p}_i \cdot b_{n,i}\left(\alpha;\alpha,\beta\right)
    = \mathbf{p}_0 
    =\Red{\boldsymbol{\lambda}_0\left(\gamma\right)}
    = \sum_{i=0}^n\Red{\boldsymbol{\lambda}_i\left(\gamma\right)}\cdot b_{n,i}\left(\alpha; \alpha, \gamma\right)
    =
    \Red{\boldsymbol{l}_n}\left(\alpha\right).
    \]
    At the same time, for all differentiation orders $i=1,\ldots,\left\lfloor\frac{n-1}{2}\right\rfloor$ one obtains both the condition $b_{n,i}^{\left(i\right)}\left(\alpha;\alpha,\gamma\right)>0$ and the equality
    \begin{align*}
        \mathbf{c}_n^{\left(i\right)}\left(\alpha\right)
        &=
        \Red{\boldsymbol{l}_n^{\Black{\left(i\right)}}}\left(\alpha\right)
        \\
        &
        =
        \sum_{j=0}^{n}\Red{\boldsymbol{\lambda}_j\left(\gamma\right)}\cdot b_{n,j}^{\left(i\right)}\left(\alpha;\alpha,\gamma\right)
        \\
        &
        =
        \sum_{j=0}^{i}\Red{\boldsymbol{\lambda}_j\left(\gamma\right)}\cdot b_{n,j}^{\left(i\right)}\left(\alpha;\alpha,\gamma\right)
        \\
        &=
        \sum_{j=0}^{i-1}\Red{\boldsymbol{\lambda}_j\left(\gamma\right)}\cdot b_{n,j}^{\left(i\right)}\left(\alpha;\alpha,\gamma\right)
        +
        \Red{\boldsymbol{\lambda}_i\left(\gamma\right)}\cdot b_{n,i}^{\left(i\right)}\left(\alpha;\alpha,\gamma\right),
    \end{align*}
    from which follows exactly formula (\ref{eq:left_subdivision_points_first_half}) for the unknown subdivision point $\Red{\boldsymbol{\lambda}_i\left(\gamma\right)}$. The remaining  recursive formulas {(\ref{eq:left_subdivision_points_second_half})--(\ref{eq:right_subdivision_points_second_half})} can be proved in a similar way. Observe that each recursion would also be valid for arbitrary values of the index $i\in\left\{1,\ldots,n\right\}$. In the statement of the theorem we have restricted the index domain of each formula, since the maximum order of B-basis function derivatives that have to be evaluated has to be as small as possible in order to ensure greater efficiency and numerical stability for the just presented general B-algorithm.
\end{proof}

The subdivision technique presented in Theorem \ref{thm:general_subdivision} can also be extended to B-surfaces of the type (\ref{eq:B-surface}) as it is shown in Fig.\ \ref{fig:subdivided_snail}.

\begin{figure}[!h]
    \centering
    \includegraphics[scale = 0.9125]{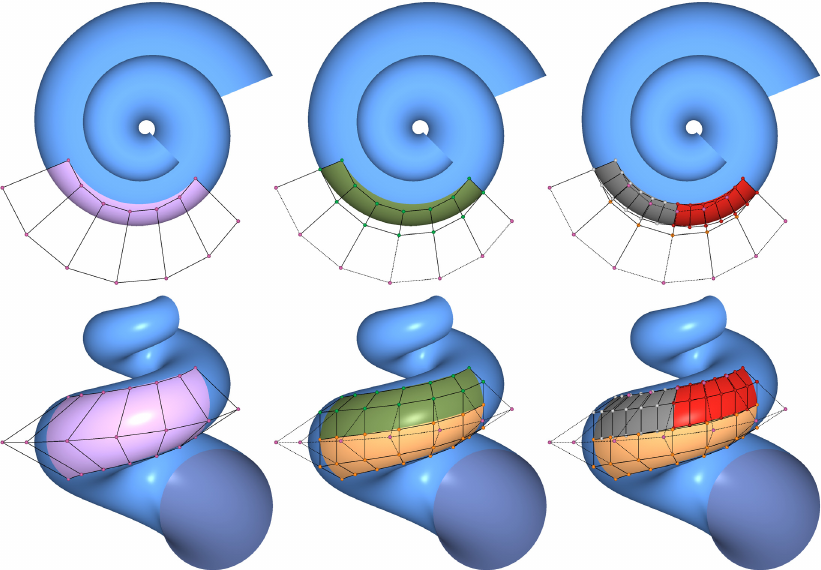}
    \caption{The surface patch $\left.\mathbf{s}\right|_{\left[\frac{11\pi}{2},\frac{49\pi}{8}\right]\times\left[-\frac{\pi}{3},\frac{\pi}{3}\right]}$ of the ordinary exponential-trigonometric integral surface (\ref{eq:snail}) is subdivided first at the parameter value $u_1=0$, then one of the obtained surface patches is further subdivided at the parameter value $u_0=\frac{93\pi}{16}$. (All images were rendered by means of the proposed function library.)}
    \label{fig:subdivided_snail}
\end{figure}

\subsection{General basis transformation}

In \citep{Roth2015b} we have already constructed the matrix of the general basis transformation that maps the normalized B-basis $\mathcal{B}_n^{\alpha,\beta}$ to the ordinary basis $\mathcal{F}_{n}^{\alpha,\beta}$ of the EC space $\mathbb{S}_n^{\alpha,\beta}$, where $\beta-\alpha \in \left(0,\ell_n^{\prime}\right)$. Namely, we have the next theorem.

\begin{theorem}
    [General basis transformation, \citep{Roth2015b}]\label{thm:basis_transformation}The matrix form of the
    linear transformation that maps the normalized B-basis $\mathcal{B}%
    _{n}^{\alpha,\beta}$ to the ordinary basis $\mathcal{F}_{n}^{\alpha,\beta}$ is%
    \begin{equation}
        \left[
        \begin{array}
            [c]{c}%
            \varphi_{n,i}\left(  u\right)
        \end{array}
        \right]  _{i=0}^{n}=\left[  t_{i,j}^{n}\right]  _{i=0,~j=0}^{n,~n}\cdot\left[
        \begin{array}
            [c]{c}%
            b_{n,i}\left(  u\right)
        \end{array}
        \right]  _{i=0}^{n},~\forall u\in\left[  \alpha,\beta\right]  ,
        \label{eq:basis_transformation}%
    \end{equation}
    where $t_{0,j}^{n}=1,~j=0,1,\ldots,n$ and $t_{i,0}^{n}= \varphi_{n,i}\left(  \alpha\right),~
    t_{i,n}^{n}=  ~\varphi_{n,i}\left(  \beta\right)  ,~i=0,1,\ldots
    ,n$, while%
    \begin{align}
        t_{i,j}^{n}=  &  ~\varphi_{n,i}\left(  \alpha\right)  -\frac{1}{b_{n,j}%
            ^{\left(  j\right)  }\left(  \alpha\right)  }\cdot\left.  \sum_{r=1}%
        ^{j-1}\frac{\varphi_{n,i}^{\left(  r\right)  }\left(  \alpha\right)  }%
        {b_{n,r}^{\left(  r\right)  }\left(  \alpha\right)  }\right(  b_{n,r}^{\left(
            j\right)  }\left(  \alpha\right)  +\label{eq:first_half}\\
        &  ~\left.  +\sum_{\ell=1}^{j-r-1}\left(  -1\right)  ^{\ell}\sum
        _{r<k_{1}<k_{2}<\ldots<k_{\ell}<j}\frac{b_{n,r}^{\left(  k_{1}\right)
            }\left(  \alpha\right)  b_{n,k_{1}}^{\left(  k_{2}\right)  }\left(
            \alpha\right)  b_{n,k_{2}}^{\left(  k_{3}\right)  }\left(  \alpha\right)
            \ldots b_{n,k_{\ell-1}}^{\left(  k_{\ell}\right)  }\left(  \alpha\right)
            b_{n,k_{\ell}}^{\left(  j\right)  }\left(  \alpha\right)  }{b_{n,k_{1}%
            }^{\left(  k_{1}\right)  }\left(  \alpha\right)  b_{n,k_{2}}^{\left(
            k_{2}\right)  }\left(  \alpha\right)  \ldots b_{n,k_{\ell}}^{\left(  k_{\ell
            }\right)  }\left(  \alpha\right)  }\right)  +\frac{\varphi_{n,i}^{\left(
            j\right)  }\left(  \alpha\right)  }{b_{n,j}^{\left(  j\right)  }\left(
        \alpha\right)  },\nonumber\\ i=&~1,2,\ldots,n,~j=1,2,\ldots,\left\lfloor \tfrac{n}{2}\right\rfloor
    ,\nonumber\\
    & \nonumber\\
    t_{i,n-j}^{n}=  &  ~\varphi_{n,i}\left(  \beta\right)  -\frac{1}%
    {b_{n,n-j}^{\left(  j\right)  }\left(  \beta\right)  }\cdot\left.  \sum
    _{r=1}^{j-1}\frac{\varphi_{n,i}^{\left(  r\right)  }\left(  \beta\right)
    }{b_{n,n-r}^{\left(  r\right)  }\left(  \beta\right)  }\right(  b_{n,n-r}%
    ^{\left(  j\right)  }\left(  \beta\right)  +\label{eq:last_half}\\
    &  ~\left.  +\sum_{\ell=1}^{j-r-1}\left(  -1\right)  ^{\ell}\sum
    _{r<k_{1}<k_{2}<\ldots<k_{\ell}<j}\frac{b_{n,n-r}^{\left(  k_{1}\right)
        }\left(  \beta\right)  b_{n,n-k_{1}}^{\left(  k_{2}\right)  }\left(
        \beta\right)  b_{n,n-k_{2}}^{\left(  k_{3}\right)  }\left(  \beta\right)
        \ldots b_{n,n-k_{\ell-1}}^{\left(  k_{\ell}\right)  }\left(  \beta\right)
        b_{n,n-k_{\ell}}^{\left(  j\right)  }\left(  \beta\right)  }{b_{n,n-k_{1}%
        }^{\left(  k_{1}\right)  }\left(  \beta\right)  b_{n,n-k_{2}}^{\left(
        k_{2}\right)  }\left(  \beta\right)  \ldots b_{n,n-k_{\ell}}^{\left(  k_{\ell
        }\right)  }\left(  \beta\right)  }\right) \nonumber\\
&  ~+\frac{\varphi_{n,i}^{\left(  j\right)  }\left(  \beta\right)  }%
{b_{n,n-j}^{\left(  j\right)  }\left(  \beta\right)  },\,~i=  ~1,2,\ldots,n,~j=1,2,\ldots,\left\lfloor \tfrac{n}{2}\right\rfloor
.\nonumber
\end{align}
\end{theorem}

Considering lookup tables that store the zeroth and higher order endpoint derivatives of the bases $\mathcal{F}_n^{\alpha,\beta}$ and  $\mathcal{B}_n^{\alpha,\beta}$, we have also investigated the computational complexity (i.e., the number of floating point operations or flops) required for the evaluation of all entries of the general transformation matrix that appears in (\ref{eq:basis_transformation}). In \citep[Theorem 2.2, p. 45]{Roth2015b} we have shown that the aforementioned complexity is exponential, but compared with other cubic time numerical algorithms (like function/curve interpolation or least squares approximation techniques based on $LU$ decomposition), the proposed general basis transformation can more efficiently be implemented up to $16$-dimensional EC spaces despite the seemingly complicated nature of formulas (\ref{eq:first_half})--(\ref{eq:last_half}).

In the next theorem we show that there is even a significantly better way for the evaluation of the matrix of the general basis transformation.

\begin{theorem}[Efficient general basis transformation]\label{thm:efficient_basis_transformation}
    Using the notations of Theorem \ref{thm:basis_transformation}, one has that the 
    non-trivial entries of the matrix $[t_{i,j}^{n}]_{i=0,\,j=0}^{n,\,n}$ of the general basis transformation (\ref{eq:basis_transformation}) can be determined by initializing the recursive formulas
    \begin{equation}
        \label{eq:efficient_first_half}
        t_{i,j}^{n} = \frac{1}{b_{n,j}^{\left(j\right)}\left(\alpha\right)}
        \left(\varphi_{n,i}^{\left(j\right)}\left(\alpha\right)-\sum_{k=0}^{j-1} t_{i,k}^{n} b_{n,k}^{\left(j\right)}\left(\alpha\right)\right),~j=1,\ldots,\left\lfloor\frac{n}{2}\right\rfloor,
    \end{equation}
    and
    \begin{equation}
        \label{eq:efficient_last_half}
        t_{i,n-j}^{n} = \frac{1}{b_{n,n-j}^{\left(j\right)}\left(\beta\right)}
        \left(\varphi_{n,i}^{\left(j\right)}\left(\beta\right)-\sum_{k=0}^{j-1} t_{i,n-k}^{n} b_{n,n-k}^{\left(j\right)}\left(\beta\right)\right),~j=1,\ldots,\left\lfloor\frac{n-1}{2}\right\rfloor,
    \end{equation}
    with the starting elements 
    $
    \{
    t_{i,0}^{n} = \varphi_{n,i}\left(\alpha\right)
    \}_{i=1}^{n}
    $
    and
    $
    \{
    t_{i,n}^{n} = \varphi_{n,i}\left(\beta\right)\}_{i=1}^{n},
    $
    respectively, for all $i=1,\ldots,n$. Moreover, if the endpoint derivatives $\{  \varphi_{n,i}^{\left(  j\right)  }\left(
    \alpha\right)  ,\varphi_{n,i}^{\left(  j\right)  }\left(  \beta\right)
    ,b_{n,i}^{\left(  j\right)  }\left(  \alpha\right)  ,b_{n,i}^{\left(
        j\right)  }\left(  \beta\right)  \}  _{i=1,~j=0}^{n,~\left\lfloor
        \frac{n}{2}\right\rfloor }$ are stored in advance in permanent lookup tables, then
    the number of flops required by the evaluation of formulas (\ref{eq:efficient_first_half}) and (\ref{eq:efficient_last_half}) is the polynomial cost
    \begin{equation}
        \label{eq:efficient_total_computational_cost}
        \kappa_{\mathrm{pol}}\left(n\right) = 
        \left\{
        \def\arraystretch{1.5}
        \begin{array}{ll}
            0,&n=0,1,\\
            n\cdot\left\lfloor\frac{n}{2}\right\rfloor\cdot\left(\left\lfloor\frac{n}{2}\right\rfloor+5\right),&n\geq 2,~n\equiv 0 \,\left(\mathrm{mod}\,2\right),
            \\
            n \cdot \left(\left\lfloor\frac{n}{2}\right\rfloor^2 + 4 \left\lfloor\frac{n}{2}\right\rfloor - 2\right),&n\geq 3,~n\equiv 1 \,\left(\mathrm{mod}\,2\right)
        \end{array}
        \def\arraystretch{1.0}
        \right.
    \end{equation} 
    which is always strictly less than the total cost
    \begin{equation}
        \label{eq:LU_cost}
        \kappa_{LU}\left(n,\delta\right)=
        \frac{2}{3}\left(n+1\right)^3-\frac{1}{2}\left(n+1\right)^2-\frac{1}{6}\left(n+1\right)+\left(2\left(n+1\right)^2-\left(n+1\right)\right)\delta
    \end{equation}
    of another numerical $\delta$-dimensional function interpolation or least squares approximation method based on $LU$ decomposition.
\end{theorem}

\begin{proof}
    The correctness of formulas (\ref{eq:efficient_first_half}) and (\ref{eq:efficient_last_half}) are immediate due to the proof of Theorem \ref{thm:basis_transformation} that can be found in \citep[pp. 52--54]{Roth2015b}, where we have used mathematical induction. Formulas (\ref{eq:efficient_first_half}) and (\ref{eq:efficient_last_half}) correspond in fact to induction steps based on forward and backward substitutions, the correctness of which were already proved. Final formulas (\ref{eq:first_half}) and (\ref{eq:last_half}) only give the closed expressions of the patterns that appear after performing all required forward of backward substitutions. Another simple way to verify formulas (\ref{eq:efficient_first_half}) and (\ref{eq:efficient_last_half}) is to differentiate the functional equalities
    \[
    \varphi_{n,i}\left(u\right) = \sum_{k=0}^{n} t_{i,k}^n b_{n,k}\left(u\right),~\forall u \in \left[\alpha,\beta\right],~i=1,\ldots,n
    \]
    for all orders $j=1,\ldots,\left\lfloor \frac{n}{2} \right\rfloor$ at the parameter values $u=\alpha$ and $u=\beta$, respectively, and to apply one of the corresponding endpoint conditions (\ref{eq:Hermite_conditions_0})--(\ref{eq:Hermite_conditions_alpha}). For example, at $u=\alpha$ one has that
    \[
    \varphi^{\left(j\right)}_{n,i}\left(\alpha\right)
    =
    \sum_{k=0}^{n} t_{i,k}^n b_{n,k}^{\left(j\right)}\left(\alpha\right)
    \overset{(\ref{eq:Hermite_conditions_0})}{=}
    \sum_{k=0}^{j} t_{i,k}^n b_{n,k}^{\left(j\right)}\left(\alpha\right)
    =
    \sum_{k=0}^{j-1} t_{i,k}^n b_{n,k}^{\left(j\right)}\left(\alpha\right)
    +
    t_{i,j}^nb_{n,j}^{\left(j\right)}\left(\alpha\right),
    \]
    where $b_{n,j}^{\left(j\right)}\left(\alpha\right)>0$. Therefore the entry $t_{i,j}^n$ can be obtained by subtraction and division. 
    
    Assuming that the endpoint derivatives $\{  \varphi_{n,i}^{\left(  j\right)  }\left(
    \alpha\right)  ,\varphi_{n,i}^{\left(  j\right)  }\left(  \beta\right)
    ,b_{n,i}^{\left(  j\right)  }\left(  \alpha\right)  ,b_{n,i}^{\left(
        j\right)  }\left(  \beta\right)  \}  _{i=1,~j=0}^{n,~\left\lfloor
        \frac{n}{2}\right\rfloor }$ are stored in advance in lookup tables, the polynomial computational cost (\ref{eq:efficient_total_computational_cost})
    follows from the simplification of the expression
    \[
    n \cdot 
    \left(
    \sum_{j=1}^{\left\lfloor\frac{n}{2}\right\rfloor}\left(j+2\right)
    +
    \sum_{j=1}^{\left\lfloor\frac{n-1}{2}\right\rfloor}\left(j+2\right)
    \right),
    \]
    where the first and second summations give the number of flops required by the evaluation of the formulas (\ref{eq:efficient_first_half}) and (\ref{eq:efficient_last_half}), respectively, while the leading scaling factor $n$ denotes the number of empty non-trivial rows that have to be calculated. At the same time, one can easily prove that  $\kappa_{\mathrm{pol}}\left(n\right) < \kappa_{LU}\left(n,\delta\right),~\forall n \geq 0, \delta \geq 1$ and
    \[
    \lim_{n \to \infty} \frac{\kappa_{\mathrm{pol}}\left(n\right)}{\kappa_{LU}\left(n,\delta\right)}=\frac{3}{8}.\qedhere
    \]
\end{proof}

Using \citep[Corollary 2.1, p. 43]{Roth2015b}, one can also provide ready to use control point configurations for the exact description of those traditional integral parametric curves and (hybrid) surfaces that are specified by coordinate functions given as (products of separable) linear combinations of ordinary basis functions. Namely, by means of general basis transformations, one can implement the control point determining formulas (\ref{eq:cpbed_ordinary_curves}) and (\ref{eq:cpbed_ordinary_surfaces}) of the next two theorems.

\begin{theorem}[Exact description of ordinary integral curves, \citep{Roth2015b}]\label{thm:integral_curves}
    Using B-curves of the type (\ref{eq:B_curve}), the ordinary integral curve
    \begin{equation}
        \mathbf{c}\left(u\right) = \sum_{i=0}^{n} \boldsymbol{\lambda}_i \varphi_{n,i}\left(u\right), ~u\in\left[\alpha,\beta\right],~0<\beta-\alpha<\ell_{n}^{\prime},~\boldsymbol{\lambda}_i \in \mathbb{R}^{\delta},~\delta \geq 2
        \label{eq:ordinary_integral_curve}
    \end{equation}
    fulfills the identity
    \[
    \mathbf{c}\left(u\right)\equiv \mathbf{c}_n\left(u\right) = \sum_{j=0}^n \mathbf{p}_j b_{n,j}\left(u\right),~\forall u \in \left[\alpha,\beta\right],
    \]
    where
    \begin{equation}
        \label{eq:cpbed_ordinary_curves}
        \left[
        \begin{array}{cccc}
            \mathbf{p}_0 & \mathbf{p}_1 & \cdots & \mathbf{p}_n
        \end{array}
        \right]
        =
        \left[
        \begin{array}{cccc}
            \boldsymbol{\lambda}_0 & \boldsymbol{\lambda}_1 & \cdots & \boldsymbol{\lambda}_n
        \end{array}
        \right]
        \cdot
        \left[
        t_{i,j}^n
        \right]_{i=0,\,j=0}^{n,\,n}.
    \end{equation}
\end{theorem}

\begin{theorem}
    [Exact description of ordinary integral surfaces \textnormal{-- extension of Theorem \ref{thm:integral_curves}}]\label{thm:integral_surfaces}%
    Let
    \[
    \mathcal{F}_{n_{r}}^{\alpha_{r},\beta_{r}}=\left\{  \varphi_{n_{r},i_{r}}\left(  u_{r}\right)  :u_{r}\in\left[
    \alpha_{r},\beta_{r}\right]  \right\}  _{i_{r}=0}^{n_{r}},~\varphi_{n_{r},0} \equiv 1,~0<\beta_r - \alpha_r < \ell^{\prime}\left(\mathbb{S}_{n_r}^{\alpha_r,\beta_r}
    \right)
    \]
    be the ordinary basis and
    $$
    \mathcal{B}_{n_{r}}^{\alpha_{r},\beta_{r}}=\left\{  b_{n_{r},j_{r}}\left(
    u_{r}\right)  :u_{r}\in\left[  \alpha_{r},\beta_{r}\right]  \right\}
    _{j_{r}=0}^{n_{r}}
    $$
    be the normalized B-basis of some EC vector space $\mathbb{S}%
    _{n_{r}}^{\alpha_{r},\beta_{r}}$ of functions and denote by $[  t_{i_{r},j_{r}}^{n_{r}}]  _{i_{r}=0,~j_{r}=0}%
    ^{n_{r},~n_{r}}$ the regular square matrix that transforms $\mathcal{B}%
    _{n_{r}}^{\alpha_{r},\beta_{r}}$ to $\mathcal{F}_{n_{r}}^{\alpha_{r},\beta
        _{r}}$, where $r=0,1$. Consider also the ordinary integral surface%
    \begin{equation}
        \mathbf{s}\left(  u_0, u_1\right)  =\left[
        \begin{array}
            [c]{ccc}%
            s^{0}\left(  u_0, u_1\right)   & s^{1}\left(  u_0, u_1\right)   &
            s^{2}\left(  u_0, u_1\right)
        \end{array}
        \right]  ^{T}\in%
        \mathbb{R}
        ^{3},~\left(u_0, u_1\right)\in\left[  \alpha_{0}%
        ,\beta_{0}\right]  \times\left[  \alpha_{1},\beta_{1}\right],
        \label{eq:ordinary_integral_surface}%
    \end{equation}
    where%
    \begin{equation*}
        s^{\ell}\left(  u_0, u_1\right)  =\sum_{\zeta=0}^{\sigma_{\ell}-1}\prod_{r=0}%
        ^{1}\left(  \sum_{i_{r}=0}^{n_{r}}\lambda_{n_r,i_{r}}^{\ell,\zeta}\varphi
        _{n_{r},i_{r}}\left(  u_{r}\right)  \right)  ,~\sigma_{\ell}\geq 1, ~\ell
        =0,1,2. 
    \end{equation*}
    Then, by using B-surfaces of the type (\ref{eq:B-surface}), the ordinary surface
    (\ref{eq:ordinary_integral_surface}) fulfills the identity
    \begin{equation*}
        \mathbf{s}\left(  u_0, u_1\right) \equiv \mathbf{s}_{n_0,n_1}\left(u_0, u_1\right) =\sum_{j_{0}=0}^{n_{0}}\sum
        _{j_{1}=0}^{n_{1}}\mathbf{p}_{j_{0},j_{1}}b_{n_{0},j_{0}}\left(  u_{0}\right)
        b_{n_{1},j_{1}}\left(  u_{1}\right)  ,~\forall\left(u_0,u_1\right)\in\left[  \alpha_{0},\beta_{0}\right]  \times\left[  \alpha
        _{1},\beta_{1}\right], %
    \end{equation*}
    where the control points $\mathbf{p}_{j_{0},j_{1}}=[
    p_{j_{0},j_{1}}^{\ell}]  _{\ell=0}^{2}\in%
    \mathbb{R}
    ^{3}$ are defined by the coordinates%
    \begin{equation}
        p_{j_{0},j_{1}}^{\ell}=\sum_{\zeta=0}^{\sigma_{\ell}-1}\prod_{r=0}^{1}\left(\sum_{i_{r}=0}^{n_{r}}\lambda_{n_r,i_{r}%
        }^{\ell,\zeta}t_{i_{r},j_{r}}^{n_{r}}\right),~\ell
        =0,1,2.\label{eq:cpbed_ordinary_surfaces}%
    \end{equation}
    
\end{theorem}

Using B-curves/surfaces of the type (\ref{eq:B_curve})/(\ref{eq:B-surface}) and applying formulas (\ref{eq:cpbed_ordinary_curves})/(\ref{eq:cpbed_ordinary_surfaces}), the proposed basis transformation can be used for the control point based exact description (or B-representation) of large families of integral or rational ordinary curves/surfaces that may be important in several areas of applied mathematics, since the investigated large class of EC vector spaces also comprise functions that appear in the traditional (or ordinary) parametric description of famous geometric objects like:
ellipses; epi- and hypocycloids; epi- and hypotrochoids; Lissajous curves; torus knots; foliums; rose curves; the witch of Agnesi; the cissoid of Diocles; Bernoulli's lemniscate;
Zhukovsky airfoil profiles; cycloids; hyperbolas; helices; catenaries; Archimedean and logarithmic spirals;
ellipsoids; tori; hyperboloids; catenoids; helicoids; ring, horn and spindle
Dupin cyclides; non-orientable surfaces such as Boy's and Steiner's surfaces
and the Klein Bottle of Gray. 

Figs.\ \ref{fig:order_elevated_snails}(\textit{a}) and \ref{fig:subdivided_snail}(\textit{b}) have already illustrated control point configurations for the B-representation of a single patch of the ordinary exponential-trigonometric integral surface (\ref{eq:snail}). In case of Fig.\ \ref{fig:exact_description_snail} we have described the entire surface (\ref{eq:snail}) with B-patches of the same order but with varying shape parameters.

\begin{figure}[!htb]
    \centering
    \includegraphics[]{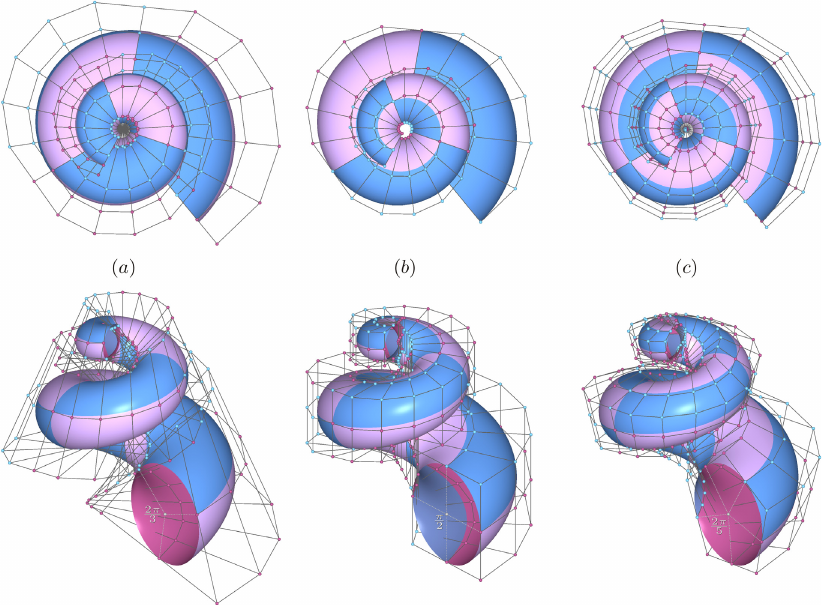}
    \caption{Different B-representations of the ordinary exponential-trigonometric integral surface (\ref{eq:snail}). Each patch is a B-surface of order $(n_0 = 6, n_1 = 2)$ that is described by means of the normalized B-bases of the EC spaces $\mathbb{ET}_{6}^{\alpha_0, \beta_0}$ and $\mathbb{T}_{2}^{\alpha_1, \beta_1}$ introduced in Example \ref{exmp:snail}, where the definition domain $\left[\alpha_0,\beta_0\right]\times\left[\alpha_1,\beta_1\right]$ corresponds to pairwise disjunct regions of $\left[\frac{7\pi}{2},\frac{49\pi}{8}\right]\times\left[-\frac{\pi}{3},\frac{5\pi}{3}\right]$. The lengths $\beta_0-\alpha_0 >0$ and $\beta_1 - \alpha_1>0$ of the varying definition domains can be considered as shape parameters. For example, in cases (\textit{a}), (\textit{b}) and (\textit{c}) the length $\beta_0-\alpha_0=\frac{29\pi}{40}$ is fixed, but the length $\beta_1-\alpha_1$ coincides with the values $\frac{2\pi}{3}$, $\frac{\pi}{2}$ and $\frac{2\pi}{5}$, respectively. (All images were rendered by means of the proposed function library.)}
    \label{fig:exact_description_snail}
\end{figure}


\section{Implementation details}\label{sec:implementation}

Our library assumes that the user has a multi-core CPU and also a GPU that is compatible at least with the desktop variant of OpenGL 3.0. In order to render the geometry, we use vertex buffer objects through the OpenGL Extension Wrangler (GLEW) library\footnote{M. Ikits, M. Magallon, and N. Stewart. 2017. GLEW: The OpenGL Extension Wrangler Library (release version 2.1.0). Retrieved July 31, 2017 from \url{http://glew.sourceforge.net/}} and for multi-threading we rely on a C++ compiler that supports at least OpenMP 2.0. Apart from GLEW no other external dependencies are used.

The entire implementation of the proposed function library is explained in the exhaustively commented listings and usage examples of Chapters \mref{2}{13--267} and \mref{3}{269--335} of the user manual \citep{Roth2018b} that is included in the supplementary material of the manuscript\footnote{Cross references of the forms \mref{{\relsize{+1}$x$}}{$y$} and \mref{{\relsize{+1}$x$}}{$y$--$z$} show that the referenced object {\relsize{+1}$x$} can be found either on the page $y$ or on pages $y$--$z$ of the user manual \citep{Roth2018b}.}.  The tree-view of the header files that can be included from our library is illustrated in Fig.\ \ref{fig:function_library_structure}. 

\begin{figure}[!htb]
    \centering
    \includegraphics[]{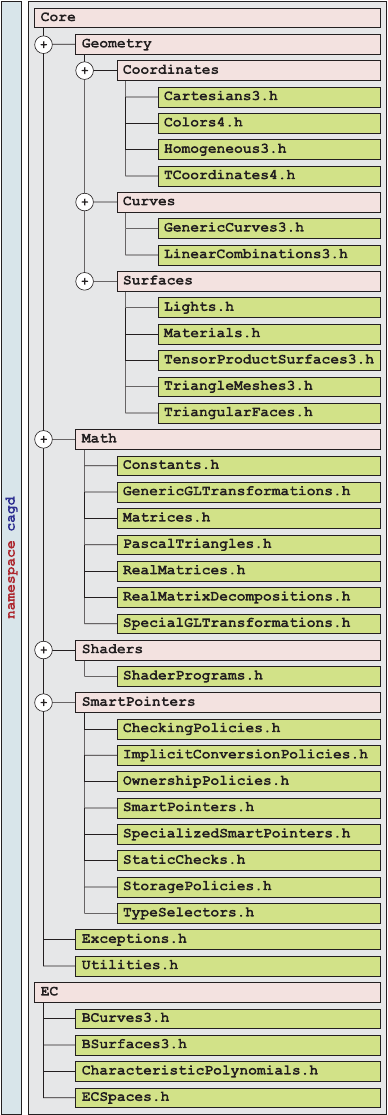}
    \caption{Tree-view of the header files of the proposed function library}
    \label{fig:function_library_structure}
\end{figure}

In its current state our library provides two main packages. The first of these is called \CBlue{\textbf{Core}} and consists of data types that represent:
\begin{itemize}
    \item
    exceptions (\CBlue{Exception});
    \item
    Cartesian (\CBlue{Cartesian3}), homogeneous (\CBlue{Homogeneous3}) and texture coordinates (\CBlue{TCoordinate4});
    \item
    color components (\CBlue{Color4}), different types of lights (\CBlue{DirectionalLight}, \CBlue{PointLight}, \CBlue{Spotlight}) and materials (\CBlue{Material});
    \item
    mathematical constants, generic rectangular (\CBlue{Matrix}$<$\CPurple{T}$>$, \CBlue{RowMatrix}$<$\CPurple{T}$>$, \CBlue{ColumnMatrix}$<$\CPurple{T}$>$) or triangular template matrices (\CBlue{TriangularMatrix}$<$\CPurple{T}$>$), real matrices (\CBlue{RealMatrix}: \CRed{public} \CBlue{Matrix}$<$\CRed{double}$>$), some real matrix decompositions (\CBlue{PLUDecomposition}, \CBlue{Facto\-rized\-Unpivoted\-LU\-De\-com\-po\-si\-tion}, \CBlue{SVDecomposition}), generic and derived OpenGL transformations (\CBlue{GLTransformation}, \CBlue{Translate}, \CBlue{Scale}, \CBlue{Rotate}, \CBlue{PerspectiveProjection}, \CBlue{OrthogonalProjection}, \CBlue{LookAt}) and Pascal triangles of binomial coefficients (\CBlue{PascalTriangle}: \CRed{public} \CBlue{TriangularMatrix}$<$\CRed{double}$>$);
    \item
    generic and specialized smart pointers (\CBlue{SmartPointer}$<$\CPurple{T},\CPurple{TSP},\CPurple{TOP},\CPurple{TICP},\CPurple{TCP}$>$, \CBlue{SP}$<$\CPurple{T}$>$::\CBlue{De\-faultPrimitive}, \CBlue{SP}$<$\CPurple{T}$>$::\CBlue{De\-fault}, \CBlue{SP}$<$\CPurple{T}$>$::\CBlue{Array}, \CBlue{SP}$<$\CPurple{T}$>$::\CBlue{DestructiveCopy}, \CBlue{SP}$<$\CPurple{T}$>$::\CBlue{NonIn\-tru\-siveReferenceCounting}) that provide different storage, ownership, implicit conversion and checking policies (\CBlue{StoragePolicy}$<$\CPurple{T}$>$::\CBlue{Default}, \CBlue{StoragePolicy}$<$\CPurple{T}$>$::\CBlue{Array}, \CBlue{OwnershipPolicy}$<$\CPurple{T}$>$, \CBlue{ImplicitConversionPolicy}, \CBlue{CheckingPolicy}$<$\CPurple{T}$>$::\CBlue{No\-Check},\break{} \CBlue{CheckingPolicy}$<$\CPurple{T}$>$::\CBlue{RejectNullDereferenceOr\-Indirection}, \CBlue{CheckingPol\-icy}$<$\CPurple{T}$>$::\CBlue{RejectNull},   \CBlue{Checking\-Pol\-icy}$<$\break{}\CPurple{T}$>$::\CBlue{AssertNullDereferenceOr\-Indirection}, \CBlue{CheckingPolicy}$<$\CPurple{T}$>$::\CBlue{AssertNull}) in order to avoid memory leaks and to ensure exception safety (one of the most frequently used smart pointers will be the specialized variant \CBlue{SP}$<$\CPurple{T}$>$::\CBlue{Default} that ensures default storage and deep copy policies, disallows implicit conversion and rejects null dereference or indirection);
    \item
    generic curves (\CBlue{GenericCurve3}) and abstract linear combinations (\CBlue{LinearCombination3});
    \item
    triangular faces (\CBlue{TriangularFace}), simple triangle meshes (\CBlue{TriangleMesh3}) and abstract tensor product surfaces (\CBlue{TensorProductSurface3});
    \item
    shader programs\footnote{For convenience we have also provided shader programs for simple (flat) color shading, for two-sided per pixel lighting that is able to handle user-defined directional, point and spotlights with uniform front and back materials, and another one for reflection lines that are combined with two-sided per pixel lighting. All figures of the current manuscript and of the user manual \citep{Roth2018b} were rendered by using these shader programs. 
    } (\CBlue{ShaderProgram}) written in the OpenGL Shading Language and used for rendering geometries (like control polygons and nets, or generic curves and triangle meshes obtained, e.g.,  as the images of linear combinations and tensor product surfaces, respectively).
\end{itemize}

The previously listed classes serve the definition, implementation and testing of the following data types that realize our main objectives and are included in the second main package called \CBlue{\textbf{EC}}:

\begin{itemize}
    \item
    the class \CBlue{CharacteristicPolynomial} ensures the factorization management and evaluation of characteristic polynomials of type (\ref{eq:characteristic_polynomial});
    
    \item
    EC spaces that comprise the constants and can be identified with the solution spaces of differential equations of type (\ref{eq:differential_equation}) will be instances of the class \CBlue{ECSpace};
    
    \item
    B-curves of type (\ref{eq:B_curve}) are represented by the class \CBlue{BCurve3} that is derived from the abstract base class \CBlue{LinearCombination3} and is based on the results of Corollary \ref{cor:B_basis_derivatives}, of  Lemma \ref{lem:general_order_elevation}, of Theorems \ref{thm:general_subdivision},  \ref{thm:efficient_basis_transformation} and \ref{thm:integral_curves};
    
    \item
    B-surfaces of type (\ref{eq:B-surface}) are represented by the class \CBlue{BSurface3} that is a descendant of the abstract base class \CBlue{TensorProductSurface3} and is based on Theorem \ref{thm:integral_surfaces} and on the natural extensions of Corollary \ref{cor:B_basis_derivatives}, of Lemma \ref{lem:general_order_elevation}, and of Theorem \ref{thm:general_subdivision}.
\end{itemize}

In what follows, we briefly detail the most important data types of the packages \CBlue{\textbf{Core}} and \CBlue{\textbf{EC}}. 

\subsection{Generic curves}
In order to store in vertex buffer objects the points and higher order derivatives of arbitrary smooth parametric (basis) functions, (B-)curves and isoparametric lines of (B-)surfaces, we introduce a class for generic curves (\CBlue{GenericCurve3}) that will be used for rendering purposes. Its diagram is illustrated in Fig.\ \ref{fig:UMLGenericCurve3}. Apart from vertex buffer object handling methods the class also provides  overloaded function operators that can be used for reading or writing the derivatives associated with a curve point. Moreover, the class also provides a method by means of which one can generate Matlab source codes to plot the curve points and to create scalable vector graphic formats (like EPS). The declaration and the full implementation of the class can be found in Listings \mref{2.37}{157--159} and \mref{2.38}{160--169} of \citep{Roth2018b}, respectively.

\begin{figure}[!h]
    
    \centering
    \includegraphics[scale = 0.990]{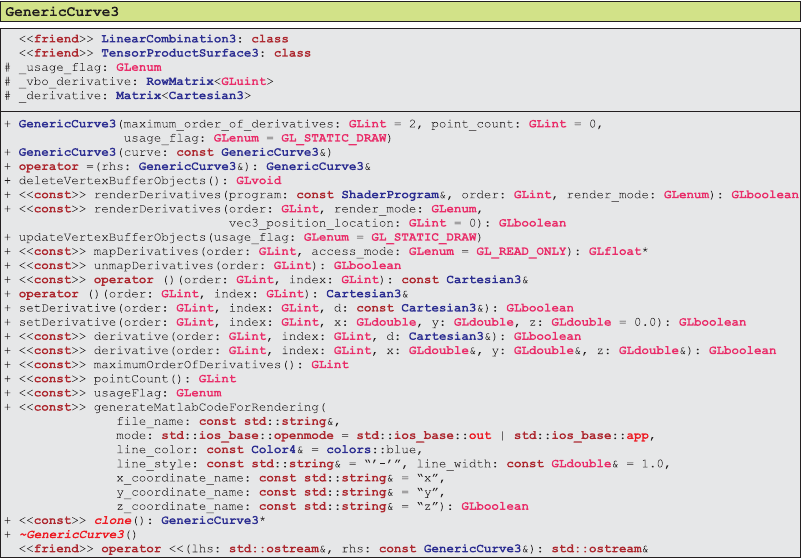}
    \caption{Class diagram of generic curves}
    \label{fig:UMLGenericCurve3}
\end{figure}

\subsection{Simple triangle meshes}

We also provide a class for simple triangle meshes (\CBlue{TriangleMesh3}), by means of which one can store in vertex buffer objects the attributes (i.e., position, normal and texture coordinates, color components and connectivity information) of vertices that form the triangular faces of the mesh. The class is also able to load triangulated object file formats and to either map or unmap vertex buffer objects associated with the aforementioned attributes. Its diagram is illustrated in Fig.\ \ref{fig:UMLTriangleMesh3}, while its definition and full implementation can be found in Listings \mref{2.40}{171--173} and \mref{2.41}{173--185} of \citep{Roth2018b}, respectively. 

\begin{figure}[!h]
    \centering
    \includegraphics[scale = 0.990]{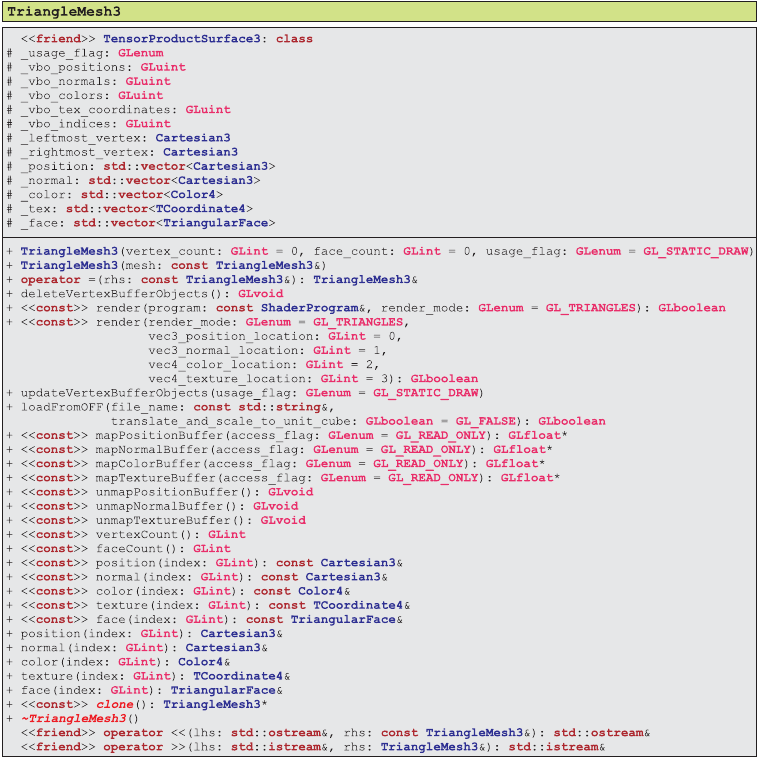}
    \caption{Class diagram of simple triangle meshes}
    \label{fig:UMLTriangleMesh3}
\end{figure}

\subsection{Abstract base classes for linear combinations and tensor product surfaces}
We also ensure abstract base classes for arbitrary linear combinations (\CBlue{LinearCombination3}) and tensor product surfaces (\CBlue{TensorProductSurface3}) 
that are able to generate their images, to update and render their control polygons or nets and to solve curve or surface interpolation problems -- provided that the user redeclares and defines in derived classes those pure virtual methods that appear in the interfaces of these abstract classes and are responsible for the evaluation of blending functions and of (partial) derivatives up to a specified maximum order of differentiation. The diagrams of these abstract classes are illustrated in Figs.\ \ref{fig:UMLLinearCombination3} and \ref{fig:UMLTensorProductSurfaces3}, while their definitions and full implementations can be found in Listing pairs \mref{2.42}{185--188}\,---\,\mref{2.43}{188--193} and \mref{2.44}{194--197}\,---\,\mref{2.45}{197--211} of \citep{Roth2018b}, respectively.

\begin{figure}[!h]
    \centering
    \vspace{0.2cm}
    \includegraphics[scale = 0.990]{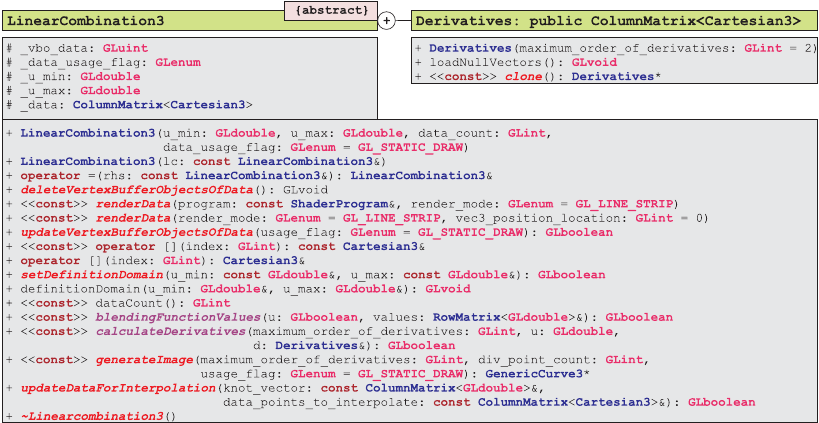}
    \caption{Class diagram of abstract linear combinations}
    \label{fig:UMLLinearCombination3}
\end{figure}

\begin{figure}[!h]
    \centering
    \includegraphics[scale = 0.990]{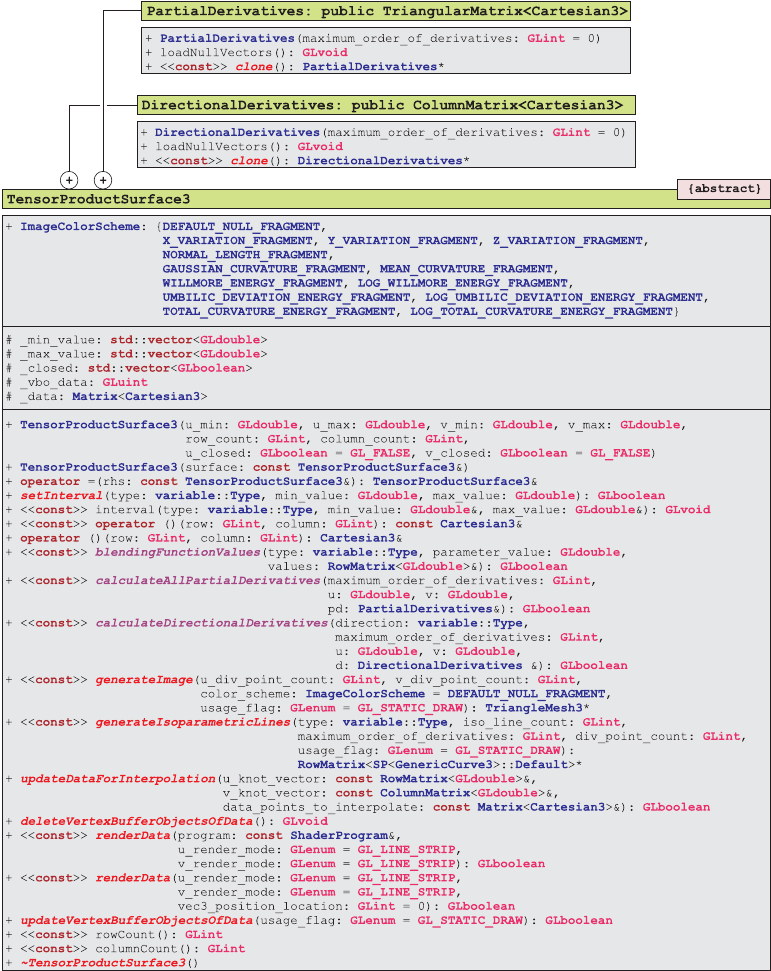}
    \caption{Class diagram of abstract tensor product surfaces}
    \label{fig:UMLTensorProductSurfaces3}
\end{figure}	

Pure virtual methods that have to redeclared and defined in derived classes are \CBlue{LinearCombination3}::\CDRPurple{\textit{blend\-ing\-Func\-tionValues}}, \CBlue{LinearCombination3}::\CDRPurple{\textit{calculateDerivatives}}, \CBlue{TensorProductSurface3}::\CDRPurple{\textit{blending\-Func\-tion\-Val\-ues}}, \CBlue{Tensor\-Product\-Surface3}::\CDRPurple{\textit{calculateAllPartialDerivatives}} and  \CBlue{Ten\-sor\-Prod\-uct\-Sur\-face3}::\CDRPurple{\textit{calculateDirectionalDerivatives}}. As we will see, B-curves and surfaces of type (\ref{eq:B_curve}) and (\ref{eq:B-surface}) will be derived from classes \CBlue{LinearCombination3} and \CBlue{TensorProductSurface3}, respectively. 

When generating the image of a tensor product surface, the user is also able to choose different color schemes that correspond to the point-wise variations of the $x$-, $y$- and $z$-coordinates, of the length of the normal vectors, of the Gaussian- and mean curvatures, of the Willmore energy and its translated logarithmic counterpart, of the umbilic deviation and its translated logarithmic scale, of the total curvature and its translated logarithmic variant, respectively. (In each case, the applied color map behaves like a temperature variation that ranges from the cold dark blue to the hot red, by passing through the colors cyan, green, yellow and orange such that the minimal and maximal values of a fixed energy type correspond to the extremal colors dark blue and red, respectively. For more details, see Fig.\ \mref{3.5}{320} of \citep{Roth2018b}.)

\subsection{Characteristic polynomials}
Characteristic polynomials of type (\ref{eq:characteristic_polynomial}) will be instances of the class \CBlue{CharacteristicPolynomial} that is able to store and update the factorization of (\ref{eq:characteristic_polynomial}) and also provides an overloaded function operator for evaluation purposes. The diagram of the class is illustrated in Fig.\ \ref{fig:UMLCharacteristicPolynomial}, while its definition and full implementation can be found in Listings \mref{2.46}{212--213} and \mref{2.47}{213--216} of \citep{Roth2018b}, respectively.

\begin{figure}[!h]
    
    \centering
    \includegraphics[scale = 0.990]{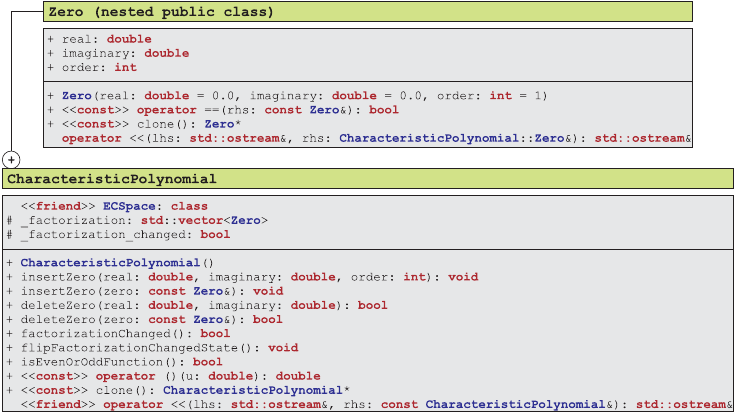}
    \caption{Class diagram of characteristic polynomials}
    \label{fig:UMLCharacteristicPolynomial}
\end{figure}

\subsection{EC spaces}
An EC space that comprises the constants, and which can be identified with the solution space of a differential equation of type (\ref{eq:differential_equation}), and whose space of derivatives is also EC, will be the instance of the class \CBlue{ECSpace} that
\begin{itemize}
    \item
    is able to generate and to update both the ordinary basis and the normalized B-basis of an EC vector space specified by the factorization of a characteristic polynomial of type (\ref{eq:characteristic_polynomial});
    \item
    provides a function operator to evaluate the zeroth and higher order derivatives of both bases at any point of the definition domain;
    \item
    can also be used to generate the general basis transformation matrix formulated in Theorem \ref{thm:basis_transformation} that maps the normalized B-basis to the ordinary basis of the underlying EC space;
    \item
    is able to decide whether the specified EC vector is reflection invariant; 
    \item
    can list the \LaTeX{} expressions of the ordinary basis functions; 
    \item
    can also be used to generate the images both of the ordinary basis and of the normalized B-basis functions.
\end{itemize}
The diagram of the class is illustrated in Fig.\ \ref{fig:UMLECSpace}, while its definition and full implementation can be found in Listings \mref{2.48}{216--220} and \mref{2.49}{220--241}, respectively.

\begin{figure}[!h]
    
    \centering
    \includegraphics[scale = 0.990]{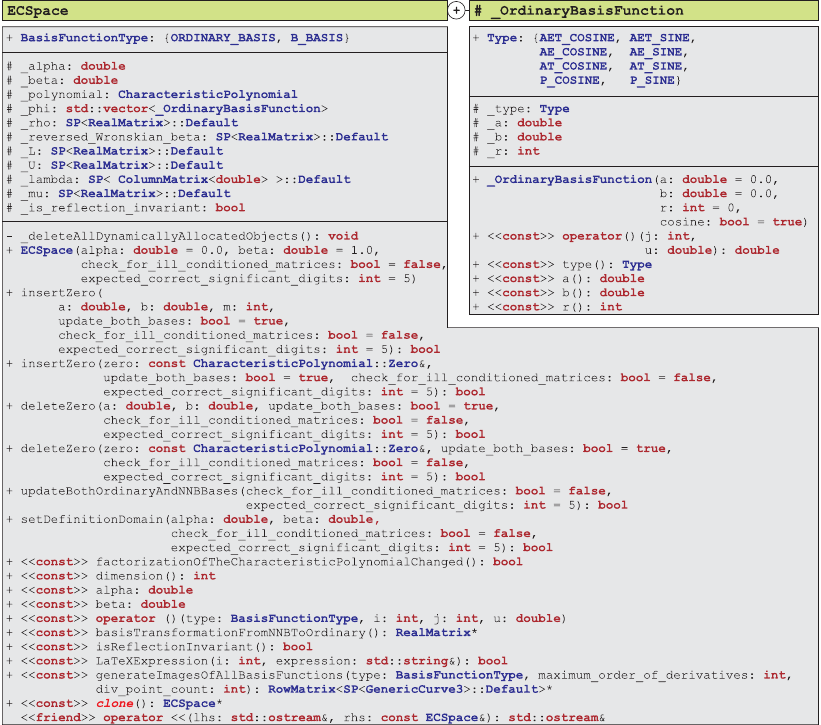}
    \caption{Class diagram of a translation invariant extended Chebyshev spaces that can be identified with the solution space of a constant-coefficient homogeneous linear differential equation}
    \label{fig:UMLECSpace}
\end{figure}

\begin{remark}[Full implementation details in the user manual]
    The construction process (\ref{eq:forward_Wronskian})--(\ref{eq:construction}) of the normalized B-basis functions of the underlying EC space and their differentiation formulas (\ref{eq:mixed_b_derivatives}) are implemented in lines \mref{566}{230}--\mref{808}{233} and \mref{833}{234}--\mref{910}{235} of Listing \mref{2.49}{220--241} in \citep{Roth2018b}, respectively. Formulas (\ref{eq:efficient_first_half})--(\ref{eq:efficient_last_half}) of the general basis transformation are implemented in lines \mref{914}{235}--\mref{981}{236} of Listing \mref{2.49}{220--241} in \citep{Roth2018b}.
\end{remark}

\begin{remark}[Examples in the user manual]
    Deriving from the base class \CBlue{ECSpace}, one can define special (like pure poly\-no\-mi\-al/trig\-o\-no\-met\-ric/hyperbolic and mixed exponential-trigonometric or algebraic-\{trig\-o\-no\-met\-ric/hy\-per\-bo\-lic/ex\-po\-nen\-tial-trigonometric\}) EC spaces as it is presented by several examples in Listings \mref{3.1}{270--273} and \mref{3.2}{273--278} of \citep{Roth2018b}. In Listings \mref{3.10}{291} and \mref{3.11}{291--294} of \citep{Roth2018b}, we also provided examples for the evaluation, differentiation and rendering of both the ordinary basis and the normalized B-basis of different types of EC spaces.
\end{remark}

\subsection{B-curves}
B-curves of type (\ref{eq:B_curve}) are represented by the class \CBlue{BCurve3} that is derived from the abstract base class \CBlue{LinearCombination3} and is based on the results of Corollary \ref{cor:B_basis_derivatives}, of Lemma \ref{lem:general_order_elevation}, and of Theorem \ref{thm:general_subdivision},  \ref{thm:efficient_basis_transformation} and \ref{thm:integral_curves}. It can be used to perform general order elevation, subdivision and to exactly describe arcs of user-specified ordinary integral curves by means of convex combinations of control points and normalized B-basis functions. The diagram of the class is illustrated in Fig.\ \ref{fig:UMLBCurve3}, while its definition and full implementation can be found in Listings \mref{2.50}{241--243} and \mref{2.51}{244--250} of \citep{Roth2018b}, respectively. Note that the class redeclares and defines those pure virtual methods that are inherited as interfaces from the abstract base class \CBlue{LinearCombination3}.

\begin{figure}[!h]
    
    \centering
    \includegraphics[scale = 0.990]{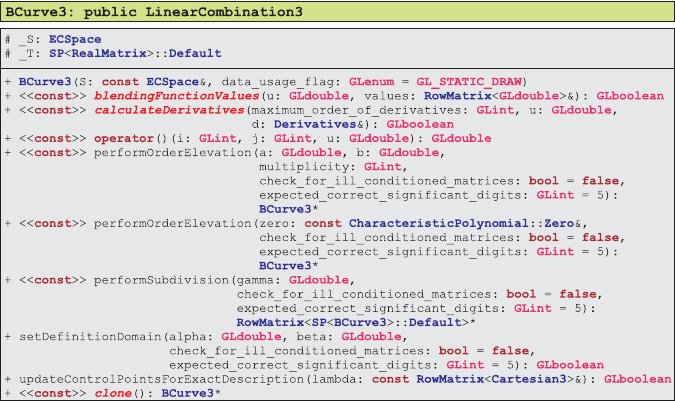}
    \caption{Class diagram of general B-curves}
    \label{fig:UMLBCurve3}
\end{figure}

\begin{remark}[Full implementation details in the user manual]
    Formulas (\ref{eq:order_elevation_inner_points_first_half})--(\ref{eq:order_elevation_inner_points_second_half}) of the general order elevation stated in Lemma \ref{lem:general_order_elevation} are implemented in lines \mref{137}{246}--\mref{160}{246} of Listing \mref{2.51}{244--250} in \citep{Roth2018b}. 
    Formulas (\ref{eq:left_subdivision_points_first_half})--(\ref{eq:right_subdivision_points_second_half}) of the general B-algorithm stated in Theorem \ref{thm:general_subdivision} are implemented in lines \mref{199}{247}--\mref{322}{249} of Listing \mref{2.51}{244--250} in \citep{Roth2018b}.
    Using B-curves of type (\ref{eq:B_curve}) and formula (\ref{eq:cpbed_ordinary_curves}) of Theorem \ref{thm:integral_curves}, the control point based exact description of ordinary integral curves of type (\ref{eq:ordinary_integral_curve}) is implemented in lines \mref{350}{249}--\mref{375}{250} of Listing \mref{2.51}{244--250} in \citep{Roth2018b}.
\end{remark}

\begin{remark}[Examples in the user manual]
    Listings \mref{3.12}{294--296}\,---\,\mref{3.13}{296--301} and Fig.\ \mref{3.2}{302} of \citep{Roth2018b} provide examples for the definition, generation, evaluation, differentiation, order elevation, subdivision and rendering of different types of B-curves. Listings \mref{3.14}{302--303}\,---\,\mref{3.15}{303--306} and Fig.\ \mref{3.3}{306} of \citep{Roth2018b} provide an example for the control point based exact description (or B-representation) of integral curves given in traditional (i.e., ordinary) parametric form.
\end{remark}

\subsection{B-surfaces}
B-surfaces of type (\ref{eq:B-surface}) are represented by the class \CBlue{BSurface3} that is a descendant of the abstract base class \CBlue{TensorProductSurface3} and is based on Theorem \ref{thm:integral_surfaces} and on the natural extensions of Corollary \ref{cor:B_basis_derivatives}, of Lemma \ref{lem:general_order_elevation} and of Theorem \ref{thm:general_subdivision}. It can be used to perform general order elevation, subdivision and to describe exactly a large family of ordinary integral surfaces of type (\ref{eq:ordinary_integral_surface}). Its diagram is presented in Fig.\ \ref{fig:UMLBSurface3}, while its definition and full implementation can be found in Listings \mref{2.52}{250--253} and \mref{2.53}{253--267} of \citep{Roth2018b}, respectively. Note that the class redeclares and defines those pure virtual methods that are inherited from the abstract base class \CBlue{TensorProductSurface3}.

\begin{figure}[!h]
    \centering
    \includegraphics[scale = 0.990]{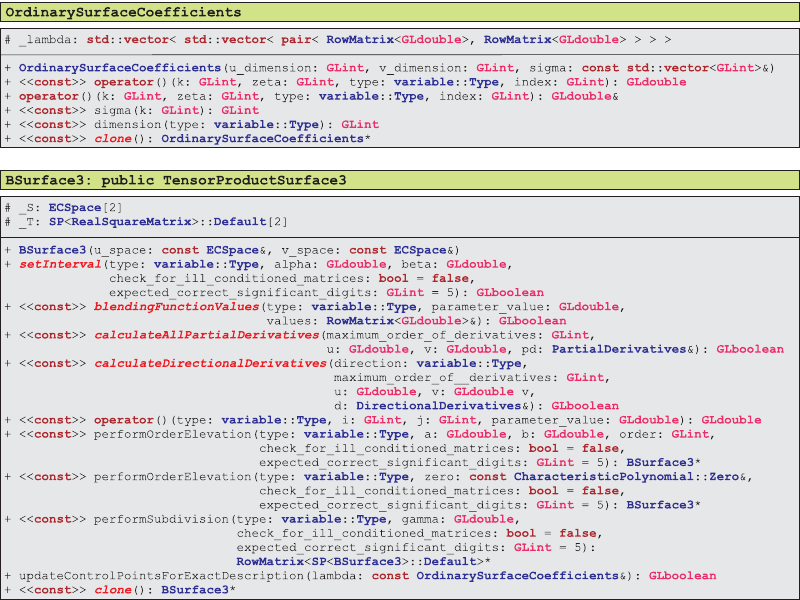}
    \caption{Class diagrams of general B-surfaces and of ordinary surface coefficients}
    \label{fig:UMLBSurface3}
\end{figure}

\begin{remark}[Full implementation details in the user manual]
    The results of Lemma \ref{lem:general_order_elevation} can also be extended to the general order elevation of B-surfaces of type (\ref{eq:B-surface}) and our function library ensures this possibility as well: the order elevation of B-surfaces is implemented in lines \mref{296}{258}--\mref{515}{261} of Listing \mref{2.53}{253--267} in \citep{Roth2018b}. 
    The subdivision technique presented in Theorem \ref{thm:general_subdivision} can also be extended to B-surfaces as it is implemented in lines \mref{521}{262}--\mref{781}{266} of Listing \mref{2.53}{253--267} in \citep{Roth2018b}.
    Using B-surfaces of type (\ref{eq:B-surface}) and formula (\ref{eq:cpbed_ordinary_surfaces}) of Theorem \ref{thm:integral_surfaces}, the control point based exact description (or B-representation) of ordinary integral surfaces of type (\ref{eq:ordinary_integral_surface}) is implemented in lines \mref{787}{266}--\mref{826}{267} of Listing \mref{2.53}{253--267} in \citep{Roth2018b}.
\end{remark}

\begin{remark}[Examples in the user manual]
    Listings \mref{3.16}{307--308}\,---\,\mref{3.17}{308--318} and Figs.\ \mref{3.4}{319}--\mref{3.7}{322} of \citep{Roth2018b} provide examples for the definition, generation, evaluation, differentiation, order elevation, subdivision and rendering of different types of B-surfaces. Listings \mref{3.18}{321--324}\,---\,\mref{3.19}{324--334} and Figs.\ 
    \mref{3.8}{334}--\mref{3.9}{335} of \citep{Roth2018b} give examples for the control point based exact description (or B-representation) of integral surfaces given in traditional (i.e., ordinary) parametric form and for the generation, evaluation, differentiation and rendering of isoparametric lines of B-surfaces.
\end{remark}

\section{Further examples, run-time statistics and handling possible numerical instabilities}\label{sec:examples_and_statistics}

We have seen that each of the proposed algorithms relies on the successful evaluation of zeroth and higher order (endpoint) derivatives of either of the ordinary basis functions (\ref{eq:ordinary_basis}) or of the normalized B-basis (\ref{eq:B-basis}). The order of (endpoint) derivatives that have to be evaluated increases proportionally with the dimension of the underlying EC space. Due to floating point arithmetical operations, the maximal dimension for which one does not bump into numerical instabilities depends both on the endpoints of the definition domain and on the type of the ordinary basis functions of the given EC space -- depending on the case, it may be smaller or greater, but considering that, in practice, curves and surfaces are mostly composed of smoothly joined lower order arcs and patches, we think that the proposed algorithms can be useful in case of real-life applications. In order to empirically underpin these statements, here we present several run-time statistics and we also describe how to detect and handle possible numerical instabilities. Using the Microsoft Visual Studio Compiler 15.0, we have tested the $64$-bit \textit{release} version of our library both on an affordable laptop with an Intel$^\text{\textregistered{}}$ Core\texttrademark{} i7-3720QM CPU$_1$ @ 2.60 GHz ($4$ cores, $8$ threads) and an nVidia GPU$_1$ GeForce GT 650M and on a desktop computer with an Intel$^\text{\textregistered{}}$ Xeon$^\text{\textregistered{}}$ E5-2670 CPU$_2$ @ 2.60 GHz ($8$ cores, $16$ threads) and an nVidia GPU$_2$ GeForce GTX 680.

Examples of the next subsections rely on the construction of the unique normalized B-bases of the following EC spaces. In case of each space we always assume that the length of its definition domain is strictly less than the corresponding critical length for design. The pure polynomial reflection invariant EC space\footnote{It is well-known \citep{Carnicer1993}, that the normalized B-basis of $\mathbb{P}_{n}^{\alpha,\beta}$ is formed by the Bernstein polynomials of degree $n$ defined over any non-empty compact interval $\left[\alpha,\beta\right]\subset \mathbb{R}$.}
\begin{equation}
    \mathbb{P}_{n}^{\alpha,\beta}:=\left\langle\mathcal{P}_n^{\alpha,\beta}\right\rangle:=\left\langle\left\{  1,u,\ldots,u^{n}:u\in\left[
    \alpha,\beta\right]  \right\}\right\rangle,~n \geq 0,~\dim \mathbb{P}_n^{\alpha,\beta} = n+1
    \label{eq:pure_polynomial}
\end{equation}
corresponds to the characteristic polynomial $p_{n+1}\left(z\right) 
=z^{n+1}$, $z\in\mathbb{C}$. Let $\left\{  \omega_{k}\right\}_{k=1}^{n}$ be pairwise distinct non-zero real numbers. Then the
characteristic polynomial 
\[
p_{\left(  n+1\right)  ^{2}}\left(  z\right)   
=z^{n+1}\prod_{k=1}%
^{n}\left(  z^{2}+\omega_{k}^{2}\right)  ^{n+1-k},~z\in\mathbb{C}
\]
generates the reflection invariant mixed algebraic-trigonometric EC space
\begin{align}
    \mathbb{AT}_{n\left(n+2\right) }^{\alpha,\beta}  := &~ \left\langle \mathcal{P}_n^{\alpha,\beta} \cup\left\{  u^{r}\cos\left(
    \omega_{k}u\right)  ,u^{r}\sin\left(  \omega_{k}u\right)  :u\in\left[
    \alpha,\beta\right]  \right\}  _{k=1,~r=0}^{n,~n-k}\right\rangle,~n \geq 2
    \label{eq:algebraic_trigonometric}
\end{align}
of dimension $\dim \mathbb{AT}_{n\left(n+2\right)}^{\alpha,\beta} = \left(n+1\right)^2$. 
The reflection invariant EC vector spaces
\begin{equation}
    \mathbb{T}_{2n}^{\alpha,\beta}:=
    \left\langle\mathcal{T}_{2n}^{\alpha,\beta}\right\rangle :=
    \left\langle\left\{1,\left\{\cos\left(ku\right),\,\sin\left(ku\right)\right\}_{k=1}^{n}:u\in\left[\alpha,\beta\right]\right\}\right\rangle,~n \geq 1,
    ~
    \dim \mathbb{T}_{2n}^{\alpha,\beta} = 2n+1
    \label{eq:pure_trigonometric}
\end{equation}
and
\begin{equation}
    \mathbb{H}_{2n}^{\alpha,\beta}:=
    \left\langle\mathcal{H}_{2n}^{\alpha,\beta}\right\rangle :=
    \left\langle\left\{1,\left\{\cosh\left(ku\right),\,\sinh\left(ku\right)\right\}_{k=1}^{n}:u\in\left[\alpha,\beta\right]\right\}\right\rangle,~n \geq 1,
    ~
    \dim \mathbb{H}_{2n}^{\alpha,\beta} = 2n+1
    \label{eq:pure_hyperbolic}
\end{equation}
of pure trigonometric\footnote{The unique normalized B-basis of $\mathbb{T}_{2n}^{\alpha,\beta}$ exists whenever $\beta-\alpha \in \left(0,\pi\right)$ and it was constructed in closed form in \citep{Sanchez1998}.} and hyperbolic\footnote{The explicit form of the unique normalized B-basis of the EC space $\mathbb{H}_{2n}^{\alpha,\beta}$ was derived in \citep{ShenWang2005} and, theoretically, the interval length $\beta-\alpha$ can be any positive number. Concerning numerical instabilities, the only limitation lies in the usage of potentially too big exponentials.} polynomials of order at most $n$ (or of degree at most $2n$) would correspond to the characteristic polynomials
\[
p_{2n+1}\left(z\right) = z \prod_{k = 1}^n \left(z^2 + k^2\right)
\text{ and }~
p_{2n+1}\left(z\right) = z \prod_{k = 1}^n \left(z^2 - k^2\right),
\]
respectively, where $z\in\mathbb{C}$. We will also consider the algebraic-exponential-trigonometric EC spaces\footnote{For arbitrary values of the order $n \geq 2$, the critical lengths $\ell^{\prime}\left(\mathbb{AT}_{n\left(n+2\right)}^{\alpha,\beta}\right)$ and $ \ell^{\prime}\left(\mathbb{AET}_{n\left(2n+3\right)}^{\alpha,\beta}\right)$ for design and the explicit forms of the unique normalized B-bases of the spaces $\mathbb{AT}_{n\left(n+2\right)}^{\alpha,\beta}$ and  $\mathbb{AET}_{n\left(2n+3\right)}^{\alpha,\beta}$, respectively, were not studied in the literature. However, normalized B-bases of some special subspaces of these parent spaces were investigated, e.g.,\ in \citep{MainarPena2004,CarnicerMainarPena2004,CarnicerMainarPena2007, MainarPena2010, CarnicerMainarPena2014} and references therein.}$^{,}$\footnote{The mixed hyperbolic-trigonometric EC space $\mathbb{M}_{n+4,a,b}^{\alpha,\beta}$ was investigated in \citep{BrilleaudMazure2012}, where, for $n=0$, it was shown that $\ell^{\prime}\left(\mathbb{M}_{4,a,b}^{\alpha,\beta}\right)$ coincides with the only solution of the equation $b\tanh\left(au\right)=a \tan\left(bu\right)$, where $u \in \left(\frac{\pi}{b},\frac{3\pi}{2b}\right)$, i.e., in general, the critical length for design is not necessarily a free parameter with respect to the parameters resulting from the differential equation (\ref{eq:differential_equation}). For example, in case of parameters $n=0$, $a=1$ and $b = 0.2$ one obtains that $\ell^{\prime}\left(\mathbb{M}_{4,1,0.2}^{\alpha,\beta}\right)\approx 16.694941067922716$, i.e., the space $\mathbb{M}_{4,1,0.2}^{\alpha,\beta}$ possesses a unique normalized B-basis provided that $\beta - \alpha \in \left(0,16.694941067922716\right)$. Moreover, $\ell^{\prime}\left(\mathbb{M}_{n+4,a,b}^{\alpha,\beta}\right)\geq \ell^{\prime}\left(\mathbb{M}_{4,a,b}^{\alpha,\beta}\right),~\forall n\geq 1$, $\forall a,b>0$.}
\begin{align}
\mathbb{AET}_{n\left(2n+3\right)}^{\alpha,\beta}
:=&~\left\langle \mathcal{AET}_{n\left(2n+3\right)}^{\alpha,\beta} \right\rangle :=
\left\langle
\mathcal{P}_{n}^{\alpha,\beta} \cup
\left\{
\left\{
u^r e^{\omega_k u}\cos\left(\omega_k u\right),
u^r e^{\omega_k u}\sin\left(\omega_k u\right)
\right\}_{r=0,\,k=1}^{n-k,\,n},
\right.
\right.
\label{eq:algebraic-exponential-trigonometric}
\\
&
\left.
\left.
\left\{
u^re^{-\omega_k u}\cos\left(\omega_k u\right),
u^re^{-\omega_k u}\sin\left(\omega_k u\right)
\right\}_{r=0,\,k=1}^{n-k,\,n}:
u \in \left[\alpha,\beta\right]
\right\}
\right\rangle, ~n \geq 2,
\nonumber
\\
&
~\dim \mathbb{AET}_{n\left(2n+3\right)}^{\alpha,\beta} = 2n^2+3n+1,
\nonumber
\end{align}
and
\begin{align}
\mathbb{M}_{n+4,a,b}^{\alpha,\beta}:=&~
\left\langle \mathcal{M}_{n+4,a,b}^{\alpha,\beta}\right\rangle
:=
\left\langle
\mathcal{P}_{n}^{\alpha,\beta}
\cup
\left\{
\cosh\left(au\right)\cos\left(bu\right),\cosh\left(au\right)\sin\left(bu\right),\right.\right.
\label{eq:Mazure_vector_space}
\\
&
\left.\left.\sinh\left(au\right)\cos\left(bu\right),\sinh\left(au\right)\sin\left(bu\right)
:
u \in \left[\alpha,\beta\right]
\right\}
\right\rangle,
\,
n \geq 0,
\,
\dim\mathbb{M}_{n+4,a,b}^{\alpha,\beta}=n+5,
\nonumber{}
\end{align}
which are reflection invariant and correspond to the characteristic polynomials
\[
p_{2n^2+3n+1}\left(z\right)
=
z^{n+1}
\prod_{k=1}^n
\left(z^2 - 2\omega_k z + 2\omega_k^2\right)^{n+1-k}
\prod_{k=1}^n
\left(z^2 + 2\omega_k z + 2\omega_k^2\right)^{n+1-k}
\]
and
\[
p_{n+5}\left(z\right) = z^{n+1}\left(z-\left(a-\mathbf{i}b\right)\right)\left(z-\left(a+\mathbf{i}b\right)\right)\left(z-\left(-a-\mathbf{i}b\right)\right)\left(z-\left(-a+\mathbf{i}b\right)\right),
\]
respectively, where $a,b>0$ and $z\in\mathbb{C}$.

\subsection{Determining the length of the definition domain}

The length $\beta-\alpha>0$ of the definition domain $\left[\alpha,\beta\right]$ should be strictly less than the critical length of the space $D\mathbb{S}_n^{\alpha,\beta}$ obtained after differentiation (i.e., the critical length for design), otherwise the given space may not provide shape preserving representations, e.g.,  the generated ``B-basis" functions may not form a strictly totally positive function system that usually leads to the violation of the convex hull and variation diminishing properties of the induced ``B-curves". This property is related to the fact that the space of derivatives fails to be EC for too large intervals \citep{CarnicerMainarPena2004} (leading to the idea of the critical length of a space of functions that is invariant under translations). We have already seen in the previous subsection that there are some types of translation invariant EC spaces in case of which we know the explicit values of the critical lengths of the corresponding spaces of the derivatives.
However, in general, the exact determination or at least the approximation of the critical length for design of an EC space -- as the supremum of the lengths of the intervals on which the space of derivatives is also EC -- is not a trivial problem. Therefore, when one creates an instance of the class \CBlue{ECSpace} that has not been previously studied in the literature, we advise to always check whether each generated ``normalized B-basis function" is indeed non-negative over the user-defined definition domain\footnote{The bicanonical basis (\ref{eq:particular_integrals}) is positive on $\left(\alpha,\beta\right)$ due to boundary conditions  (\ref{eq:boundary_conditions}), but this theoretically expected property may be violated at run-time due to the accumulated numerical errors resulting from poorly selected input parameters.} or, even if it is non-negative, behaves as theoretically expected. If there are subregions on which at least a function becomes negative or has an unexpected (even chaotic) behavior, the length of the definition domain should be decreased and the verifying test should be repeated in an interactive manner. (Given a user-defined instance $\mathbb{S}_n^{\alpha,\beta}$ of the class \CBlue{ECSpace}, the proposed function library is able to generate and render the shape of the normalized B-basis functions by using the methods \CBlue{ECSpace}::generateImagesOfAllBasisFunctions,  \CBlue{GenericCurve3}::updateVertexBufferObjects and \CBlue{GenericCurve3}::renderDerivatives. For more details consider, e.g.,  Listings \mref{3.1}{270--273}\,---\,\mref{3.2}{273--278} and \mref{3.10}{291}\,---\,\mref{3.11}{291--294} of \citep{Roth2018b}. This means that users have a graphical feedback in order to decide whether the length $\beta-\alpha$ of the underlying definition domain exceeded or it is very close from below to the usually unknown critical length $\ell^{\prime}_n$ for design. 
)

Since the proposed function library is designed for constant-comprising translation invariant solution spaces of constant-coefficient homogeneous linear differential equations, by means of \citep[Proposition 3.1/(iii), Theorems 2.4 \& 4.1, and Corollary 4.1]{CarnicerMainarPena2004} there always exist sufficiently small intervals for which the considered vector spaces are EC and also possess unique normalized B-bases.

Numerical instabilities may also appear when the length of the definition domain of the underlying EC space is too small, since this may lead both to higher order endpoint derivatives with too big absolute values and to almost singular or badly scaled systems of linear equations. However, note that, when $\beta-\alpha \searrow 0$ the B-curve (\ref{eq:B_curve}) always approaches the polynomial B\'ezier curve of degree $n$ (see \citep[Theorem 3.13]{Pottmann1993}), i.e., in this case there is no real interest in replacing the standard polynomial B\'ezier curve by the limiting case $\beta-\alpha\searrow 0$ of a non-polynomial B-curve. The strongest shape effects are obtained when $\beta-\alpha \nearrow \ell^{\prime}_n$.

As one can see, the length of the definition domain influences both the correctness and the numerical stability of all proposed algorithms. The user should avoid the usage of either too large or too small definition domains. This does not mean that the proposed function library cannot be used in the limiting cases $\beta-\alpha \searrow 0$ and $\beta-\alpha \nearrow \ell^{\prime}_n$. Nevertheless, considering programming and stable numerical evaluations, it is important to avoid pathological cases like $\beta \in \left(\alpha, \alpha + \varepsilon\right)$ and $\beta \in \left(\alpha + \ell^{\prime}_n - \varepsilon, \alpha + \ell^{\prime}_n\right)$, where $\varepsilon > 0$ is a too small user-defined threshold parameter. 

In order to illustrate the correctness of the proposed function library, in Figs.\ \ref{fig:shape_effect_in_pure_trigonometric_hyperbolic_cases} and \ref{fig:shape_effect_in_case_of_hyperbolic_trigonometric_M_spaces} we have reconstructed some already known examples in case of which we will consider shape effects obtained under these limiting cases.

\begin{figure}[!h]
    \centering
    \includegraphics[]{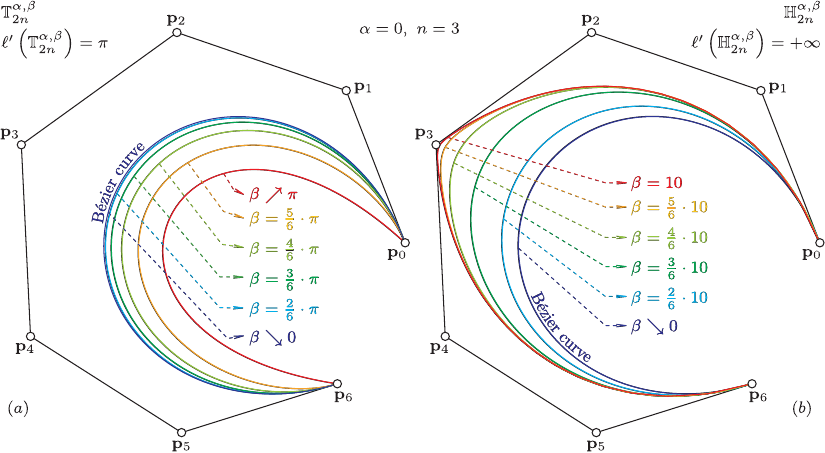}
    \caption{(\textit{a}) Shape variations of pure (\textit{a}) trigonometric and (\textit{b}) hyperbolic B-curves of order 3 (degree 6), where $\alpha = 0$ is fixed but $\beta$ varies in the ranges $\left[\varepsilon_1,\pi - \varepsilon_2\right)$ and $\left(\varepsilon_1,10\right]$, respectively, where $\varepsilon_1 = 0.024$ and $\varepsilon_2 = 10^{-15}$ (in case (\textit{b}), numerical errors appear if $\beta \gtrapprox 30.5$). In both cases, visually there is no difference between the B-curve obtained for $\beta = \varepsilon_1$ and the B\'ezier curve that is expected in the limiting case $\beta\searrow 0$.}
    \label{fig:shape_effect_in_pure_trigonometric_hyperbolic_cases}
\end{figure}

\begin{figure}[!h]
    \centering
    \includegraphics[]{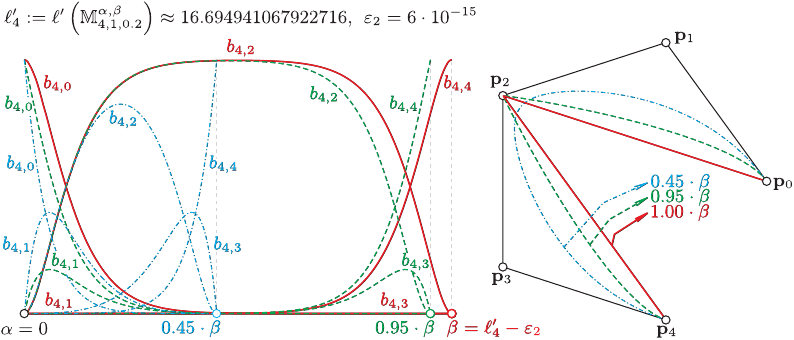}
    \caption{On the left side one can follow the shape variation of the normalized B-basis functions of the vector space (\ref{eq:Mazure_vector_space}), where parameters $a=1$, $b=0.2$ and $\alpha = 0$ are fixed, while the endpoint $\beta$ varies in the set $\left\{0.45 \cdot \left(\ell^{\prime}_{4} - \varepsilon_2\right),~ 0.95 \cdot \left(\ell^{\prime}_{4} - \varepsilon_2\right),~\ell^{\prime}_{4} - \varepsilon_2\right\}$, where $\varepsilon_2 = 6\cdot 10^{-15}$ and the critical length $\ell_4^{\prime}=\ell^{\prime}\left(\mathbb{M}^{\alpha,\beta}_{4,1,0.2}\right)\approx 16.694941067922716$ for design was specified in \citep{BrilleaudMazure2012}. For this $5$-dimensional case, the minimal value of $\beta$ for which does not appear numerical instabilities is $\varepsilon_1 \approx \frac{\ell_4^{\prime}}{10^3}$. (Observe that basis functions $b_{4,1}$ and $b_{4,3}$ tend to the constant function $0$ as $\beta \nearrow \ell^{\prime}_4$.) On the right side one can observe the effect of the varying endpoint $\beta$ on the shape of a hyperbolic-trigonometric B-curve of order $4$ determined by a fixed control polygon.}
    
    \label{fig:shape_effect_in_case_of_hyperbolic_trigonometric_M_spaces}
\end{figure}

\subsection{Determining the maximum dimension}

The construction process ((\ref{eq:forward_Wronskian})--(\ref{eq:boundary_conditions}), (\ref{eq:LU_factorization_of_reversed_system})--(\ref{eq:construction})) of the normalized B-basis of an EC space is based on the solution of several systems of linear equations that may be ill-conditioned either for relatively large dimension numbers or for poorly selected endpoints of the definition domain. For this reason several methods of the proposed function library expect a boolean flag named ``check\_for\_ill\_con\-di\-tioned\_ma\-tri\-ces" and a non-negative integer called ``expected\_correct\_sig\-nif\-i\-cant\_digits". If the flag ``check\_\-for\_\-ill\_\-con\-di\-tioned\_\-matrices" is set to true, these methods will calculate the condition number of each matrix that appears in the construction process  ((\ref{eq:forward_Wronskian})--(\ref{eq:boundary_conditions}), (\ref{eq:LU_factorization_of_reversed_system})--(\ref{eq:construction})). Using singular value decomposition \citep{PressTeukolskyVetterlingFlannery2007}, each condition number is determined as the ratio of the largest and smallest singular values of the corresponding matrices. If at least one of the obtained condition numbers is too large (i.e., when the number of estimated correct significant digits is less than the number of expected ones), these methods will throw an exception that one of the systems of linear equations is ill-conditioned and therefore its solution may be inaccurate. If the user catches such an exception, one can try:
\begin{itemize} 
    \item
    to lower the number of expected correct significant digits;
    \item
    to decrease the dimension of the underlying EC space;
    \item
    to change the endpoints of the definition domain $\left[\alpha,\beta\right]$;
    \item
    to run the code without testing for ill-conditioned matrices and hope for the best.
\end{itemize}

Note that in certain cases the standard condition number may lead to an overly pessimistic estimate for the overall error and at the same time, by activating the boolean flag ``check\_for\_ill\_con\-di\-tioned\_matrices", the run-time of the aforementioned methods will increase. Several numerical tests show that  ill-conditioned matrices appear during the construction process ((\ref{eq:forward_Wronskian})--(\ref{eq:boundary_conditions}), (\ref{eq:LU_factorization_of_reversed_system})--(\ref{eq:construction})) when one tries to define EC spaces with relatively big dimensions. Considering that, in practice, (spline) curves and surfaces are mostly described by basis functions of lower dimensional vector spaces, by default we opted for speed, i.e., initially the flag ``check\_for\_ill\_conditioned\_ma\-trices" is set to false. If one obtains mathematically or geometrically unexpected results (that violate either the convex hull or variation diminishing properties, or they are simply meaninglessly noisy and chaotic), then one should (also) study the values of the condition numbers mentioned above.

The maximal dimension for which one does not bump into numerical instabilities also depends on possible operations that have to performed on the obtained B-curves/surfaces. If one only intends to evaluate, (partially) differentiate and render B-curves/surfaces without performing order elevations, subdivisions or B-representations on them, the dimensions of the applied EC spaces can be bigger than otherwise. The reason of this is that differentiations and arithmetical floating point operations that appear in formulas of Lemma \ref{lem:general_order_elevation} and of Theorems \ref{thm:general_subdivision} and \ref{thm:efficient_basis_transformation} further increase the accumulated numerical errors that causes the quicker disappearance of the correct digits.

Assuming that the user tries to model B-curves/surfaces without performing order elevations, subdivision and B-representations on them, Fig.\ \ref{fig:SpaceCreationRuntimeStatistics} illustrates in horizontal direction the maximal values of $n$ for which one does not enter into numerical instabilities in spite of the fact that these values are greater than those for which the estimated number of correct digits would equal to only $1$ based on the detected corresponding maximal condition numbers and on the machine epsilon $\varepsilon \approx 2.220446 \cdot 10^{-16}$ of double precision types. The latter values of $n$ are highlighted with the tipping points of the tick vertical dashed lines with arrows pointing to their left and right sides. 

\begin{figure}[!h]
    
    \centering
    \includegraphics[]{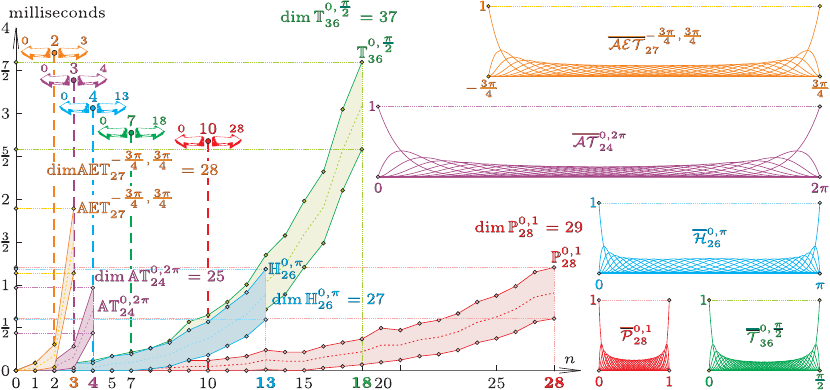}
    \caption{Maximal dimensions for which one does not enter into numerical instabilities in case of B-curve/surface generation and differentiation in EC spaces (\ref{eq:pure_polynomial})--(\ref{eq:algebraic-exponential-trigonometric}). The lower and upper boundary of the shaded regions correspond to the endpoints of confidence intervals that include the unknown theoretical mean value of the time (measured in milliseconds on CPU$_1$) that is required to construct all data needed for the evaluation and differentiation of both the ordinary basis and the normalized B-basis of the given EC spaces for each value of $n$. (The dotted centerlines of the shaded areas interpolate the sample mean of the elapsed time values.) The tipping points of the thick vertical dashed lines (with arrows pointing to their left and right sides) highlight those values of $n$ for which the estimated number of correct digits would equal to only $1$ based on the detected corresponding maximal condition numbers.}
    \label{fig:SpaceCreationRuntimeStatistics}
\end{figure}

Fig.\ \ref{fig:SpaceCreationRuntimeStatistics} also shows in vertical direction confidence intervals for the unknown theoretical mean value of the time (measured in milliseconds) that is required to calculate all information needed for the evaluation and differentiation of both the ordinary basis and the normalized B-basis of the given EC spaces for each value of $n$. The run-time related confidence intervals were determined as follows.  Consider the fixed significance level $s=0.01$  and let $N\geq 1000$ be a fixed number of independent trials for each values of $n$. Assuming that $\{\tau_{n,i}\}_{i=1}^{N}$ is the elapsed time sample obtained by repeatedly testing an algorithm for a fixed value of $n$ and denoting by $F_{\mathcal{S}(N-1)}$ Student's $T$-distribution function of $N-1$ degrees of freedom, the endpoints of the confidence intervals were computed as
\[
\left(\tau_{n,\min}, \tau_{n,\max}\right) = 
\left(
\max\left\{
\overline{\tau}_n - \frac{\overline{\sigma}_n}{\sqrt{N}}\cdot x_{1-\frac{s}{2},N-1},0\right\},
\overline{\tau}_n + \frac{\overline{\sigma}_n}{\sqrt{N}} \cdot x_{1-\frac{s}{2},N-1}
\right),
\]
where
\[
\overline{\tau}_n= \frac{1}{N}\sum_{i=1}^{N} \tau_{n,i},~
\overline{\sigma}_n =
\sqrt{\frac{1}{N-1}\sum_{i=1}^N\left(\tau_{n,i} - \overline{\tau}_n\right)^2},~
x_{1-\frac{s}{2},N-1} = F_{\mathcal{S}(N-1)}^{-1}\left(1-\frac{s}{2}\right).
\]

If one also intends to perform special operations (like subdivision or order elevation) on B-curves/surfaces constructed in EC spaces (\ref{eq:pure_polynomial})--(\ref{eq:algebraic-exponential-trigonometric}), usually one should work with numbers $n$ that are smaller than or equal to those values that are highlighted by the tipping points of the thick vertical dashed lines of Fig.\ \ref{fig:SpaceCreationRuntimeStatistics}. (In order to avoid badly scaled or close to singular matrices that may appear in the B-basis construction process ((\ref{eq:forward_Wronskian})--(\ref{eq:boundary_conditions}), (\ref{eq:LU_factorization_of_reversed_system})--(\ref{eq:construction})) during repeated subdivisions, usually the dimension of the underlying EC space should be decreased proportionally to the shrinking length of the definition domain.)

Table \ref{tbl:differentiation_up_to_a_maximal_order}
provides run-time statistics for the differentiation of those normalized B-basis functions that are illustrated in Fig.\ \ref{fig:SpaceCreationRuntimeStatistics}. (Note that, in practice, is highly unlikely that one would use as many basis functions as shown in Fig.\ \ref{fig:SpaceCreationRuntimeStatistics} in order to describe arcs and patches of composite curves and surfaces, respectively. Here we tried to push the boundaries of the proposed function library in some cases that may also appear in real-world applications. In lower dimensional EC spaces, the endpoints of run-time related confidence intervals will be significantly smaller.)

\small{
    \begin{table}[htbp]
        \bigskip{}
        \caption{Run-time statistics of normalized B-basis function differentiation up to a given maximal order}
        \label{tbl:differentiation_up_to_a_maximal_order}
        \def\arraystretch{1.0}
        \centering
        
        \begin{tabular}{@{}p{0.51\textwidth}p{0.44\textwidth}@{}}
            \toprule
            Task: multi-threaded differentiation of all functions of a specific normalized B-basis at $100$ uniform subdivision points from order $0$ up to a maximal order $d_{\max}$
            &
            Run-time related confidence intervals\tablefootnote{The endpoints are measured in milliseconds and are calculated by using Student's $T$-distribution with $1000$ independent trials at the fixed significance level $0.01$.}
            \\
            \midrule
            $\overline{\mathcal{P}}_{28}^{0,1}$, $d_{\max}=\left\lfloor \frac{1}{2}\dim \mathbb{P}_{28}^{0,1} \right\rfloor = 14$
            &
            $\left(91.105, 93.887\right)_{\mathrm{CPU}_1}; ~\left(49.676, 51.288\right)_{\mathrm{CPU}_2}$
            \\
            \midrule			
            $\overline{\mathcal{T}}_{36}^{0,\frac{\pi}{2}}$, $d_{\max}=\left\lfloor \frac{1}{2}\dim \mathbb{T}_{36}^{0,\frac{\pi}{2}} \right\rfloor = 18$
            &
            $\left(423.504, 431.568\right)_{\mathrm{CPU}_1};~\left(217.626, 220.150\right)_{\mathrm{CPU}_2}$
            \\
            \midrule
            $\overline{\mathcal{H}}_{26}^{0,\pi}$, $d_{\max}=\left\lfloor \frac{1}{2}\dim \mathbb{H}_{26}^{0,\pi} \right\rfloor=13$
            &
            $\left(86.4903, 88.9317\right)_{\mathrm{CPU}_1};~\left(46.1756, 48.8821\right)_{\mathrm{CPU}_2}$
            \\
            \midrule
            $\overline{\mathcal{AT}}_{24}^{0,2\pi}$, $d_{\max}=\left\lfloor \frac{1}{2}\dim \mathbb{AT}_{24}^{0,2\pi} \right\rfloor=12$
            &
            $\left(75.4589, 77.7331\right)_{\mathrm{CPU}_1};~\left(42.9555, 44.7865\right)_{\mathrm{CPU}_2}$
            \\
            \midrule
            $\overline{\mathcal{AET}}_{27}^{-\frac{3\pi}{4},\frac{3\pi}{4}}$, $d_{\max}=\left\lfloor \frac{1}{2}\dim \mathbb{AET}_{27}^{-\frac{3\pi}{4},\frac{3\pi}{4}} \right\rfloor=14$
            &
            $\left(257.381, 262.379\right)_{\mathrm{CPU}_1};~\left(133.511, 135.691\right)_{\mathrm{CPU}_2}$
            \\
            \bottomrule
        \end{tabular}
        \def\arraystretch{1.0}
    \end{table}
}\normalsize	

Apart from possibly incorrect order elevation and subdivision, too big dimension numbers will generate undesired point-wise perturbations in the shapes of the generated B-basis functions similarly to Fig.\ \ref{fig:numerical_instability_beyond_maximal_dimension}(\textit{b}). Such numerical instabilities will also be inherited by the shapes of B-curves/surfaces in the form of unwanted undulations. 

\begin{figure}[!h]
    \centering
    \includegraphics[]{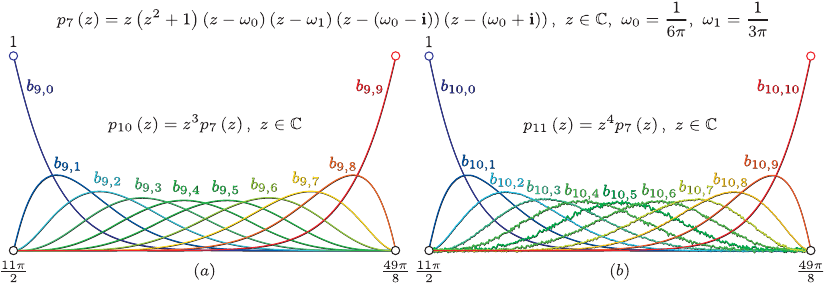}
    \caption{Using the notations of Example \ref{exmp:snail}, in cases (\textit{a}) and (\textit{b}) we have elevated the dimension of $\mathbb{ET}_6^{\alpha,\beta}$, by increasing the multiplicity of the zero $z=0$ of $p_7$ from $1$ to $3$ and $4$, respectively. Observe that in case (\textit{b}) the proposed  B-basis generation method became numerically unstable, i.e., over the given definition domain the dimension of the initial EC space $\mathbb{ET}_6^{\alpha,\beta}$ should be increased from $7$ only up to $10$, by appending its ordinary basis with the monomials $\left\{u,u^2 : u \in \left[\alpha, \beta\right]\right\}$. (If one defines the initial space on a different domain or appends its ordinary basis with other linearly independent functions, the maximum dimension may differ from the previously determined $10$.) 
    }
    \label{fig:numerical_instability_beyond_maximal_dimension}
\end{figure}

If in case of an \CBlue{ECSpace} object one obtains an image similar to Fig.\ \ref{fig:numerical_instability_beyond_maximal_dimension}(\textit{b}) or even a more drastic one from the perspective of the undesired error suggesting point-wise perturbations, then one can conclude that the dimension of the constructed EC space is too big.  In such cases, one may try to translate or decrease the length of the definition domain (since its endpoints have a significant effect on the endpoint derivatives that are required by all algorithms), or if this does not remedy the situation, then one has to either exclude some of the ordinary basis functions that span the underlying space, or at least to try to modify their possible shape parameters in a such a way that potentially increases the numerical stability.

Since there are infinitely many EC spaces with a vast possibility of inner structure, we cannot give a general recipe for the critical maximal dimension number for which the outputs of the proposed algorithms are correct -- \textit{its value should be determined empirically by the user}.


\subsection{Further examples and run-time statistics}

The proposed library is also able to evaluate and render the zeroth and higher order derivatives of B-curves and of isoparametric lines of B-surfaces as it is illustrated in Fig.\ \ref{fig:isoparametric_lines_snail}.

\begin{figure}[!h]
    \centering
    \includegraphics[]{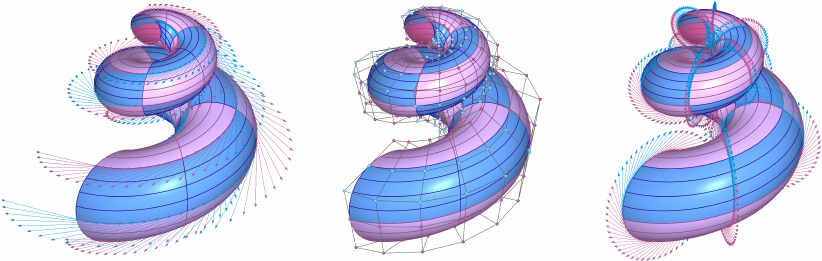}
    \caption{Isoparametric lines and their first order derivatives in case of a possible B-representation of the ordinary integral surface (\ref{eq:snail}). (The illustrated B-surface patches coincide with those in Fig. \ref{fig:exact_description_snail}(\textit{c}). Along all $25$ B-surface patches we have generated $5$ and $3$ isoparametric lines in directions $u_0$ and $u_1$, respectively. The $u_0$- and $u_1$-isoparametric lines consist of $20$ and $13$ subdivision points, respectively. For better visibility, we have rendered the tangent vectors only of some of the isoparametric lines.)}
    \label{fig:isoparametric_lines_snail}
\end{figure}

Moreover, when generating the image of a B-surface, the user is also able to choose different color schemes (as in Fig.\ \ref{fig:color_schemes}) that correspond to the point-wise temperature variations of various energy quantities which, instead of rough approximations, are evaluated exactly based on the coefficients of the first and second fundamental forms of the generated B-surface. (For more details, see also Fig.\ \mref{3.5}{320} of \citep{Roth2018b}.)

\begin{figure}[!h]
    \centering
    \includegraphics[]{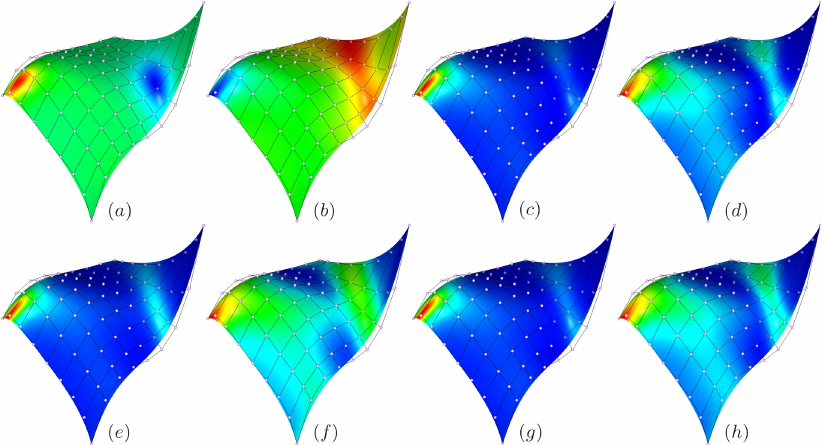}
    \caption{Point-wise temperature variations of the (\textit{a}) Gaussian-curvature,  of the (\textit{b}) mean curvature,  of the (\textit{c})--(\textit{d}) Willmore energy and its translated logarithmic scale, of the (\textit{e})--(\textit{f}) umbilic deviation and its translated logarithmic scale, of the (\textit{g})--(\textit{h}) total curvature and its translated logarithmic scale.}
    \label{fig:color_schemes}
\end{figure}

Using formulas proposed in Theorems \ref{thm:integral_curves}--\ref{thm:integral_surfaces}, one can provide infinitely order of precision concerning the zeroth and higher order (partial) derivatives. Apart from rational polynomial curves/surfaces, this cannot be achieved by the standard rational B\'ezier or NURBS curve/surface modeling tools (e.g.,\ elliptic/hyperbolic arcs of conic sections can be described as rational B\'ezier curves, but these will provide neither natural parametrization nor higher order precision concerning the derivatives). Moreover, these standard rational polynomial models cannot encompass transcendental curves/surfaces and they also rely on  non-negative weight vectors/matrices of rank at least $1$ -- the calculation of which, apart from some simple cases, is cumbersome for the designer.

Table \ref{tbl:examples_of_section_2} provides further run-time related confidence intervals that were obtained in case of some of the remaining examples presented by the manuscript.

\small{
    \begin{table}[!h]
        \vspace{5mm}\caption{Run-time statistics of examples presented in Section \ref{sec:proposed_algorithms}}
        \label{tbl:examples_of_section_2}
        \def\arraystretch{1.0}
        \centering
        
        \begin{tabular}{@{}p{0.05\textwidth}p{0.68\textwidth}p{0.22\textwidth}@{}}
            \toprule
            Figure
            &
            Task
            &
            Run-time related confidence intervals\tablefootnote{The endpoints are measured in milliseconds and are calculated by using Student's $T$-distribution with $5000$ independent trials at the fixed significance level $0.01$.}
            \\
            \midrule
            \ref{fig:order_elevated_snails}(\textit{a}) 
            &
            Creating a B-surface object of minimal order and calculating the transformation matrices for the B-representation of the patch $\left.\mathbf{s}\right|_{\left[\frac{11\pi}{2},\frac{49\pi}{8}\right]\times\left[-\frac{\pi}{3},\frac{\pi}{3}\right]}$ of the ordinary exponential-trigonometric integral surface (\ref{eq:snail}).
            &
            $\left(0.033229, 0.11037\right)_{\mathrm{CPU}_1}$
            $\left(0.083522, 0.11440\right)_{\mathrm{CPU}_2}$
            \\
            \cmidrule{2-3}
            &
            Calculating the control net for the B-representation of the given patch and updating the vertex buffer object of the obtained control net.
            &
            $\left(0.00000, 0.03861\right)_{\mathrm{CPU}_1}$
            $\left(0.01892, 0.03412\right)_{\mathrm{CPU}_2}$
            \\
            \cmidrule{2-3}
            &
            Generating its triangle mesh and updating the vertex buffer objects of the obtained mesh.\tablefootnote{All triangle meshes of the presented examples store $50 \cdot 100 = 5000$ unique vertices (with associated unit normals, colors and texture coordinates) and $2 \cdot \left(50-1\right)\cdot\left(100-1\right)=9702$ triangular faces. Instead of rough approximations, the unit normals are calculated by normalizing the vector products of the calculated first order partial derivatives. Apart from Fig.\ \ref{fig:color_schemes}, in case of each mesh generation, we have used the default color scheme \CBlue{TensorProductSurface3}::\CBlue{DEFAULT\_NULL\_FRAGMENT}. Note that, in practice, usually it would be sufficient to generate triangle meshes with significantly less number of attributes and faces.
            }
            &
            $\left(47.3069, 48.0527\right)_{\mathrm{CPU}_1}$
            $\left(24.6402, 25.5222\right)_{\mathrm{CPU}_2}$
            \\
            \midrule
            \ref{fig:order_elevated_snails}(\textit{b})
            &
            Performing order elevation and updating the vertex buffer object of the order elevated control net. (More details can be found in Example \ref{exmp:snail}.)
            &
            $\left(0.08402, 0.19038\right)_{\mathrm{CPU}_1}$
            $\left(0.16458, 0.20607\right)_{\mathrm{CPU}_2}$
            \\
            \cmidrule{2-3}
            &
            Generating the triangle mesh of the order elevated B-surface and updating the vertex buffer objects of the obtained mesh.
            &
            $\left(70.3128, 71.1488\right)_{\mathrm{CPU}_1}$
            $\left(36.5207, 37.3169\right)_{\mathrm{CPU}_2}$
            \\
            \midrule
            \ref{fig:order_elevated_snails}(\textit{c})
            &
            Performing order elevation and updating the vertex buffer object of the order elevated control net. (More details can be found in Example \ref{exmp:snail}.)
            &
            $\left(0.32332, 0.50668\right)_{\mathrm{CPU}_1}$
            $\left(0.39463, 0.45833\right)_{\mathrm{CPU}_2}$
            \\
            \cmidrule{2-3}
            &
            Generating the triangle mesh of the order elevated B-surface and updating its vertex buffer objects.
            &
            $\left(99.6233, 100.582\right)_{\mathrm{CPU}_1}$
            $\left(52.1324, 53.1988\right)_{\mathrm{CPU}_2}$
            \\
            \midrule
            \ref{fig:subdivided_snail}
            &
            Creating four B-surfaces and updating the vertex buffer objects of their control nets, by subdividing the initial surface patch $\left.\mathbf{s}\right|_{\left[\frac{11\pi}{2},\frac{49\pi}{8}\right]\times\left[-\frac{\pi}{3},\frac{\pi}{3}\right]}$ at the parameter value $u_1=0$, then by subdividing one of the obtained B-surface elements at the parameter value $u_0=\frac{93\pi}{16}$.
            &
            $\left(0.46939, 0.69221\right)_{\mathrm{CPU}_1}$
            $\left(0.72089, 0.80915\right)_{\mathrm{CPU}_2}$
            \\
            \cmidrule{2-3}
            &
            Generating the triangle meshes of the obtained B-surfaces and updating their vertex buffer objects.
            &
            $\left(194.946, 196.301\right)_{\mathrm{CPU}_1}$
            $\left(101.882, 103.559\right)_{\mathrm{CPU}_2}$
            \\
            \midrule
            \ref{fig:exact_description_snail}(\textit{c})
            &
            Multi-threaded calculation of control nets and the sequential update of their vertex buffer objects for the B-representation of the ordinary integral surface (\ref{eq:snail}) by generating a $5\times 5$ matrix of exponential-trigonometric B-surface patches of order $(n_0=6,n_1=2)$ and common shape parameters $(\beta_0 - \alpha_0 = \frac{29\pi}{40},\beta_1-\alpha_1=\frac{2\pi}{5})$.
            &
            $\left(0.817715, 1.09069\right)_{\mathrm{CPU}_1}$
            $\left(0.729526, 1.03447\right)_{\mathrm{CPU}_2}$
            \\
            \cmidrule{2-3}
            &
            Multi-threaded generation of each of the $25$ triangle meshes of the obtained B-surface patches, then the sequential update of their vertex buffer objects.
            &
            $\left(1242.90, 1250.62\right)_{\mathrm{CPU}_1}$ $\left(669.218, 678.874\right)_{\mathrm{CPU}_2}$
            \\
            \midrule
            \ref{fig:isoparametric_lines_snail}
            &
            Multi-threaded generation and first order differentiation of $5\cdot5\cdot5\cdot3=375$ isoparametric lines of Fig.\ \ref{fig:isoparametric_lines_snail}, then updating their vertex buffer objects.
            &
            $\left(30.3413, 30.8631\right)_{\mathrm{CPU}_1}$ $\left(19.6889, 20.5075\right)_{\mathrm{CPU}_2}$
            \\
            \bottomrule
        \end{tabular}
        \def\arraystretch{1.0}
    \end{table}
}\normalsize

\section{Closure}{\label{sec:closure}}

Using the unique normalized B-basis of an EC space that includes the constants and can be identified with the translation invariant solution space of a constant-coefficient homogeneous linear differential equation defined on sufficiently small interval (i.e., on which the space spanned by the first order derivatives is also EC)
, we have proposed a platform-independent OpenGL and C++ based multi-threaded robust and flexible function library for free-form curve and surface modeling. The proposed data structures are able to generate, differentiate and render both the ordinary basis and the normalized B-basis of the underlying EC spaces. Our library can also create, (partially) differentiate, modify and render a large family of B-curves and tensor product B-surfaces, and is also able to perform operations (like order elevation and subdivision) on them (at least up to a reasonable number of dimension and with acceptable numerical precision). The user also has the possibility to solve interpolation problems and to describe exactly arcs/patches of arbitrary ordinary integral curves/surfaces by means of B-curves/surfaces. 

We think, Subsection \ref{subsec:motivations} provided sufficient motivations both for the usage of not necessarily polynomial normalized B-basis functions and for the application of B-curves/surfaces implied by them. Moreover, as it is suggested by the included run-time statistics, the proposed function library is quite responsive and reliable even on an affordable laptop that has a multi-core CPU and a dedicated GPU that is compatible at least with the desktop variant of OpenGL 3.0. 
We deliberately chose the transcendental surface (\ref{eq:snail}) in case of Figs. \ref{fig:order_elevated_snails}, \ref{fig:subdivided_snail}, \ref{fig:exact_description_snail} and \ref{fig:isoparametric_lines_snail}, since it cannot be described exactly by means of the standard (rational) B\'ezier and (non-uniform) B-spline modeling tools.  Apart from some image post-processing (like adding \LaTeX{}-like  descriptive elements), all curve or surface illustrating figures were generated by means of the proposed function library which is not merely the collage implementation of already existing theoretical or numerical methods. Although some parts of our implementation rely on \citep{MazureLaurent1998,CarnicerMainarPena2004,Roth2015b}, in Section \ref{sec:proposed_algorithms} we have also formulated new original results in Corollary \ref{cor:B_basis_derivatives} and in Theorems \ref{thm:general_subdivision}, \ref{thm:efficient_basis_transformation} and \ref{thm:integral_surfaces} which describe constructive formulas that provide implementation advantages. To the best of our knowledge, such a unifying general \textit{programming framework} for curve and surface modeling was not presented in the literature so far. Even special cases of normalized B-bases (like the well-known Bernstein polynomials and their application possibilities) are considered to be important \citep{TsaiFarouki2001}.

We have also included our detailed user manual \citep{Roth2018b} in the supplementary material of the manuscript that covers full implementation details and source code listings, by using which one is able both to reproduce all presented examples and to create new types of EC spaces, B-curves and B-surfaces in just a few lines of code. The user manual consists of three chapters. The first one provides a theoretical introduction without the proofs of the current Section \ref{sec:proposed_algorithms}. Labels and formulas of this chapter are used to clarify the full implementation details included in the second chapter of the user manual. Assuming that users provide an OpenGL based class that is able to create a rendering context and to handle possible events, the third chapter of the user manual concentrates on the constructor and rendering methods of the aforementioned class in order to interactively manipulate different types of B-curves/surfaces, to perform operations (like order elevation and subdivision) on them and to provide B-representations for ordinary integral curves/surfaces.

We believe, the proposed library will help the development of other software packages for a large variety of real-world applications that arise from Approximation Theory, Computer Aided Geometric Design and Manufacturing, Computer Graphics, Isogeometric and Numerical Analysis.

\bibliographystyle{ACM-Reference-Format}
\bibliography{Roth__Bibliography__Curve_and_surface_modeling_in_EC_spaces.R1}


\begin{thebibliography}{00}


\ifx \showCODEN    \undefined \def \showCODEN     #1{\unskip}     \fi
\ifx \showDOI      \undefined \def \showDOI       #1{{\tt DOI:}\penalty0{#1}\ }
  \fi
\ifx \showISBNx    \undefined \def \showISBNx     #1{\unskip}     \fi
\ifx \showISBNxiii \undefined \def \showISBNxiii  #1{\unskip}     \fi
\ifx \showISSN     \undefined \def \showISSN      #1{\unskip}     \fi
\ifx \showLCCN     \undefined \def \showLCCN      #1{\unskip}     \fi
\ifx \shownote     \undefined \def \shownote      #1{#1}          \fi
\ifx \showarticletitle \undefined \def \showarticletitle #1{#1}   \fi
\ifx \showURL      \undefined \def \showURL       #1{#1}          \fi
\providecommand\bibfield[2]{#2}
\providecommand\bibinfo[2]{#2}
\providecommand\natexlab[1]{#1}
\providecommand\showeprint[2][]{arXiv:#2}

\bibitem[\protect\citeauthoryear{Brilleaud and Mazure}{Brilleaud and
  Mazure}{2012}]%
        {BrilleaudMazure2012}
\bibfield{author}{\bibinfo{person}{M. Brilleaud} {and} \bibinfo{person}{M.-L.
  Mazure}.} \bibinfo{year}{2012}\natexlab{}.
\newblock \showarticletitle{Mixed hyperbolic/trigonometric spaces for design}.
\newblock \bibinfo{journal}{{\em Computers and Mathematics with
  Applications\/}} \bibinfo{volume}{64}, \bibinfo{number}{8}
  (\bibinfo{year}{2012}), \bibinfo{pages}{2459--2477}.
\newblock
\showDOI{%
\url{http://dx.doi.org/10.1016/j.camwa.2012.05.019}}


\bibitem[\protect\citeauthoryear{Carnicer, Mainar, and Pe{\~{n}}a}{Carnicer
  et~al\mbox{.}}{2004}]%
        {CarnicerMainarPena2004}
\bibfield{author}{\bibinfo{person}{J.-M. Carnicer}, \bibinfo{person}{E.
  Mainar}, {and} \bibinfo{person}{J.~M. Pe{\~{n}}a}.}
  \bibinfo{year}{2004}\natexlab{}.
\newblock \showarticletitle{Critical length for design purposes and extended
  Chebyshev spaces}.
\newblock \bibinfo{journal}{{\em Constructive Approximation\/}}
  \bibinfo{volume}{20}, \bibinfo{number}{1} (\bibinfo{year}{2004}),
  \bibinfo{pages}{55--71}.
\newblock
\showDOI{%
\url{http://dx.doi.org/10.1007/s00365-002-0530-1}}


\bibitem[\protect\citeauthoryear{Carnicer, Mainar, and Pe{\~{n}}a}{Carnicer
  et~al\mbox{.}}{2007}]%
        {CarnicerMainarPena2007}
\bibfield{author}{\bibinfo{person}{J.-M. Carnicer}, \bibinfo{person}{E.
  Mainar}, {and} \bibinfo{person}{J.~M. Pe{\~{n}}a}.}
  \bibinfo{year}{2007}\natexlab{}.
\newblock \showarticletitle{Shape preservation regions for six-dimensional
  spaces}.
\newblock \bibinfo{journal}{{\em Advances in Computational Mathematics\/}}
  \bibinfo{volume}{26}, \bibinfo{number}{1--3} (\bibinfo{year}{2007}),
  \bibinfo{pages}{121--136}.
\newblock
\showDOI{%
\url{http://dx.doi.org/10.1007/s10444-005-7505-2}}


\bibitem[\protect\citeauthoryear{Carnicer, Mainar, and Pe{\~{n}}a}{Carnicer
  et~al\mbox{.}}{2014}]%
        {CarnicerMainarPena2014}
\bibfield{author}{\bibinfo{person}{J.-M. Carnicer}, \bibinfo{person}{E.
  Mainar}, {and} \bibinfo{person}{J.~M. Pe{\~{n}}a}.}
  \bibinfo{year}{2014}\natexlab{}.
\newblock \showarticletitle{On the critical length of cycloidal spaces}.
\newblock \bibinfo{journal}{{\em Constructive Approximation\/}}
  \bibinfo{volume}{39}, \bibinfo{number}{3} (\bibinfo{year}{2014}),
  \bibinfo{pages}{573--583}.
\newblock
\showDOI{%
\url{http://dx.doi.org/10.1007/s00365-013-9223-1}}


\bibitem[\protect\citeauthoryear{Carnicer and Pe{\~{n}}a}{Carnicer and
  Pe{\~{n}}a}{1993}]%
        {Carnicer1993}
\bibfield{author}{\bibinfo{person}{J.-M. Carnicer} {and} \bibinfo{person}{J.~M.
  Pe{\~{n}}a}.} \bibinfo{year}{1993}\natexlab{}.
\newblock \showarticletitle{Shape preserving representations and optimality of
  the Bernstein basis}.
\newblock \bibinfo{journal}{{\em Advances in Computational Mathematics\/}}
  \bibinfo{volume}{1}, \bibinfo{number}{2} (\bibinfo{year}{1993}),
  \bibinfo{pages}{173--196}.
\newblock
\showDOI{%
\url{http://dx.doi.org/10.1007/BF02071384}}


\bibitem[\protect\citeauthoryear{Carnicer and Pe{\~{n}}a}{Carnicer and
  Pe{\~{n}}a}{1995}]%
        {CarnicerPena1995}
\bibfield{author}{\bibinfo{person}{J.-M. Carnicer} {and} \bibinfo{person}{J.~M.
  Pe{\~{n}}a}.} \bibinfo{year}{1995}\natexlab{}.
\newblock \showarticletitle{On transforming a Tchebycheff system into a
  strictly totally positive system}.
\newblock \bibinfo{journal}{{\em Journal of Approximation Theory\/}}
  \bibinfo{volume}{81}, \bibinfo{number}{2} (\bibinfo{year}{1995}),
  \bibinfo{pages}{274--295}.
\newblock
\showDOI{%
\url{http://dx.doi.org/10.1006/jath.1995.1050}}


\bibitem[\protect\citeauthoryear{Costantini, Lyche, and Manni}{Costantini
  et~al\mbox{.}}{2005}]%
        {CostantiniLycheManni2005}
\bibfield{author}{\bibinfo{person}{P. Costantini}, \bibinfo{person}{T. Lyche},
  {and} \bibinfo{person}{C. Manni}.} \bibinfo{year}{2005}\natexlab{}.
\newblock \showarticletitle{On a class of weak Tchebycheff systems}.
\newblock \bibinfo{journal}{{\it Numer. Math.}} \bibinfo{volume}{101},
  \bibinfo{number}{2} (\bibinfo{year}{2005}), \bibinfo{pages}{333--354}.
\newblock
\showDOI{%
\url{http://dx.doi.org/10.1007/s00211-005-0613-6}}


\bibitem[\protect\citeauthoryear{Gasca and Pe{\~{n}}a}{Gasca and
  Pe{\~{n}}a}{1996}]%
        {GascaPena1996}
\bibfield{author}{\bibinfo{person}{M. Gasca} {and} \bibinfo{person}{J.~M.
  Pe{\~{n}}a}.} \bibinfo{year}{1996}\natexlab{}.
\newblock \showarticletitle{On factorizations of totally positive matrices}.
\newblock In \bibinfo{booktitle}{{\em Total Positivity and its Applications}},
  \bibfield{editor}{\bibinfo{person}{M.~Gasca} {and} \bibinfo{person}{C.~A.
  Micchelli}} (Eds.). \bibinfo{publisher}{Kluwer Academic},
  \bibinfo{address}{Dordrecht}, \bibinfo{pages}{109--130}.
\newblock
\showISBNx{978-90-481-4667-3}
\showDOI{%
\url{http://dx.doi.org/10.1007/978-94-015-8674-0_7}}


\bibitem[\protect\citeauthoryear{Karlin and Studden}{Karlin and
  Studden}{1966}]%
        {KarlinStudden1966}
\bibfield{author}{\bibinfo{person}{S. Karlin} {and} \bibinfo{person}{W.
  Studden}.} \bibinfo{year}{1966}\natexlab{}.
\newblock \bibinfo{booktitle}{{\em Tchebycheff Systems: with Applications in
  Analysis and Statistics}}.
\newblock \bibinfo{publisher}{Wiley}, \bibinfo{address}{New York, NY}.
\newblock


\bibitem[\protect\citeauthoryear{L{\"u}, Wang, and Yang}{L{\"u}
  et~al\mbox{.}}{2002}]%
        {LuWangYang2002}
\bibfield{author}{\bibinfo{person}{Y. L{\"u}}, \bibinfo{person}{G. Wang}, {and}
  \bibinfo{person}{X. Yang}.} \bibinfo{year}{2002}\natexlab{}.
\newblock \showarticletitle{Uniform hyperbolic polynomial B-spline curves}.
\newblock \bibinfo{journal}{{\em Computer Aided Geometric Design\/}}
  \bibinfo{volume}{19}, \bibinfo{number}{6} (\bibinfo{year}{2002}),
  \bibinfo{pages}{379--393}.
\newblock
\showDOI{%
\url{http://dx.doi.org/10.1016/S0167-8396(02)00092-4}}


\bibitem[\protect\citeauthoryear{Lyche}{Lyche}{1985}]%
        {Lyche1985}
\bibfield{author}{\bibinfo{person}{T. Lyche}.} \bibinfo{year}{1985}\natexlab{}.
\newblock \showarticletitle{A recurrence relation for Chebyshevian B-splines}.
\newblock \bibinfo{journal}{{\em Constructive Approximation\/}}
  \bibinfo{volume}{1}, \bibinfo{number}{1} (\bibinfo{year}{1985}),
  \bibinfo{pages}{155--173}.
\newblock
\showDOI{%
\url{http://dx.doi.org/10.1007/BF01890028}}


\bibitem[\protect\citeauthoryear{Mainar and Pe{\~{n}}a}{Mainar and
  Pe{\~{n}}a}{1999}]%
        {MainarPena1999}
\bibfield{author}{\bibinfo{person}{E. Mainar} {and} \bibinfo{person}{J.~M.
  Pe{\~{n}}a}.} \bibinfo{year}{1999}\natexlab{}.
\newblock \showarticletitle{Corner cutting algorithms associated with optimal
  shape preserving representations}.
\newblock \bibinfo{journal}{{\em Computer Aided Geometric Design\/}}
  \bibinfo{volume}{16}, \bibinfo{number}{9} (\bibinfo{year}{1999}),
  \bibinfo{pages}{883--906}.
\newblock
\showDOI{%
\url{http://dx.doi.org/10.1016/S0167-8396(99)00035-7}}


\bibitem[\protect\citeauthoryear{Mainar and Pe{\~{n}}a}{Mainar and
  Pe{\~{n}}a}{2004}]%
        {MainarPena2004}
\bibfield{author}{\bibinfo{person}{E. Mainar} {and} \bibinfo{person}{J.~M.
  Pe{\~{n}}a}.} \bibinfo{year}{2004}\natexlab{}.
\newblock \showarticletitle{Quadratic-cycloidal curves}.
\newblock \bibinfo{journal}{{\em Advances in Computational Mathematics\/}}
  \bibinfo{volume}{20}, \bibinfo{number}{1--3} (\bibinfo{year}{2004}),
  \bibinfo{pages}{161--175}.
\newblock
\showDOI{%
\url{http://dx.doi.org/10.1023/A:1025813919473}}


\bibitem[\protect\citeauthoryear{Mainar and Pe{\~{n}}a}{Mainar and
  Pe{\~{n}}a}{2010}]%
        {MainarPena2010}
\bibfield{author}{\bibinfo{person}{E. Mainar} {and} \bibinfo{person}{J.~M.
  Pe{\~{n}}a}.} \bibinfo{year}{2010}\natexlab{}.
\newblock \showarticletitle{Optimal bases for a class of mixed spaces and their
  associated spline spaces}.
\newblock \bibinfo{journal}{{\em Computers and Mathematics with
  Applications\/}} \bibinfo{volume}{59}, \bibinfo{number}{4}
  (\bibinfo{year}{2010}), \bibinfo{pages}{1509--1523}.
\newblock
\showDOI{%
\url{http://dx.doi.org/10.1016/j.camwa.2009.11.009}}


\bibitem[\protect\citeauthoryear{Mainar, Pe{\~{n}}a, and
  S{\'{a}}nchez-Reyes}{Mainar et~al\mbox{.}}{2001}]%
        {MainarPenaSanchez2001}
\bibfield{author}{\bibinfo{person}{E. Mainar}, \bibinfo{person}{J.~M.
  Pe{\~{n}}a}, {and} \bibinfo{person}{J.~M. S{\'{a}}nchez-Reyes}.}
  \bibinfo{year}{2001}\natexlab{}.
\newblock \showarticletitle{Shape preserving alternatives to the rational
  B{\'e}zier model}.
\newblock \bibinfo{journal}{{\em Computer Aided Geometric Design\/}}
  \bibinfo{volume}{18}, \bibinfo{number}{1} (\bibinfo{year}{2001}),
  \bibinfo{pages}{37--60}.
\newblock
\showDOI{%
\url{http://dx.doi.org/10.1016/S0167-8396(01)00011-5}}


\bibitem[\protect\citeauthoryear{Manni, Pelosi, and Sampoli}{Manni
  et~al\mbox{.}}{2011}]%
        {ManniPelosiSampoli2011}
\bibfield{author}{\bibinfo{person}{C. Manni}, \bibinfo{person}{F. Pelosi},
  {and} \bibinfo{person}{M.~L. Sampoli}.} \bibinfo{year}{2011}\natexlab{}.
\newblock \showarticletitle{Generalized B-splines as a tool in isogeometric
  analysis}.
\newblock \bibinfo{journal}{{\em Computer Methods in Applied Mechanics and
  Engineering\/}} \bibinfo{volume}{200}, \bibinfo{number}{5--8}
  (\bibinfo{year}{2011}), \bibinfo{pages}{867--881}.
\newblock
\showDOI{%
\url{http://dx.doi.org/10.1016/j.cma.2010.10.010}}


\bibitem[\protect\citeauthoryear{Mazure}{Mazure}{1999}]%
        {Mazure1999}
\bibfield{author}{\bibinfo{person}{M.-L. Mazure}.}
  \bibinfo{year}{1999}\natexlab{}.
\newblock \showarticletitle{Chebyshev--Bernstein bases}.
\newblock \bibinfo{journal}{{\em Computer Aided Geometric Design\/}}
  \bibinfo{volume}{16}, \bibinfo{number}{7} (\bibinfo{year}{1999}),
  \bibinfo{pages}{649--669}.
\newblock
\showDOI{%
\url{http://dx.doi.org/10.1016/S0167-8396(99)00029-1}}


\bibitem[\protect\citeauthoryear{Mazure and Laurent}{Mazure and
  Laurent}{1998}]%
        {MazureLaurent1998}
\bibfield{author}{\bibinfo{person}{M.-L. Mazure} {and} \bibinfo{person}{P.-J.
  Laurent}.} \bibinfo{year}{1998}\natexlab{}.
\newblock \showarticletitle{Nested sequences of Chebyshev spaces and shape
  parameters}.
\newblock \bibinfo{journal}{{\em RAIRO--Mod{\'e}lisation math{\'e}matique et
  analyse num{\'e}rique\/}} \bibinfo{volume}{32}, \bibinfo{number}{6}
  (\bibinfo{year}{1998}), \bibinfo{pages}{773--788}.
\newblock
\showURL{%
\url{http://www.numdam.org/article/M2AN_1998__32_6_773_0.pdf}}


\bibitem[\protect\citeauthoryear{Pe{\~{n}}a}{Pe{\~{n}}a}{1999}]%
        {Pena1999}
\bibfield{author}{\bibinfo{person}{J.~M. Pe{\~{n}}a}.}
  \bibinfo{year}{1999}\natexlab{}.
\newblock \bibinfo{booktitle}{{\em Shape Preserving Representations in
  Computer-Aided Geometric Design}}.
\newblock \bibinfo{publisher}{Nova Science Publishers},
  \bibinfo{address}{Commack, NY}.
\newblock


\bibitem[\protect\citeauthoryear{Pottmann}{Pottmann}{1993}]%
        {Pottmann1993}
\bibfield{author}{\bibinfo{person}{H. Pottmann}.}
  \bibinfo{year}{1993}\natexlab{}.
\newblock \showarticletitle{The geometry of Tchebycheffian splines}.
\newblock \bibinfo{journal}{{\em Computer Aided Geometric Design\/}}
  \bibinfo{volume}{10}, \bibinfo{number}{3--4} (\bibinfo{year}{1993}),
  \bibinfo{pages}{181--210}.
\newblock
\showDOI{%
\url{http://dx.doi.org/10.1016/0167-8396(93)90036-3}}


\bibitem[\protect\citeauthoryear{Press, Teukolsky, Vetterling, and
  Flannery}{Press et~al\mbox{.}}{2007}]%
        {PressTeukolskyVetterlingFlannery2007}
\bibfield{author}{\bibinfo{person}{W.~H. Press}, \bibinfo{person}{S.~A.
  Teukolsky}, \bibinfo{person}{W.~T. Vetterling}, {and} \bibinfo{person}{B.~P.
  Flannery}.} \bibinfo{year}{2007}\natexlab{}.
\newblock \bibinfo{booktitle}{{\em Numerical Recipes 3rd Edition: The art of
  Scientific Computing}}.
\newblock \bibinfo{publisher}{Cambridge University Press},
  \bibinfo{address}{New York}.
\newblock


\bibitem[\protect\citeauthoryear{R{\'o}th}{R{\'o}th}{2015a}]%
        {Roth2015b}
\bibfield{author}{\bibinfo{person}{{\'A}. R{\'o}th}.}
  \bibinfo{year}{2015}\natexlab{a}.
\newblock \showarticletitle{Control point based exact description of curves and
  surfaces in extended Chebyshev spaces}.
\newblock \bibinfo{journal}{{\em Computer Aided Geometric Design\/}}
  \bibinfo{volume}{40} (\bibinfo{date}{December 14} \bibinfo{year}{2015}),
  \bibinfo{pages}{40--58}.
\newblock
\showDOI{%
\url{http://dx.doi.org/10.1016/j.cagd.2015.09.005}}


\bibitem[\protect\citeauthoryear{R{\'o}th}{R{\'o}th}{2015b}]%
        {Roth2015a}
\bibfield{author}{\bibinfo{person}{{\'A}. R{\'o}th}.}
  \bibinfo{year}{2015}\natexlab{b}.
\newblock \showarticletitle{Control point based exact description of
  trigonometric/hyperbolic curves, surfaces and volumes}.
\newblock \bibinfo{journal}{{\it J. Comput. Appl. Math.}}
  \bibinfo{volume}{290}, \bibinfo{number}{C} (\bibinfo{date}{December 1}
  \bibinfo{year}{2015}), \bibinfo{pages}{74--91}.
\newblock
\showDOI{%
\url{http://dx.doi.org/10.1016/j.cam.2015.05.003}}


\bibitem[\protect\citeauthoryear{R{\'o}th}{R{\'o}th}{2018}]%
        {Roth2018b}
\bibfield{author}{\bibinfo{person}{{\'A}. R{\'o}th}.}
  \bibinfo{year}{2018}\natexlab{}.
\newblock \bibinfo{booktitle}{{\em User manual. An OpenGL and C++ based
  function library for curve and surface modeling in a large class of extended
  Chebyshev spaces\/} (\bibinfo{edition}{1st} ed.)}.
\newblock \bibinfo{address}{Babe{\c{s}}--Bolyai University, Department of
  Mathematics and Computer Science of the Hungarian Line, RO--400084
  Cluj-Napoca, Romania}.
\newblock
\newblock
\shownote{Supplementary material.}


\bibitem[\protect\citeauthoryear{S{\'a}nchez-Reyes}{S{\'a}nchez-Reyes}{1998}]%
        {Sanchez1998}
\bibfield{author}{\bibinfo{person}{J. S{\'a}nchez-Reyes}.}
  \bibinfo{year}{1998}\natexlab{}.
\newblock \showarticletitle{Harmonic rational B{\'e}zier curves, p-B{\'e}zier
  curves and trigonometric polynomials}.
\newblock \bibinfo{journal}{{\em Computer Aided Geometric Design\/}}
  \bibinfo{volume}{15}, \bibinfo{number}{9} (\bibinfo{year}{1998}),
  \bibinfo{pages}{909--923}.
\newblock
\showDOI{%
\url{http://dx.doi.org/10.1016/S0167-8396(98)00031-4}}


\bibitem[\protect\citeauthoryear{Schumaker}{Schumaker}{2007}]%
        {Schumaker2007}
\bibfield{author}{\bibinfo{person}{L.~L. Schumaker}.}
  \bibinfo{year}{2007}\natexlab{}.
\newblock \bibinfo{booktitle}{{\em Spline Functions: Basic Theory, 3rd
  edition}}.
\newblock \bibinfo{publisher}{Cambridge University Press},
  \bibinfo{address}{UK}.
\newblock
\showDOI{%
\url{http://dx.doi.org/10.1017/CBO9780511618994}}


\bibitem[\protect\citeauthoryear{Shen and Wang}{Shen and Wang}{2005}]%
        {ShenWang2005}
\bibfield{author}{\bibinfo{person}{W.-Q. Shen} {and} \bibinfo{person}{G.-Z.
  Wang}.} \bibinfo{year}{2005}\natexlab{}.
\newblock \showarticletitle{A class of quasi B{\'e}zier curves based on
  hyperbolic polynomials}.
\newblock \bibinfo{journal}{{\em Journal of Zhejiang University SCIENCE\/}}
  \bibinfo{volume}{6A (Suppl. I)}, \bibinfo{number}{9} (\bibinfo{year}{2005}),
  \bibinfo{pages}{116--123}.
\newblock
\showDOI{%
\url{http://dx.doi.org/10.1007/BF02887226}}


\bibitem[\protect\citeauthoryear{Tsai and Farouki}{Tsai and Farouki}{2001}]%
        {TsaiFarouki2001}
\bibfield{author}{\bibinfo{person}{Y.-F. Tsai} {and} \bibinfo{person}{R.~T.
  Farouki}.} \bibinfo{year}{2001}\natexlab{}.
\newblock \showarticletitle{Algorithm 812: BPOLY: An object-oriented library of
  numerical algorithms for polynomials in Bernstein form}.
\newblock \bibinfo{journal}{{\it ACM Trans. Math. Software}}
  \bibinfo{volume}{27}, \bibinfo{number}{2} (\bibinfo{year}{2001}),
  \bibinfo{pages}{267--296}.
\newblock
\showDOI{%
\url{http://dx.doi.org/10.1145/383738.383743}}


\end{thebibliography}

\end{document}